\documentclass{article}
\usepackage{latexsym}
\usepackage{amssymb}
\usepackage{amsmath,amsthm}
\usepackage[numbers,sort&compress]{natbib}
\usepackage{color}
\usepackage{amsfonts,amssymb}
\usepackage{mathrsfs}
\usepackage{dsfont}
\usepackage{graphicx}
\usepackage{subfigure}
\usepackage{float}
\usepackage{ulem}
\usepackage{color}
\usepackage{bm}
\usepackage{ytableau}
\usepackage{microtype}

\newcommand {\dif}{\mathrm{d}}

\newcommand {\I}{\mathrm{i}}

\newtheorem{prop}{Proposition}
\newtheorem{rmk}{Remark}
\theoremstyle{definition}
\newtheorem{theo}{Theorem}
\begin{document}
	\title{\bf Integrable variant Blaszak-Szum lattice equation}
	\author{Wei-Kang Xie$^{1}$ and Guo-Fu Yu$^{1}\footnote{Corresponding author. Email address: gfyu@sjtu.edu.cn}$ \\	 
	 $^{1}$	 School of Mathematical Sciences, CMA-Shanghai,\\ Shanghai Jiao Tong University,Shanghai 200240, PR China\\
}
	\date{}
	\maketitle
	
\begin{abstract}
We derive a novel variant of the Blaszak-Szum lattice equation by introducing a new class of trigonometric-type bilinear operators. By employing Hirota's bilinear method, we obtain the Gram-type determinant solution of the variant Blaszak-Szum lattice equation. One-soliton and two-soliton solutions are constructed, with a detailed analysis of the asymptotic behaviors of the two-soliton solution. A B\"acklund transformation for the variant Blaszak-Szum lattice equation is established. By virtue of this B\"acklund transformation, multi-lump solutions of the equation are further constructed.  Rational solutions are derived by introducing two differential operators applied to the determinant elements; in particular, lump solutions derived via these differential operators can be formulated in terms of Schur polynomials. Through parameter variation, three types of breather solutions are obtained, including the Akhmediev breather, Kuznetsov-Ma breather, and a general breather that propagates along arbitrary oblique trajectories. Finally, numerical three-periodic wave solutions to the variant Blaszak-Szum lattice equation are computed using the Gauss-Newton method.
\end{abstract}
	
	
\section{Introduction}
A rogue wave refers to a type of abnormally large surface wave that occurs in the ocean, characterized by its sudden appearance and disappearance. The emergence of rogue waves causes significant threats to maritime vessels. As a result, the investigation of rational solutions in nonlinear systems, such as rogue waves and lump solutions, has become a major focus of current research. The study of rational solutions has extended from oceanography into other branches of physics, encompassing fields such as nonlinear optics, plasma physics, and Bose-Einstein condensates. There are many mathematical approaches to study rational solutions of integrable systems, such as Darboux Transformation \cite{Gu}, Hirota's bilinear method \cite{Hirota}, inverse scattering method \cite{Ablowitz1991} and so on. In recent years, Yang and Yang studied the rogue wave patterns in the nonlinear Schr\"odinger equation and presented that the patterns are analytically determined by the root structures of the Yablonskii-Vorob’ev polynomial \cite{Yang,Yang1}. Later, they disclosed the relationship between pattern of Kadomtsev-Petviashvili (KP) I equation and Yablonskii-Vorob’ev polynomial \cite{Yang2}. Based on their work, Chakravarty and Zowada investigated the connection between the theory of integer partitions and high-order lump solutions of the KP I equation \cite{Chakravarty,Chakravarty2}.

Recently, Liu et al. proposed an effective method to generate novel discrete integrable systems \cite{Liu1,Liu2,Liu3} by introducing a new class of  trigonometric-type bilinear operators $\sin(\delta D_z)$ and $\cos(\delta D_z)$ as
\begin{align}\label{tri}
    \sin(\delta D_z)=\sum_{j=0}^{\infty}\frac{(-1)^j(\delta D_z)^{2j+1}}{(2j+1)!},\quad \cos(\delta D_z)=\sum_{j=0}^{\infty}\frac{(-1)^j(\delta D_z)^{2j}}{(2j)!}.
\end{align}
With the help of new bilinear operators, Liu et al. derived novel variant discrete equations, such as the Toda lattice equation, the Leznov lattice equation and the differential-difference KP equation. They also investigated the dynamics of the variant discrete equations, including soliton, breather, lump and numerical three-periodic wave solutions.  It is well-documented that both the Leznov lattice equation and the Blaszak-Szum (BS) lattice equation originate from the two-dimensional Toda lattice. From this perspective, they share certain analogous algebraic structures of solutions and relevant solution properties. Inspired by these pioneering studies, the present work is devoted to investigating a novel variant of the BS lattice equation.

In 2001, Blaszak and Szum presented an integrable special (2+1)-dimensional lattice equation \cite{BS-1}
\begin{align}
    \frac{\partial u(n)}{\partial t} &= u(n)\mathcal{H}^{-1}p(n-1),\label{BS1} \\
    \frac{\partial v(n)}{\partial t} &= u(n+1) - u(n) + (E+1)^{-1}\frac{\partial p(n)}{\partial y}, \\
    \frac{\partial p(n)}{\partial t} &= v(n+1) - v(n) - p(n)\mathcal{H}^{-1}p(n),\label{BS3}
\end{align}
where $Eu(n)=u(n+1)$ and $\mathcal{H}=(E+1)/(E-1)$. Putting $w(n)=(E-1)^{-1}p(n)$, the BS lattice (\ref{BS1})-(\ref{BS3}) becomes \cite{TamHuQian}
\begin{align}
    &\frac{\partial}{\partial t} u(n)=u(n)(w(n)-w(n-1)),\label{gBS-1}\\
    &\frac{\partial}{\partial t} v(n)=u(n+1)-u(n)+\frac{\partial}{\partial y}w(n),\label{gBS-2}\\
    &\frac{\partial}{\partial t}w(n+1)+\frac{\partial}{\partial t}w(n)=v(n+1)-v(n)-w(n+1)^2+w(n)^2.\label{gBS-3}
\end{align}
Through the dependent variable transformation
\begin{align}
    u(n)=\frac{\tau(n+1)\tau(n-1)}{\tau(n)^2},\quad v(n)=\frac{D_t^2\tau(n)\cdot\tau(n+1)}{\tau(n)\tau(n+1)},\quad w(n)=\left(\ln\frac{\tau(n+1)}{\tau(n)}\right)_t,
\end{align}
and introducing an auxiliary variable $z$, we have bilinear BS lattice equation
\begin{align}
    &\left(D_ze^{(1/2)D_n}-D_t^2e^{(1/2)D_n}\right)\tau(n)\cdot\tau(n)=0,\label{BSb-1}\\
    &(D_tD_z-D_tD_y)\tau(n)\cdot\tau(n)=4\sinh^2\left(\frac{1}{2}D_n\right)\tau(n)\cdot\tau(n).\label{BSb-2}
\end{align}
The Hirota's bilinear differential operators are defined by
\begin{align}
    &D_{x}^{k}D_{y}^{m}f \cdot g \equiv \left( \frac{\partial}{\partial x} - \frac{\partial}{\partial x^{\prime}} \right)^{k}\left( \frac{\partial}{\partial y} - \frac{\partial}{\partial y^{\prime}} \right)^{m} f(x, y) g(x^{\prime}, y^{\prime}) \bigg|_{ x^{\prime}=x, y^{\prime}=y},\\
    &e^{\delta D_n}a(n)\cdot b(n)=a(n+\delta)b(n-\delta),\\
    &\sinh(\delta D_n)=\frac{1}{2}\left(e^{\delta D_n}-e^{-\delta D_n}\right)=\sum_{j=0}^{\infty}\frac{(\delta D_n)^{2j+1}}{(2j+1)!}.
\end{align}
In 2005, Yu et al. first presented the Casorati and Grammian determinant solutions to the bilinear equations (\ref{BSb-1}) and (\ref{BSb-2}) \cite{Yu}. Li et al. gave the Wronskian solutions to the B\"acklund transformation of (\ref{BSb-1}) and (\ref{BSb-2}) \cite{BT}. Sheng et al. derived solitons, breathers and rational solutions to the BS lattice equation both on the constant and periodic background by Hirota's bilinear method \cite{Sheng}.

This paper is structured as follows. In Section \ref{sec2}, we derive the integrable variant of the BS lattice equation and its Gram-type determinant solution. In Section \ref{sec3}, we investigate the dynamics of one-soliton and two-soliton solutions, alongside an analysis of the asymptotic behaviors of the two-soliton solution. In Section \ref{sec4}, bright lump, dark lump, and fundamental lump solutions are constructed via the B\"acklund transformation, with their formulations expressed in terms of Schur polynomials. In Section \ref{sec5}, breather solutions to the variant BS lattice equation are derived, including the Kuznetsov-Ma breather, Akhmediev breather, and a general breather solution. In Section \ref{sec6}, numerical three-periodic wave solutions are computed using the Gauss-Newton method. Finally, conclusions and discussions are presented.

\section{The derivation of variant BS lattice equation}\label{sec2}
In this section, we derive the variant BS lattice equation through the transformation
\begin{align}\label{trans}
    (D_t,D_y,D_z,D_n)\rightarrow\I(D_t,D_y,D_z,D_n),
\end{align}
where $\I=\sqrt{-1}$.
With the help of (\ref{trans}), we get the novel bilinear equations from \eqref{BSb-1}-\eqref{BSb-2}
\begin{align}
    &\left(D_z\sin\left(\frac{1}{2}D_n\right)-D_t^2\cos\left(\frac{1}{2}D_n\right)\right)\tau(n)\cdot\tau(n)=0,\label{bilinear-1}\\
    &\left(D_tD_z-D_tD_y\right)\tau(n)\cdot\tau(n)=4\sin^2\left(\frac{1}{2}D_n\right)\tau(n)\cdot\tau(n),\label{bilinear-2}
\end{align}
which can be viewed as the integrable variants of (\ref{BSb-1}) and (\ref{BSb-2}). By the dependent variable transformation

\begin{align}\label{solution}
    u(n)=\frac{\cos(D_n)\tau(n)\cdot\tau(n)}{\tau(n)^2},\quad v(n)=\frac{D_t^2\cos(\frac{1}{2}D_n)\tau(n)\cdot\tau(n)}{\cos(\frac{1}{2}D_n)\tau(n)\cdot\tau(n)},\quad w(n)=-2\frac{\partial}{\partial t}\sin\left(\frac{1}{2}\frac{\partial}{\partial n}\right)\ln\tau(n),
\end{align}
we obtain the variant BS lattice equation
\begin{align}
           &\frac{\partial}{\partial t} u(n)=2 u(n)\sin\left(\frac{1}{2}\frac{\partial}{\partial n}\right)w(n),\label{variants-1}\\
        &\frac{\partial}{\partial t} v(n)=-2\sin\left(\frac{1}{2}\frac{\partial}{\partial n}\right)u(n)-\frac{\partial}{\partial y}w(n),\label{variants-2}\\
        &\frac{\partial}{\partial t}\cos\left(\frac{1}{2}\frac{\partial}{\partial n}\right)w(n)=-\sin\left(\frac{1}{2}\frac{\partial}{\partial n}\right)(v(n)+w(n)^2).\label{variants-3}
 \end{align}
It is easy to see that $u(n)$, $v(n)$ and $w(n)$ are real-valued functions of real variables $y$, $t$, $z$ and $n$. Based on the trigonometric-type bilinear operators (\ref{tri}), the variant BS lattice equation can be rewritten as
\begin{align}
    &\frac{\partial}{\partial t} u(n)=\I u(n)\left(w(n-\frac{\I}{2})-w(n+\frac{\I}{2})\right),\\
        &\frac{\partial}{\partial t} v(n)=\I u(n+\frac{\I}{2})-\I u(n-\frac{\I}{2})-\frac{\partial}{\partial y}w(n),\\
        &\frac{\partial}{\partial t}w(n+\frac{\I}{2})+\frac{\partial}{\partial t}w(n-\frac{\I}{2})=\I\left(v(n+\frac{\I}{2})-v(n-\frac{\I}{2})+w(n+\frac{\I}{2})^2-w(n-\frac{\I}{2})^2\right),
\end{align}
through the dependent variable transformation
\begin{align}
    u(n)=\frac{\tau(n+\I)\tau(n-\I)}{\tau(n)^2},\quad v(n)=\frac{D_t^2\tau(n+\frac{\I}{2})\cdot\tau(n-\frac{\I}{2})}{\tau(n+\frac{\I}{2})\tau(n-\frac{\I}{2})},\quad w(n)=\I\frac{\partial}{\partial t}\left(\ln\frac{\tau(n+\frac{\I}{2})}{\tau(n-\frac{\I}{2})}\right).
\end{align}

\begin{theo}
The bilinear equation (\ref{bilinear-1}) and (\ref{bilinear-2}) admit the Gram-type determinant solution
\begin{align}
    \tau(n)=\left|c_{jk}+\int^t\phi_j(n)\psi_k(n)\dif t\right|_{1\leq j,k\leq N},
\end{align}
where $\phi_j(n)$, $\psi_k(n)$ satisfy the following relations
\begin{align}
    &\partial_t\phi_j(n)=-\I\phi_j(n+\I),\quad \partial_y\phi_j(n)=-\I\big(\phi_j(n-\I)+\phi_j(n+2\I)\big),\quad \partial_z\phi_j(n)=-\I\phi_j(n+2\I), \label{dde-1}\\
    &\partial_t\psi_k(n)=\I\psi_k(n-\I),\quad \partial_y\psi_k(n)=\I\big(\psi_k(n+\I)+\psi_k(n-2\I)\big),\quad \partial_z\psi_k(n)=\I\psi_k(n-2\I). \label{dde-2}
\end{align}
\end{theo}
\begin{proof}
    We set
    \begin{align}
        m_{jk}=c_{jk}+\int^t\phi_j(n)\psi_k(n)\dif t.
    \end{align}
    Upon using the differential-difference relations (\ref{dde-1}) and (\ref{dde-2}), we obtain
    \begin{align}
        &\frac{\partial}{\partial t}m_{jk}=\phi_j(n)\psi_k(n),\\
        &\frac{\partial}{\partial y}m_{jk}=-\phi_j(n-\I)\psi_k(n+\I)+\frac{\partial}{\partial z}m_{jk},\\
        &\frac{\partial}{\partial z}m_{jk}=\phi_j(n+\I)\psi_k(n)+\phi_j(n)\psi_k(n-\I),\\
        &\frac{\partial^2}{\partial t^2}m_{jk}=-\I\phi_j(n+\I)\psi_k(n)+\I\phi_j(n)\psi_k(n-\I),\\
        &\frac{\partial^2}{\partial z\partial t}m_{jk}=-\I\phi_j(n+2\I)\psi_k(n)+\I\phi_j(n)\psi_k(n-2\I),\\
        &\frac{\partial^2}{\partial y\partial t}m_{jk}=-\I\phi_j(n-\I)\psi_k(n)-\I\phi_j(n+2\I)\psi_k(n)+\I\phi_j(n)\psi_k(n+\I)+\I\phi_j(n)\psi_k(n-2\I).
    \end{align}
    Therefore, the derivatives of the determinant $\tau(n)$  have the following expression,
    \begin{align*}
        &\frac{\partial}{\partial t}\tau(n)=\left|\begin{array}{cc}
            m_{jk} & \phi_j(n) \\
            -\psi_k(n) & 0
        \end{array}\right|,\\
        &\frac{\partial^2}{\partial t^2}\tau(n)=\I\left|\begin{array}{cc}
            m_{jk} & \phi_j(n) \\
            -\psi_k(n-\I) & 0
        \end{array}\right|-\I\left|\begin{array}{cc}
            m_{jk} & \phi_j(n+\I) \\
            -\psi_k(n) & 0
        \end{array}\right|,\\
        &\frac{\partial}{\partial z}\tau(n)=\left|\begin{array}{cc}
            m_{jk} & \phi_j(n+\I) \\
            -\psi_k(n) & 0
        \end{array}\right|+\left|\begin{array}{cc}
            m_{jk} & \phi_j(n) \\
            -\psi_k(n-\I) & 0
        \end{array}\right|,\\
        &\frac{\partial}{\partial y}\tau(n)=\left|\begin{array}{cc}
            m_{jk} & \phi_j(n-\I) \\
            \psi_k(n+\I) & 0
        \end{array}\right|+\left|\begin{array}{cc}
            m_{jk} & \phi_j(n+\I) \\
            -\psi_k(n) & 0
        \end{array}\right|+\left|\begin{array}{cc}
            m_{jk} & \phi_j(n) \\
            -\psi_k(n-\I) & 0
        \end{array}\right|,\\
        &\frac{\partial^2}{\partial y\partial t}\tau(n)=\I\left|\begin{array}{cc}
            m_{jk} & \phi_j(n-\I) \\
            \psi_k(n) & 0
        \end{array}\right|-\I\left|\begin{array}{cc}
            m_{jk} & \phi_j(n) \\
            \psi_k(n+\I) & 0
        \end{array}\right|-\I\left|\begin{array}{cc}
            m_{jk} & \phi_j(n+2\I) \\
            -\psi_k(n) & 0
        \end{array}\right|\\
        &~~~~~~~~~~~~~+\I\left|\begin{array}{cc}
            m_{jk} & \phi_j(n) \\
            -\psi_k(n-2\I) & 0
        \end{array}\right|+\left|\begin{array}{ccc}
            m_{jk} & \phi_j(n-\I) & \phi_j(n) \\
            \psi_k(n+\I) & 0 & 0 \\
            -\psi_k(n) & 0 & 0
        \end{array}\right|,\\
        &\frac{\partial^2}{\partial z\partial t}\tau(n)=-\I\left|\begin{array}{cc}
            m_{jk} & \phi_j(n+2\I) \\
            -\psi_k(n) & 0
        \end{array}\right|+\I\left|\begin{array}{cc}
            m_{jk} & \phi_j(n) \\
            -\psi_k(n-2\I) & 0
        \end{array}\right|,\\
        &\tau(n+\I)=\tau(n)+\I\left|\begin{array}{cc}
            m_{jk} & \phi_j(n) \\
            -\psi_k(n+\I) & 0
        \end{array}\right|,\\
        &\tau(n-\I)=\tau(n)+\I\left|\begin{array}{cc}
            m_{jk} & \phi_j(n-\I) \\
            \psi_k(n) & 0
        \end{array}\right|,\\
        &\frac{\partial}{\partial z}\tau(n+\I)=\left|\begin{array}{cc}
            m_{jk} & \phi_j(n+\I) \\
            -\psi_k(n) & 0
        \end{array}\right|+\left|\begin{array}{cc}
            m_{jk} & \phi_j(n+2\I) \\
            -\psi_k(n+\I) & 0
        \end{array}\right|+\I\left|\begin{array}{ccc}
            m_{jk} & \phi_j(n) & \phi_j(n+\I) \\
            -\psi_k(n+\I) & 0 & 0\\
            -\psi_k(n) & 0 & 0
        \end{array}\right|,\\
        &\frac{\partial}{\partial t}\tau(n+\I)=\left|\begin{array}{cc}
            m_{jk} & \phi_j(n+\I) \\
            -\psi_k(n+\I) & 0
        \end{array}\right|,\\
        &\frac{\partial^2}{\partial t^2}\tau(n+\I)=\I\left|\begin{array}{cc}
            m_{jk} & \phi_j(n+\I) \\
            -\psi_k(n) & 0
        \end{array}\right|-\I\left|\begin{array}{cc}
            m_{jk} & \phi_j(n+2\I) \\
            -\psi_k(n+\I) & 0
        \end{array}\right|-\left|\begin{array}{ccc}
            m_{jk} & \phi_j(n) & \phi_j(n+\I) \\
            -\psi_k(n+\I) & 0 & 0\\
            -\psi_k(n) & 0 & 0
        \end{array}\right|.
    \end{align*}
    Substituting the above  expressions into the equations (\ref{bilinear-1}) and (\ref{bilinear-2}), we arrive at the Jacobi identities for the determinants
    \begin{align*}
        \left|\begin{array}{cc}
            m_{jk} & \phi_j(n) \\
            -\psi_k(n) & 0
        \end{array}\right|\cdot\left|\begin{array}{cc}
            m_{jk} & \phi_j(n+\I) \\
            -\psi_k(n+\I) & 0
        \end{array}\right|-\left|\begin{array}{cc}
            m_{jk} & \phi_j(n+\I) \\
            -\psi_k(n) & 0
        \end{array}\right|\cdot\left|\begin{array}{cc}
            m_{jk} & \phi_j(n) \\
            -\psi_k(n+\I) & 0
        \end{array}\right|~~~~~~\\
        +\tau_n\cdot\left|\begin{array}{ccc}
            m_{jk} & \phi_j(n) & \phi_j(n+\I) \\
            -\psi_k(n+\I) & 0 & 0\\
            -\psi_k(n) & 0 & 0
        \end{array}\right|=0,\\
        \left|\begin{array}{cc}
            m_{jk} & \phi_j(n) \\
            \psi_k(n) & 0
        \end{array}\right|\cdot\left|\begin{array}{cc}
            m_{jk} & \phi_j(n-\I) \\
            \psi_k(n+\I) & 0
        \end{array}\right|-\left|\begin{array}{cc}
            m_{jk} & \phi_j(n-\I) \\
            \psi_k(n) & 0
        \end{array}\right|\cdot\left|\begin{array}{cc}
            m_{jk} & \phi_j(n) \\
            \psi_k(n+\I) & 0
        \end{array}\right|~~~~~~\\
        +\tau_n\cdot\left|\begin{array}{ccc}
            m_{jk} & \phi_j(n-\I) & \phi_j(n) \\
            \psi_k(n+\I) & 0 & 0\\
            -\psi_k(n) & 0 & 0
        \end{array}\right|=0.
    \end{align*}
    We complete the proof.
\end{proof}

\section{Soliton solutions to the variant BS lattice equation}\label{sec3}
In this part, we construct soliton solutions to the variant BS lattice equation (\ref{variants-1})-(\ref{variants-3}). We take
\begin{align}
    &\phi_j(n)=e^{\eta_j}, \quad \psi_k(n)=e^{\xi_k}, \quad m_{jk}(n)=\delta_{jk}+\frac{\I}{e^{\I p_j}-e^{-\I q_k}}\phi_j(n)\psi_k(n),
\end{align}
with
\begin{align}
    &\eta_j=p_jn-\I e^{\I p_j}t-\I(e^{-\I p_j}+e^{2\I p_j})y-\I e^{2\I p_j}z+\eta_{0,j},\\
    &\xi_k=q_kn+\I e^{-\I q_k}t+\I(e^{\I q_k}+e^{-2\I q_k})y+\I e^{-2\I q_k}z+\xi_{0,k},
\end{align}
where $p_j$, $q_k$, $\eta_{0,j}$ and $\xi_{0,k}$ are arbitrary complex constants.

By setting $N=1$, $p_1=a+\I b$, $q_1=p_1^*$, $c_{11}=1$, $z=0$ and $\eta_{0,1}=\xi_{0,1}=0$, $\tau(n)$ is given by
\begin{align}
    \tau(n)=1+\frac{1}{2\sin(a)}e^{b+2an+2e^{-b}\sin(a)t+2(-e^b\sin(a)+e^{-2b}\sin(2a))y}.
\end{align}
The one-soliton solutions $u(n)$ and $w(n)$ are non-singular when $\sin(a)>0$. Choosing parameters $a=1$, $b=1$, we have one-soliton solutions
\begin{align}
    &u(n)=\frac{\tau(n+\I)\tau(n-\I)}{\tau(n)^2}=1-\frac{1-\cos(2)}{1+\cosh(\zeta-\ln(2\sin(1)))},\label{solitonu}\\
    &w(n)=\I\left(\ln\frac{\tau(n+\I/2)}{\tau(n-\I/2)}\right)_t=\frac{-2e^{\zeta-1}\sin(1)}{1+\frac{1}{4\sin^2(1)}e^{2\zeta}+\cot(1)e^{\zeta}},
\end{align}
where $\zeta=1+2n+2e^{-1}\sin(1)t+2(-e\sin(1)+e^{-2}\sin(2))y$. Expression (\ref{solitonu}) describes a dark one-soliton solution of $u(n)$ with the minimum amplitude $\frac{1+\cos(2)}{2}$. The profiles of $u(n)$ and $w(n)$ are plotted in Fig. \ref{fig1}.

\begin{figure}
\begin{center}
\begin{tabular}{cc}
\includegraphics[height=0.220\textwidth,angle=0]{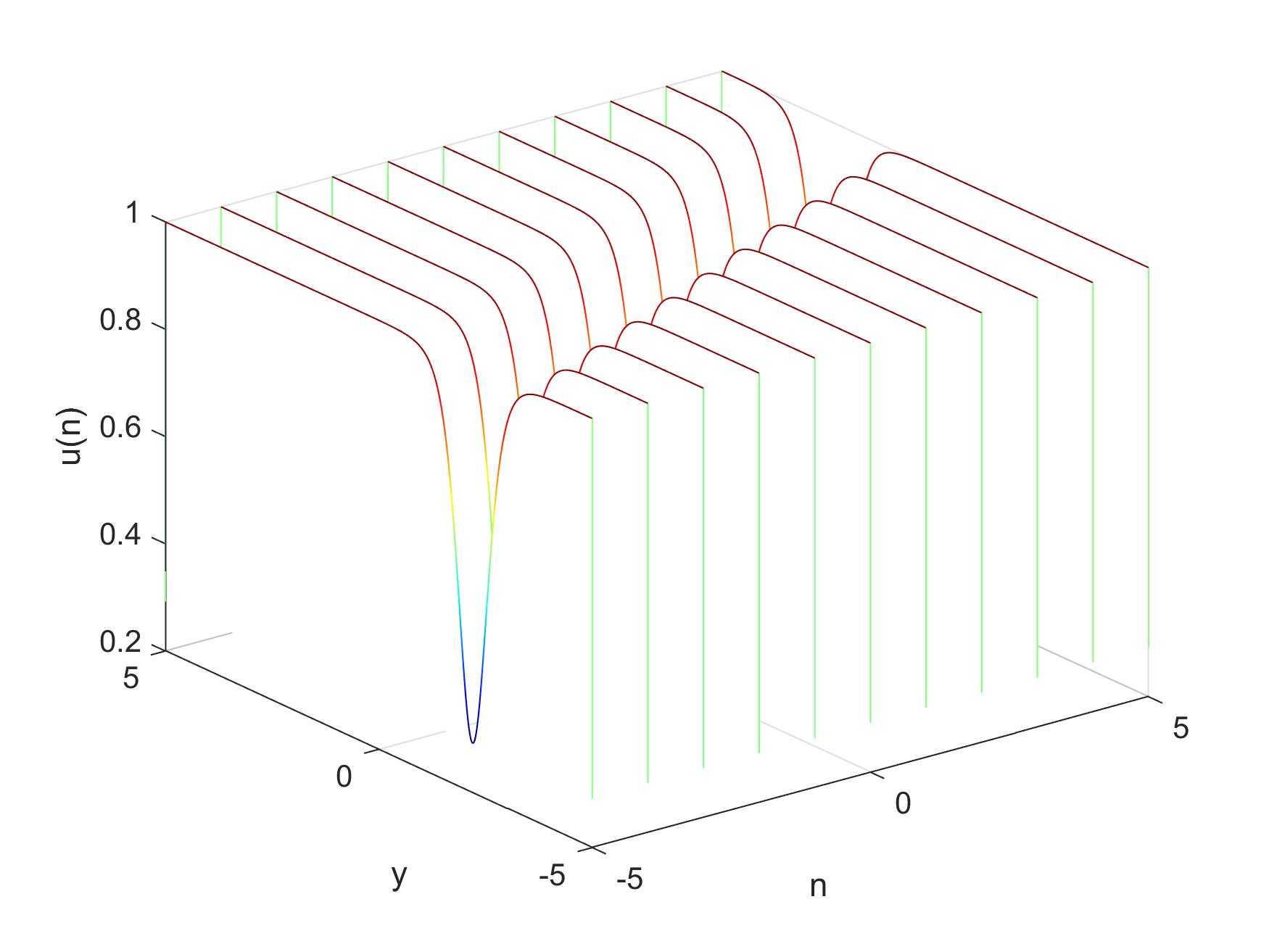} &
\includegraphics[height=0.220\textwidth,angle=0]{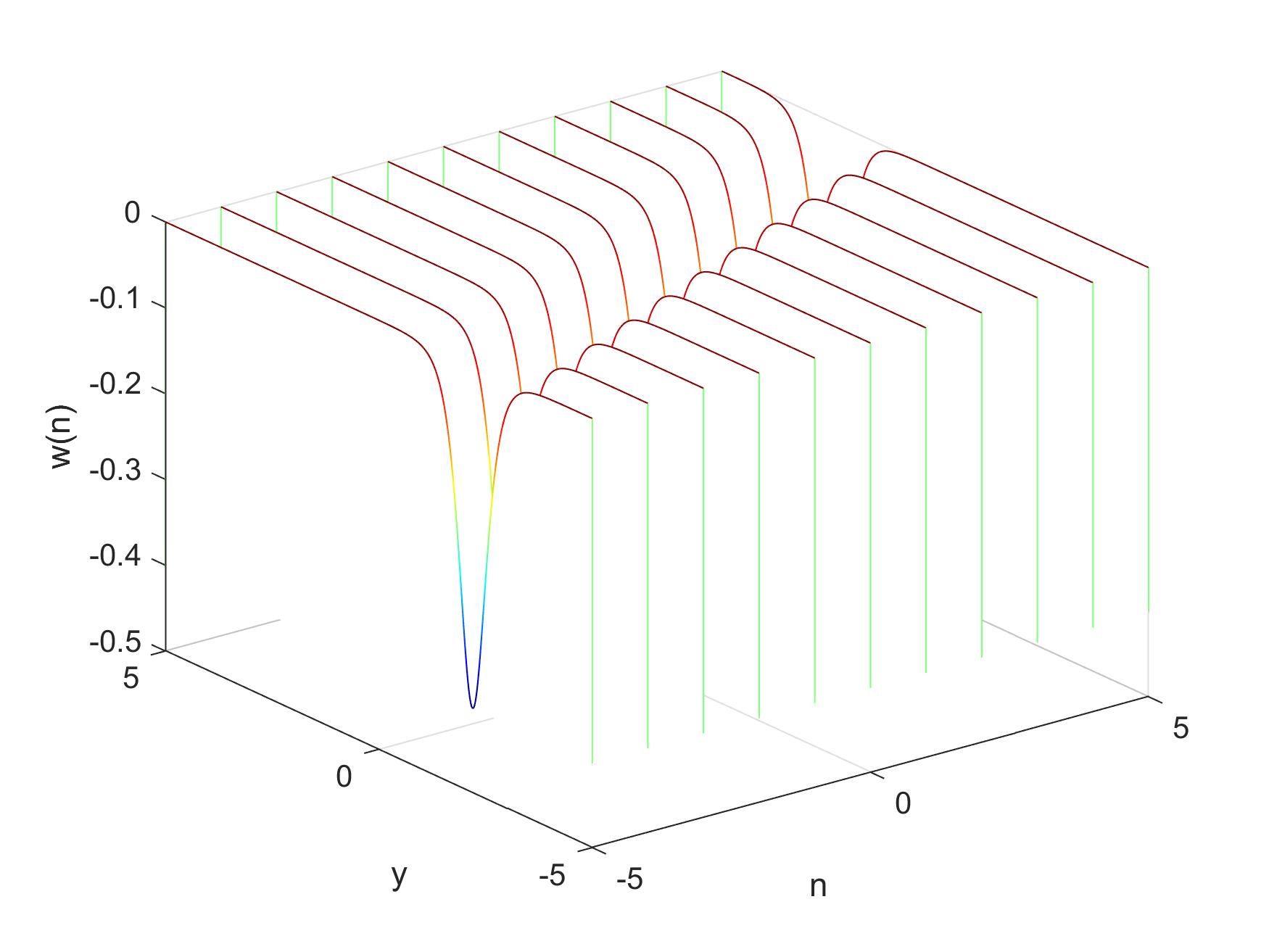} \\
(a) $u(n)$  & \quad  (b) $w(n)$ 
\end{tabular}
\end{center}
\caption{
One-soliton with parameters:
$p_1=1+\I$, $q_1=p_1^*$, $c_{11}=1$, $\eta_{0,1}=\xi_{0,1}=0$, $z=0$ and $t=0$.}\label{fig1}
\end{figure}

Now we take $N=2$, $c_{ij}=\delta_{ij}$, $p_i=a_i+\I b_i$, $q_i=p_i^*$, $(i=1,\ 2)$, which leads to
\begin{align}
    \tau(n)&=\left|\delta_{ij}+\frac{\I}{e^{\I p_j}-e^{-\I q_k}}\phi_j(n)\psi_k(n)\right|_{1\leq i,j\leq 2}\\
    &=1+\left(\frac{1}{4\sin(a_1)\sin(a_2)}+\frac{1}{|Z|^2}\right)e^{\zeta_1+\zeta_2}+\frac{e^{\zeta_1}}{2\sin(a_1)}+\frac{e^{\zeta_2}}{2\sin(a_2)},
\end{align}
where $Z=-\I e^{\I p_1}+\I e^{-\I q_2}$, and
\begin{align}
    \zeta_1=b_1+2a_1n+2e^{-b_1}\sin(a_1)t+2(-e^{b_1}\sin(a_1)+e^{-2b_1}\sin(2a_1))y,\\
    \zeta_2=b_2+2a_2n+2e^{-b_2}\sin(a_2)t+2(-e^{b_2}\sin(a_2)+e^{-2b_2}\sin(2a_2))y.
\end{align}
For $a_1=a_2=1$, $b_1=1, b_2=-\frac{1}{4}$ and $n=1$, we have a two-soliton solution $u(1)$  displayed in Fig. \ref{fig2}. The detailed expression of the two-soliton solution  is very complicated and we omit it here.

We analyze the asymptotic behaviors of the two-soliton solution $u(1)$. The left-propagating soliton along the direction $\zeta_1$ is denoted as Soliton 1, and the right-propagating one along the direction $\zeta_2$ as Soliton 2. Its asymptotic behaviors  are elaborated in the following.

(1) Before collision $(t\rightarrow-\infty)$:

Soliton 1 $(\zeta_1 =\mathcal{O}(1), \zeta_2\rightarrow-\infty)$
\begin{align}
    u(1)\sim 1-\frac{1-\cos(2)}{1+\cosh(\zeta_1-\ln(2\sin(1)))},
\end{align}

Soliton 2 $(\zeta_2=\mathcal{O}(1), \zeta_1\rightarrow+\infty)$
\begin{align}
    u(1)\sim 1-\frac{1-\cos(\frac{1}{2})}{1+\cosh(\zeta_2+\ln(2\sin(1))+\Omega)},
\end{align}
with $\Omega=\ln\left(\frac{1}{4\sin(a_1)\sin(a_2)}+\frac{1}{|Z|^2}\right)$.

(2) After collision $t\rightarrow+\infty$:

Soliton 1 $(\zeta_1 =\mathcal{O}(1), \zeta_2\rightarrow+\infty)$
\begin{align}
    u(1)\sim 1-\frac{1-\cos(2)}{1+\cosh(\zeta_1+\ln(2\sin(1))+\Omega)},
\end{align}

Soliton 2 $(\zeta_2 =\mathcal{O}(1), \zeta_1\rightarrow-\infty)$
\begin{align}
    u(1)\sim 1-\frac{1-\cos(\frac{1}{2})}{1+\cosh(\zeta_2-\ln(2\sin(1)))}.
\end{align}
We conclude that the collision between the two solitons is demonstrated to be completely elastic.
\begin{figure}
\begin{center}
\begin{tabular}{ccc}
\includegraphics[height=0.220\textwidth,angle=0]{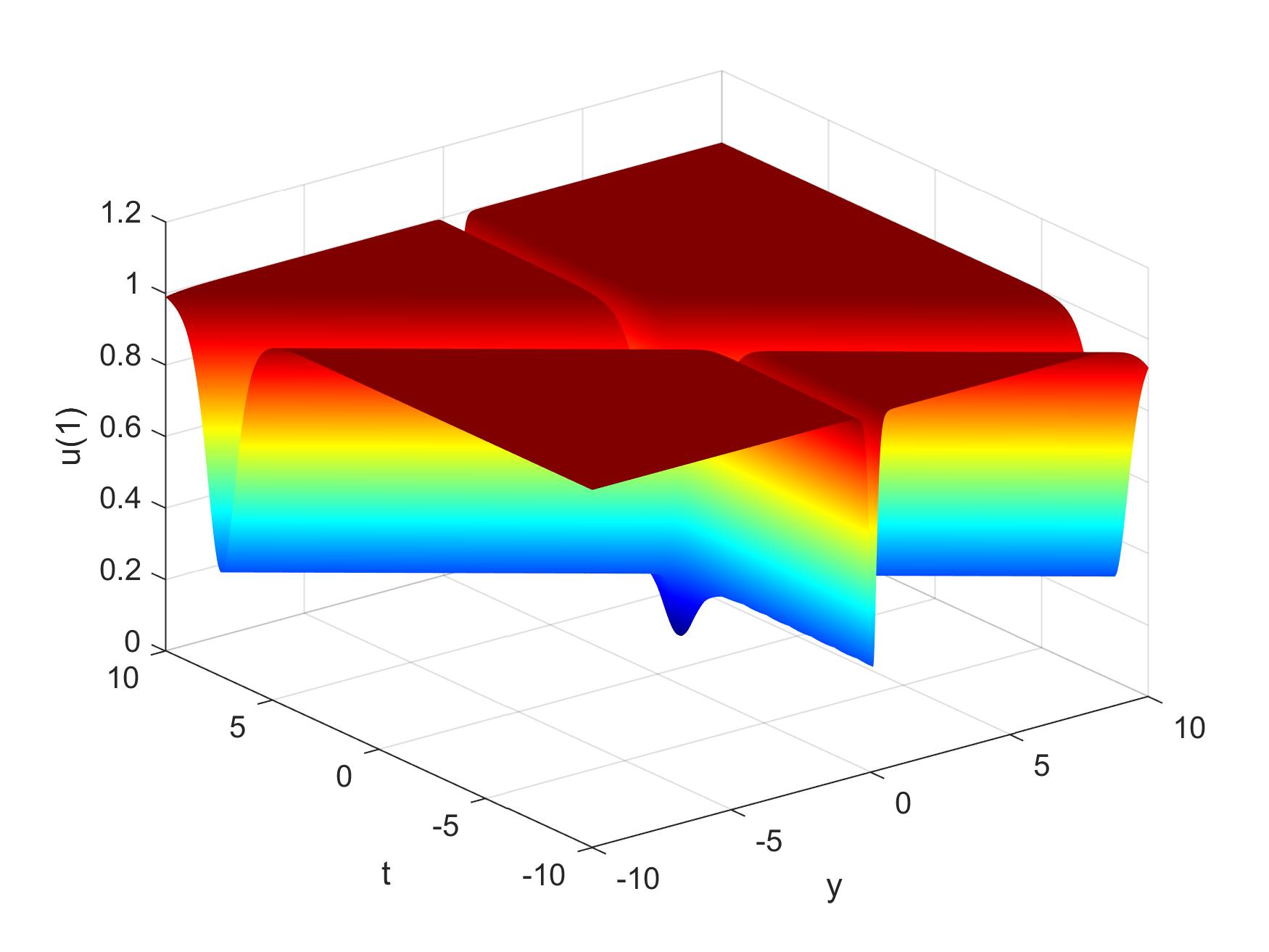} &
\includegraphics[height=0.220\textwidth,angle=0]{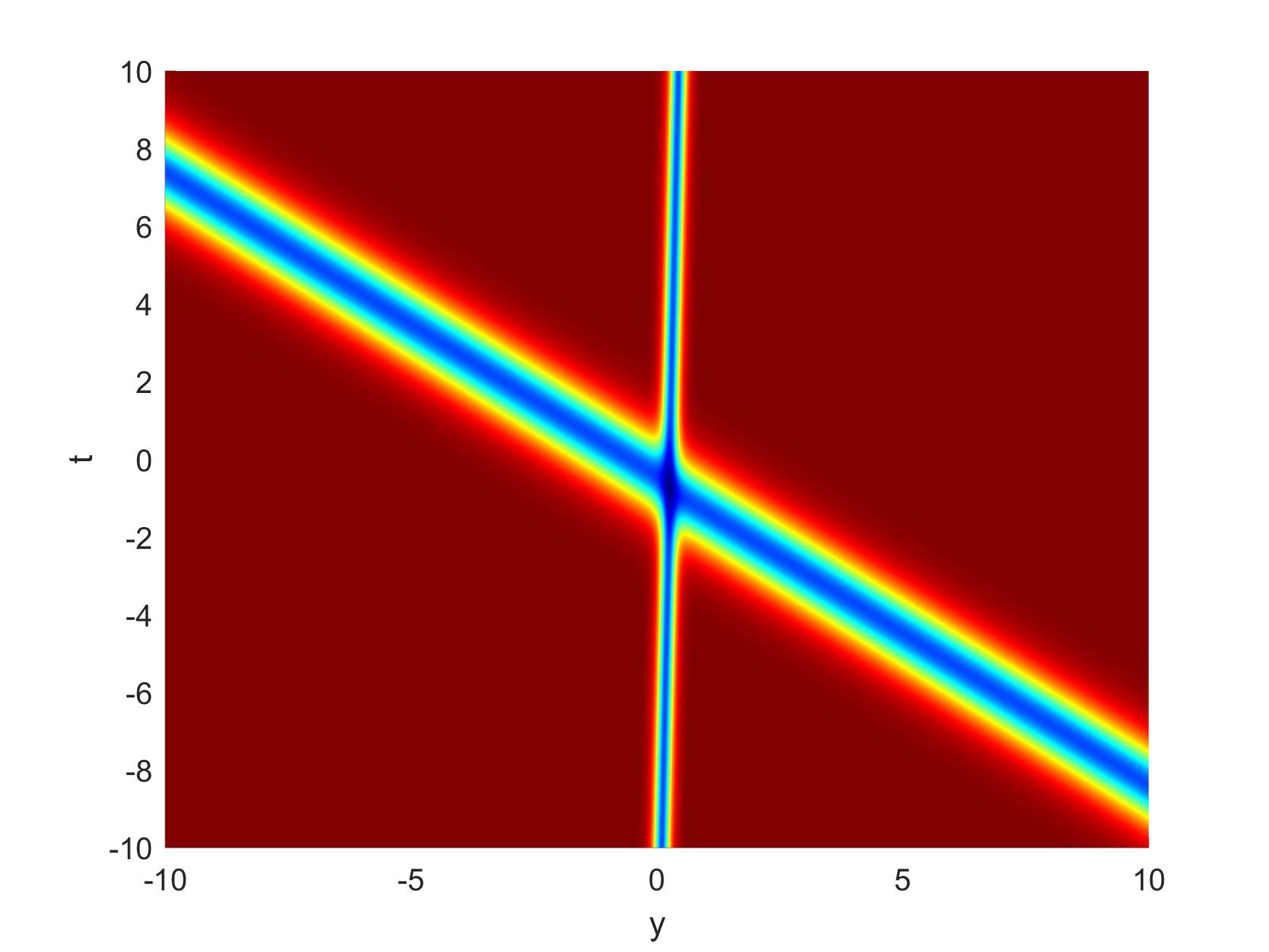} \\
(a) $u(1)$  & \quad  (b) density plot of (a)
\end{tabular}
\end{center}
\caption{
Two-soliton solution with parameters:
$p_1=1+\I$, $p_2=1-\frac{1}{4}\I$, $\eta_{0,i}=\xi_{0,i}=0$ and $z=0$.}\label{fig2}
\end{figure}

\section{Lump solutions to the variant BS lattice equation}\label{sec4}
In this section, lump solutions $u(n)$ to the variant BS lattice equations are investigated via the B\"acklund transformation (BT) and Schur polynomials.

\subsection{B\"acklund transformation and fundamental lump solutions}\label{sec4.1}
The BT to the BS lattice was presented in  \cite{BT}.  Here we  construct the BT to the variant BS lattice equation \eqref{bilinear-1}-\eqref{bilinear-2}.

\begin{prop}
The variant BS lattice equations admit the following B\"acklund transformation
\begin{align}
    &(\I D_t+\lambda^{-1}e^{-\I D_n}+\mu)\tau(n)\cdot\sigma(n)=0,\\
    &(\I D_ze^{-\frac{\I D_n}{2}}-\I D_ye^{-\frac{\I D_n}{2}}-\lambda e^{\frac{\I D_n}{2}}+\gamma e^{-\frac{\I D_n}{2}})\tau(n)\cdot\sigma(n)=0,\\
    &(\I D_z-\I\lambda^{-1}D_te^{-\I D_n}-\lambda^{-1}\mu e^{-\I D_n}-\omega)\tau(n)\cdot\sigma(n)=0,
\end{align}
where $\lambda$, $\mu$, $\gamma$ and $\omega$ are arbitrary constants.
\end{prop}
The proof follows a line of reasoning analogous to that presented in \cite{BT}, and relevant details can be found therein.
Let $\tau_0(n)$ be a nonzero solution of \eqref{variants-1}-\eqref{variants-3} and suppose that $\tau_1(n)$ and $\tau_2(n)$ are solutions of \eqref{variants-1}-\eqref{variants-3} such that they are linked through the BT, i.e, $\tau_0(n)\xrightarrow{\lambda_j,\ \mu_j,\ \gamma_j,\ \omega_j}\tau_j(n)$. Then we can derive the nonlinear superposition formula
\begin{align}\label{NSF}
    e^{-\frac{\I}{2}D_n}\tau_0(n)\cdot\tau_{12}(n)=c\left(\lambda_1e^{-\frac{\I}{2}D_n}-\lambda_2e^{\frac{\I}{2}D_n}\right)\tau_1(n)\cdot\tau_2(n),
\end{align}
where $\tau_{12}(n)$ is new solution of bilinear equations (\ref{bilinear-1}) and (\ref{bilinear-2}), and $c$ is an arbitrary constant. Starting  from $\tau_0(n)=1$ and the parameter settings $\mu_j=-\lambda_j^{-1}$, $\gamma_j=\lambda_j$ $\omega_j=\lambda_j^{-2}$,  we derive the new corresponding solution via the BT,
\begin{align}\label{lumptau}
    \tau_j(n)=n+\lambda_j^{-1}t+\left(2\lambda_j^{-2}-\lambda_j\right)y+2\lambda_j^{-2}z+\beta_j\equiv\theta_j+\beta_j,\ (j=1,\ 2).
\end{align}
We take $c=\frac{1}{\lambda_1-\lambda_2}$, $\beta_1=\frac{\lambda_1\I}{\lambda_1-\lambda_2}$, $\beta_1=\frac{\lambda_2\I}{\lambda_2-\lambda_1}$ and $\lambda_1=\lambda_2^*$.  It is easy to verify that
\begin{align}\label{tau12}
    \tau_{12}(n)=\theta_1\theta_2-\frac{\lambda_1\lambda_2}{(\lambda_1-\lambda_2)^2}>0.
\end{align}
Substituting $\tau_{12}(n)$ into (\ref{solution}) leads to one-lump solution $u(n)$ expressed as
\begin{align}
    u(n)=\frac{(\tau_{12}(n)-1)^2+(\theta_1+\theta_2)^2}{\tau_{12}^2(n)}.
\end{align}
Let $\lambda_1=\lambda_2^*=a+\I b$, then the expression $\theta_1$ can be rewritten as
\begin{equation}\label{theta1}
    \begin{split}
        \theta_1&=n+\frac{a}{a^2+b^2}t+\left(\frac{2a^2-2b^2}{(a^2+b^2)^2}-a\right)y+\I\left(-\frac{b}{a^2+b^2}t-\left(\frac{2ab}{(a^2+b^2)^2}+b\right)y\right)\\
        &\equiv\mathcal{R}+\I\mathcal{I},
    \end{split}
\end{equation}
where $\mathcal{R}$ and $\mathcal{I}$ denote the real part and imaginary part of $\theta_1$, respectively.
Since $\lambda_1=\lambda^*_2$, we have $\theta_2=\mathcal{R}-\I \mathcal{I}$. We obtain one-lump solution \begin{align}\label{u-1lump}
    u(n)=1+\frac{2\left(\mathcal{R}^2-\mathcal{I}^2+\frac{b^2-a^2}{4b^2}\right)}{\left(\mathcal{R}^2+\mathcal{I}^2+\frac{a^2+b^2}{4b^2}\right)^2}.
\end{align}
When $\frac{|b|}{\sqrt{3}}<|a|<\sqrt{3}|b|$, we obtain a fundamental lump with two global maximum points and two global minimum points.
When $|a|\geq\sqrt{3}|b|$, we have the dark lump solution, which possesses two global maximum points and one global minimum point.
When $|a|\leq\frac{|b|}{\sqrt{3}}$, the bright lump solution is obtained, featuring one global maximum point and two global minimum points.
For instance, by choosing the parameter sets $a=1$, $b=4$; $a=2$, $b=1$ and $a=b=1$, the bright lump, dark lump and fundamental lump solution are obtained, which are displayed in Fig. \ref{fig3} (a) (b) and (c), respectively. Utilizing (\ref{tau12}) and (\ref{theta1}), we get  one-lump solution
\begin{align}
    w(n)=\frac{\frac{2a}{a^2+b^2}(\mathcal{R}^2-\mathcal{I}^2-\frac{a^2}{4b^2})-\frac{4b}{a^2+b^2}\mathcal{RI}}{\left(\mathcal{R}^2+\mathcal{I}^2+\frac{a^2}{4b^2}\right)^2+\mathcal{R}^2}.
\end{align}
The one-lump solution $w(n)$ is non-singular when $a\neq 0$. The dark lump solution is obtained under the condition $a>0$, whereas the bright lump solution is constructed by $a<0$. Visualization of these two types of one-lump solutions are presented in Fig. \ref{wlump}.
\begin{figure}
\begin{center}
\begin{tabular}{ccc}
\includegraphics[height=0.220\textwidth,angle=0]{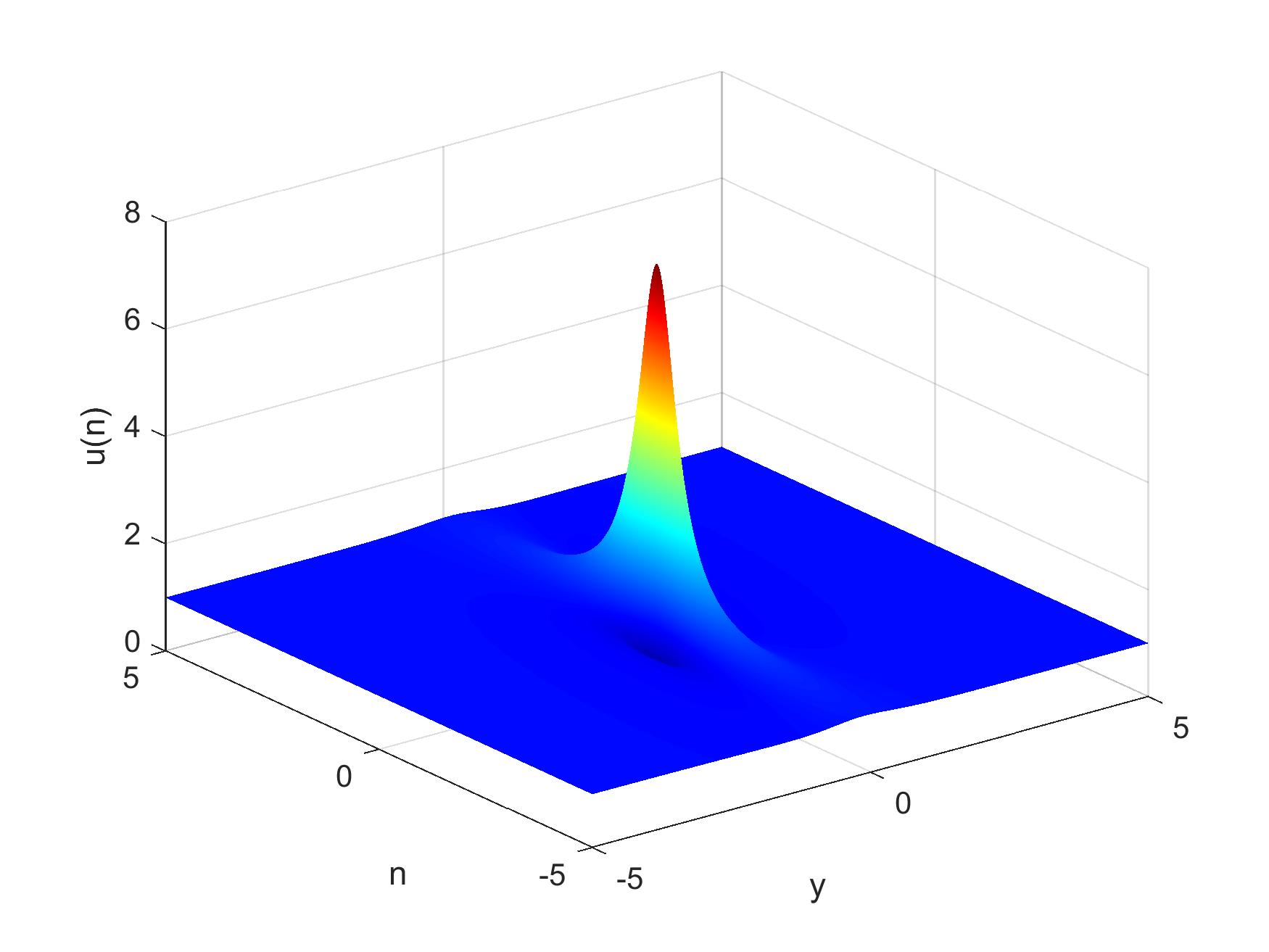} &
\includegraphics[height=0.220\textwidth,angle=0]{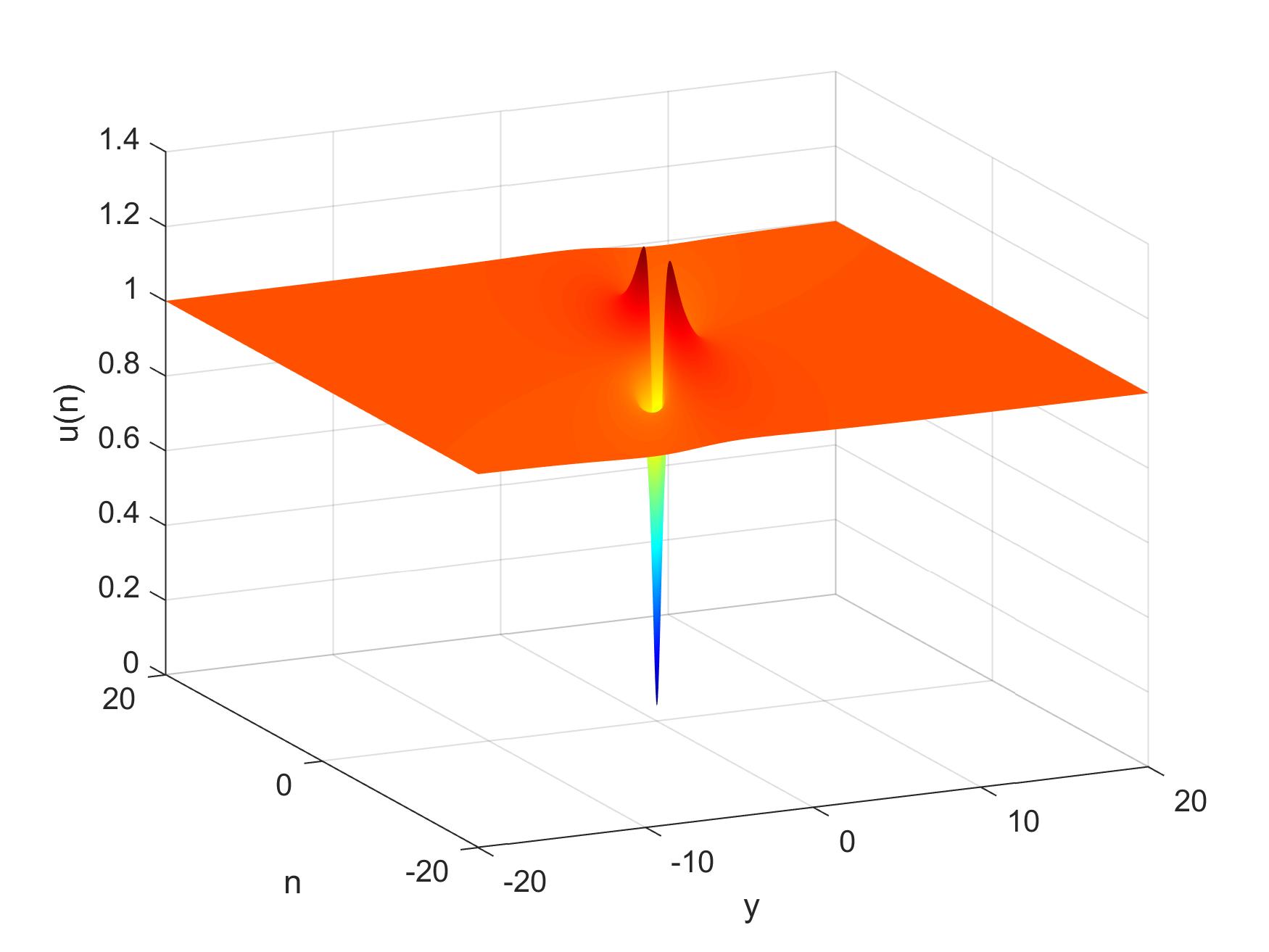} &
\includegraphics[height=0.220\textwidth,angle=0]{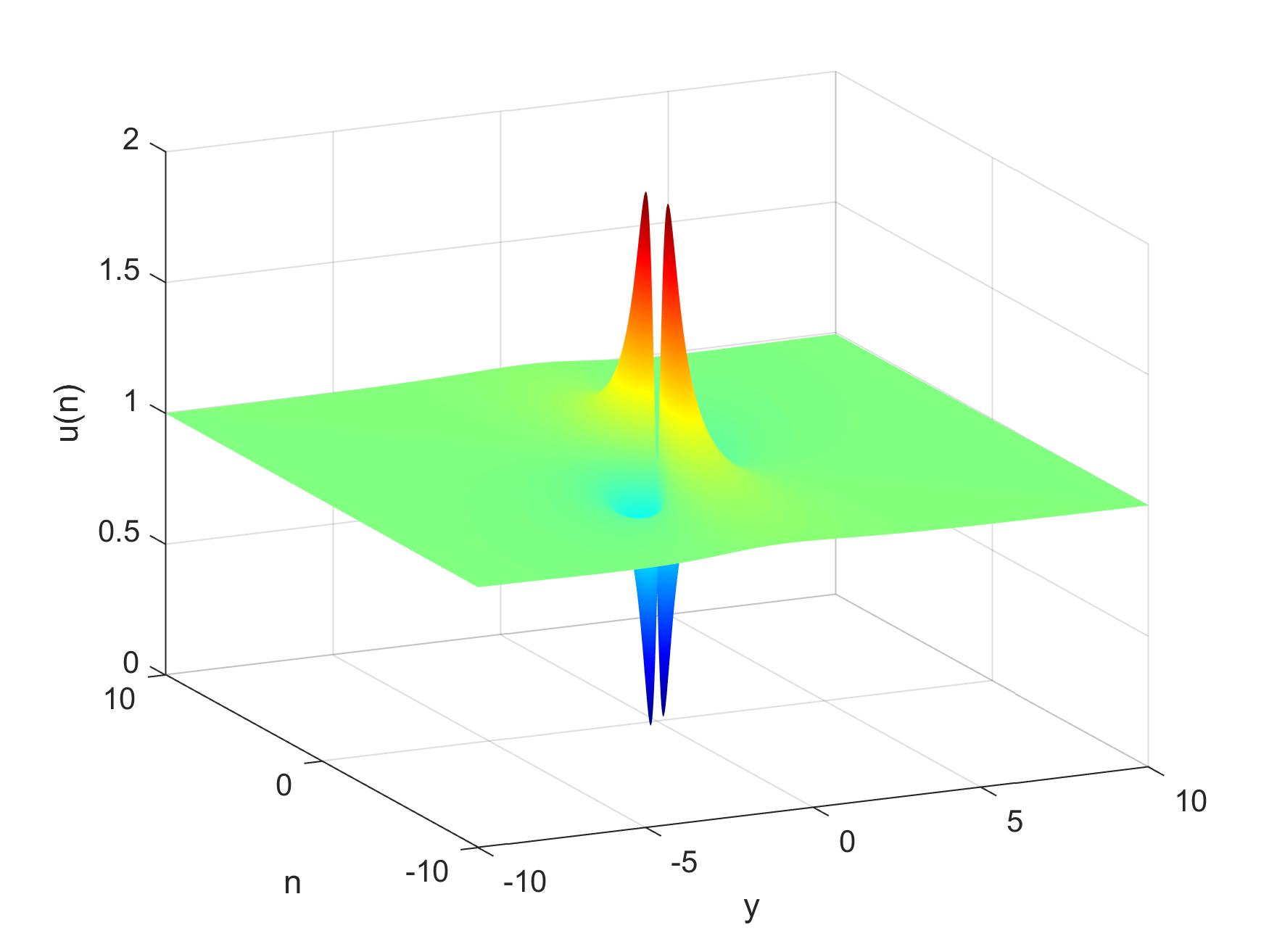}\\
(a) Bright lump $u(n)$  & \quad  (b) Dark lump $u(n)$ & \quad (c) Fundamental lump $u(n)$\\
\includegraphics[height=0.220\textwidth,angle=0]{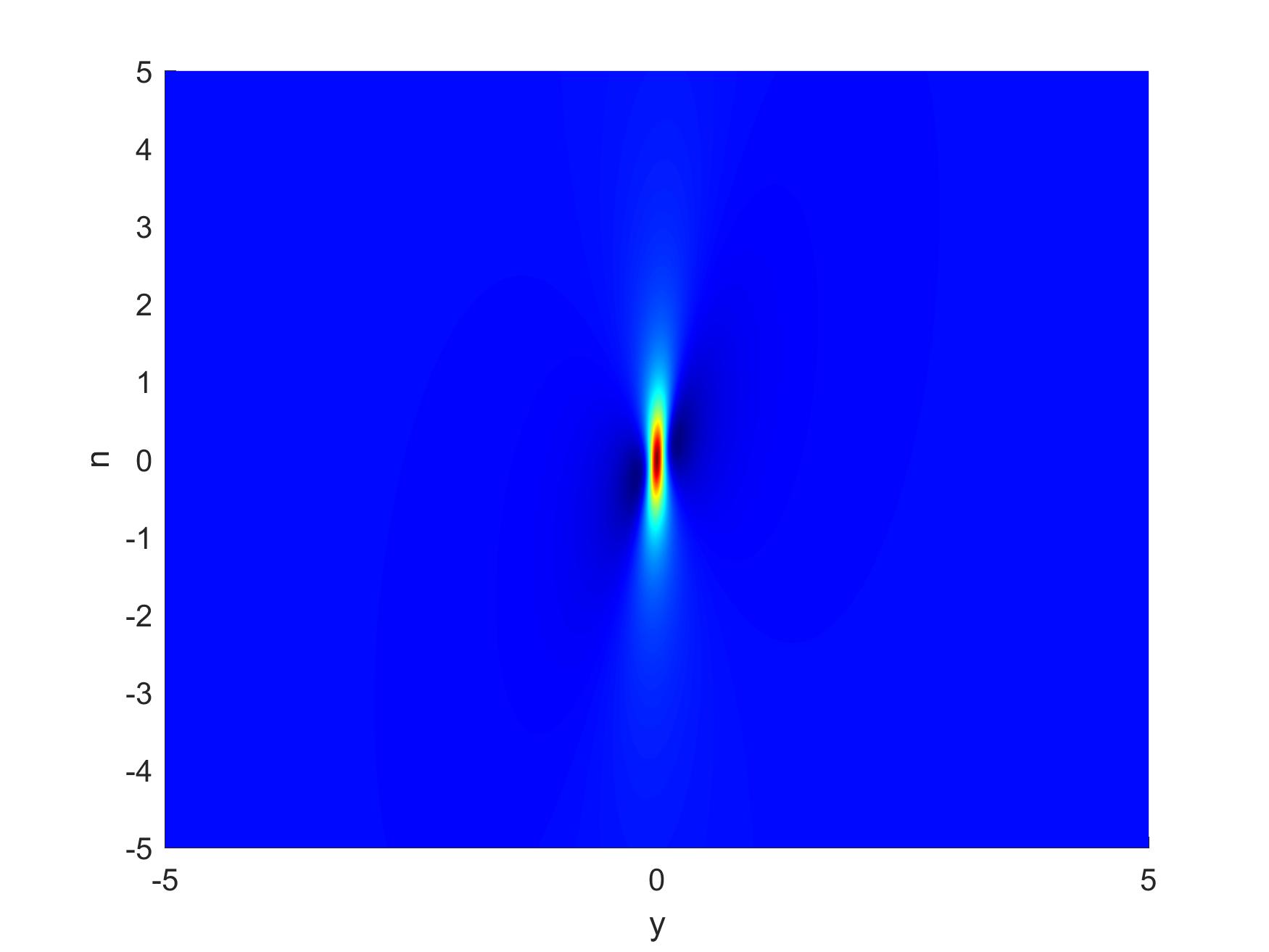} &
\includegraphics[height=0.220\textwidth,angle=0]{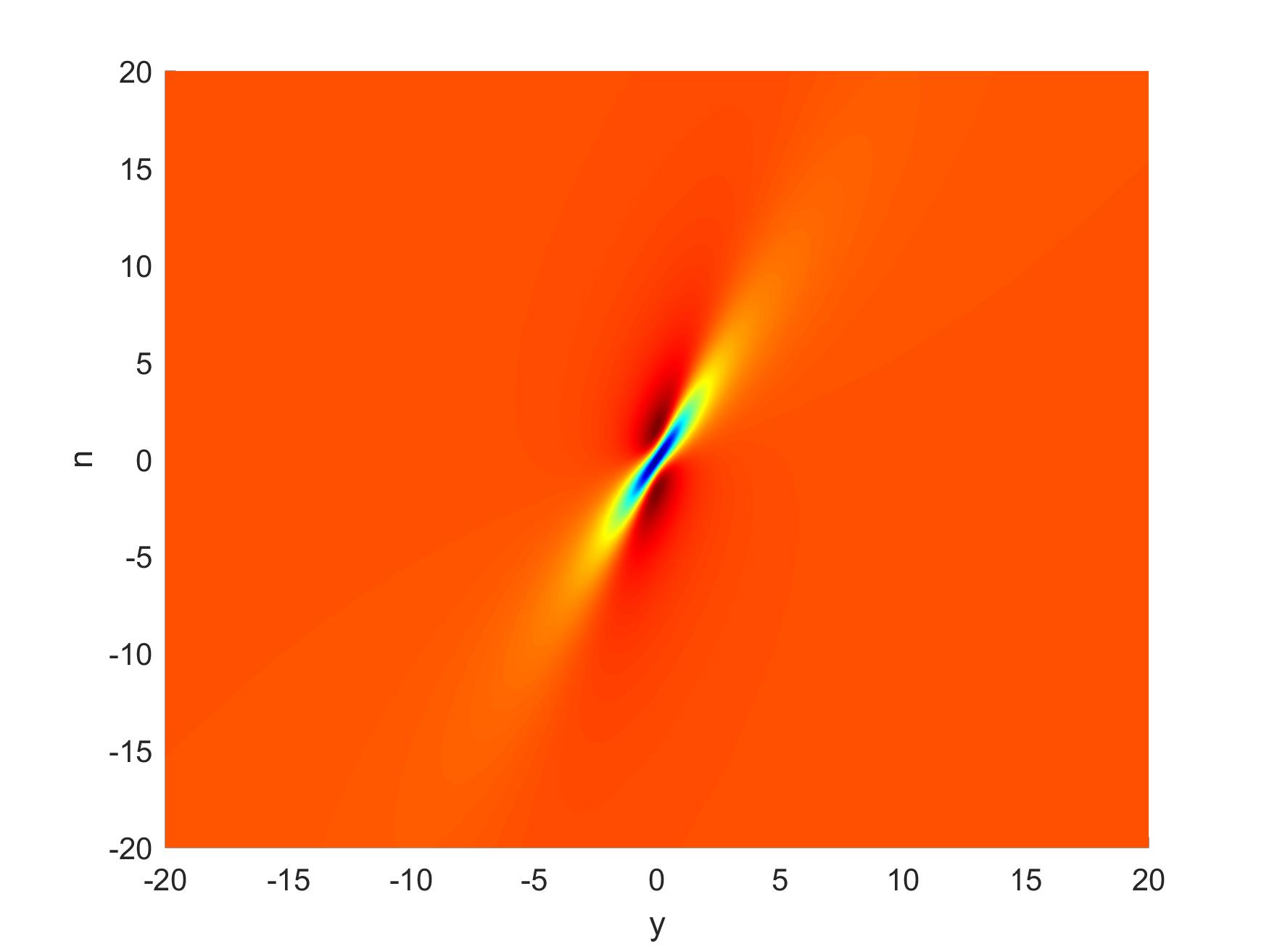} &
\includegraphics[height=0.220\textwidth,angle=0]{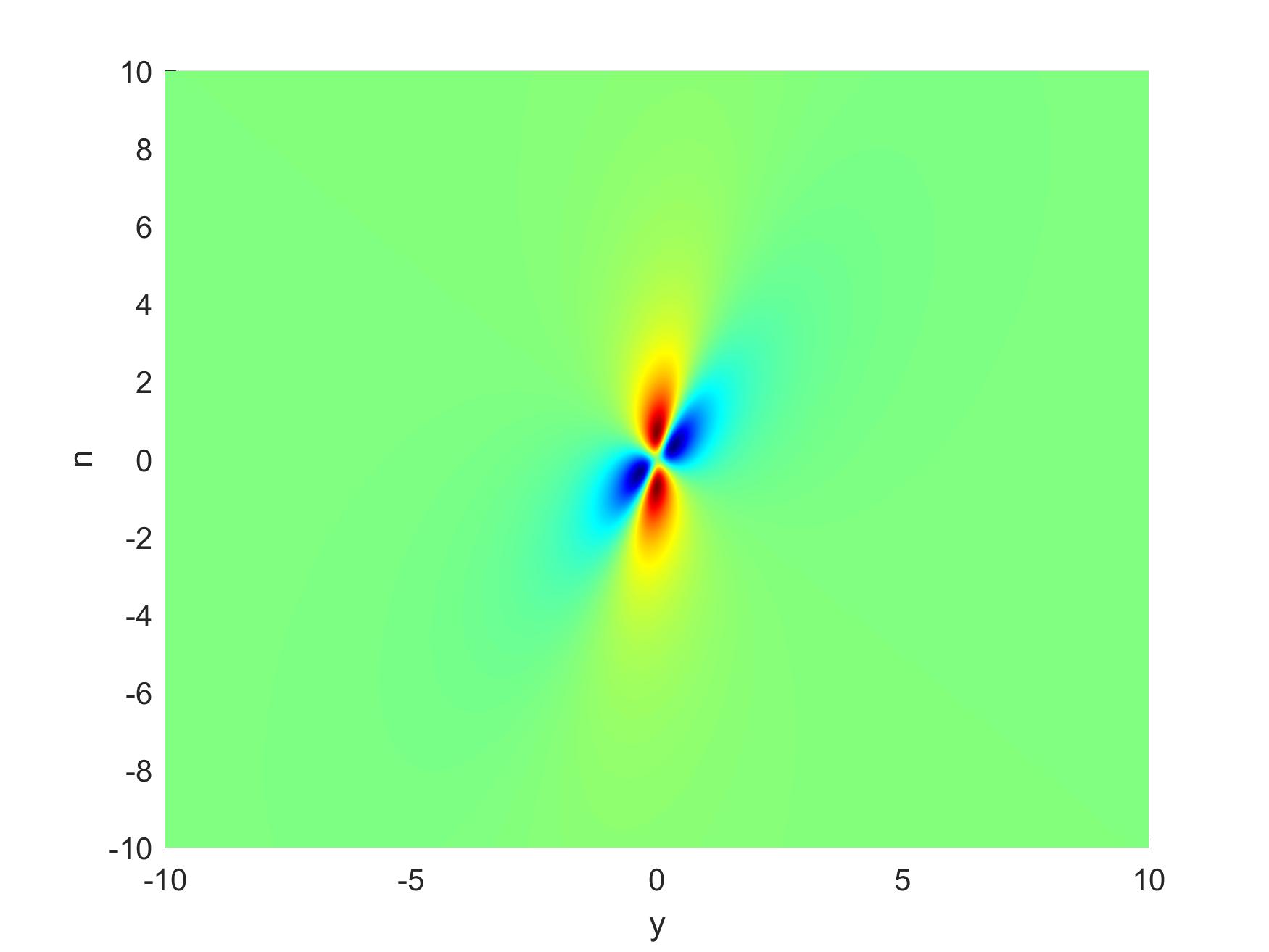}\\
(d) density plot of (a)  & \quad (e) density plot of (b) & \quad (f) density plot of (c)
\end{tabular}
\end{center}
\caption{
lump solutions with the parameters:
(a) $\lambda_1=1+4\I$; (b) $\lambda_1=2+\I$; (c) $\lambda_1=1+\I$.}\label{fig3}
\end{figure}

\begin{figure}
\begin{center}
\begin{tabular}{ccc}
\includegraphics[height=0.220\textwidth,angle=0]{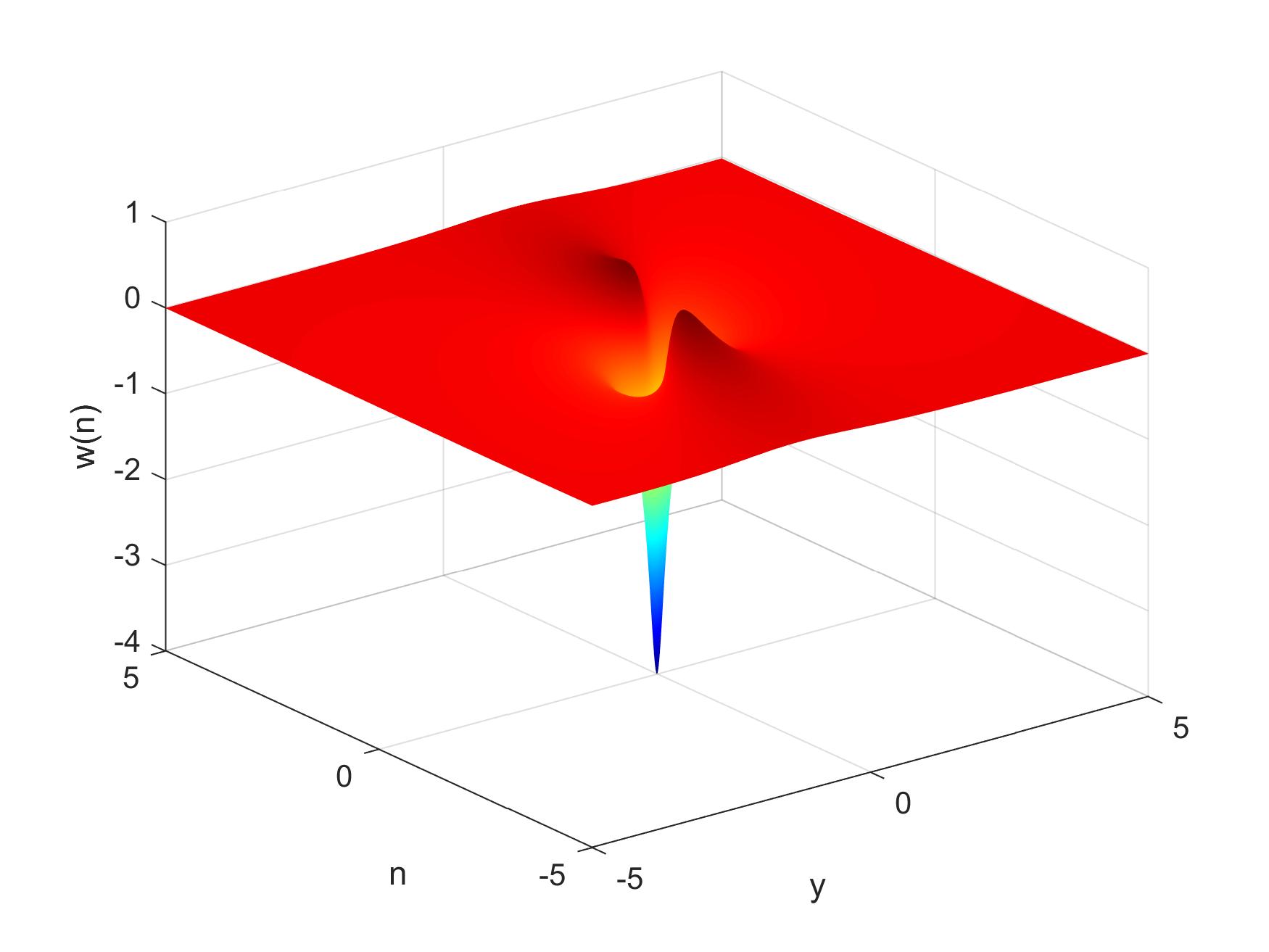} &
\includegraphics[height=0.220\textwidth,angle=0]{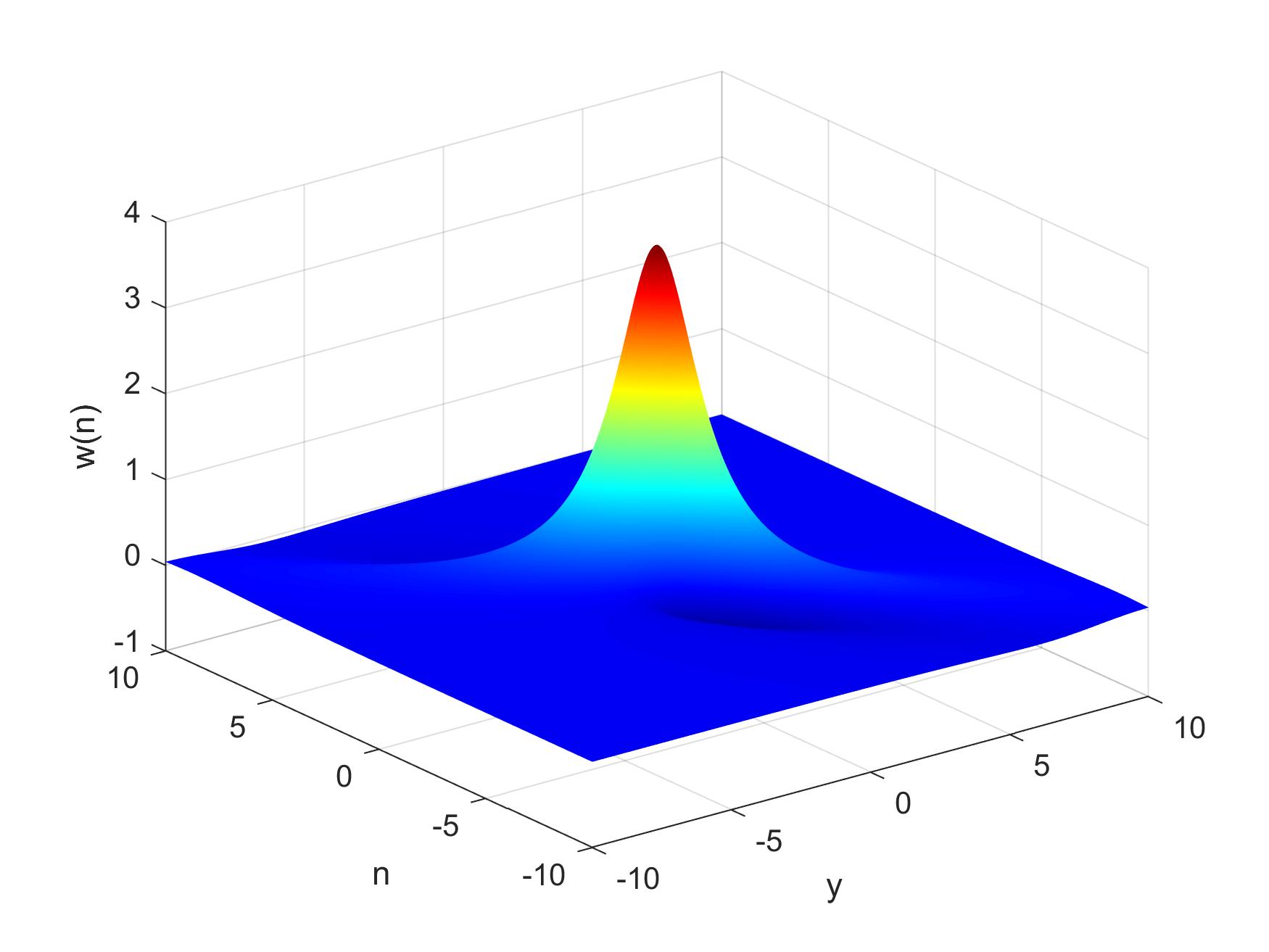} \\
(a) Dark lump $w(n)$  & \quad  (b) Bright lump $w(n)$ \\
\includegraphics[height=0.220\textwidth,angle=0]{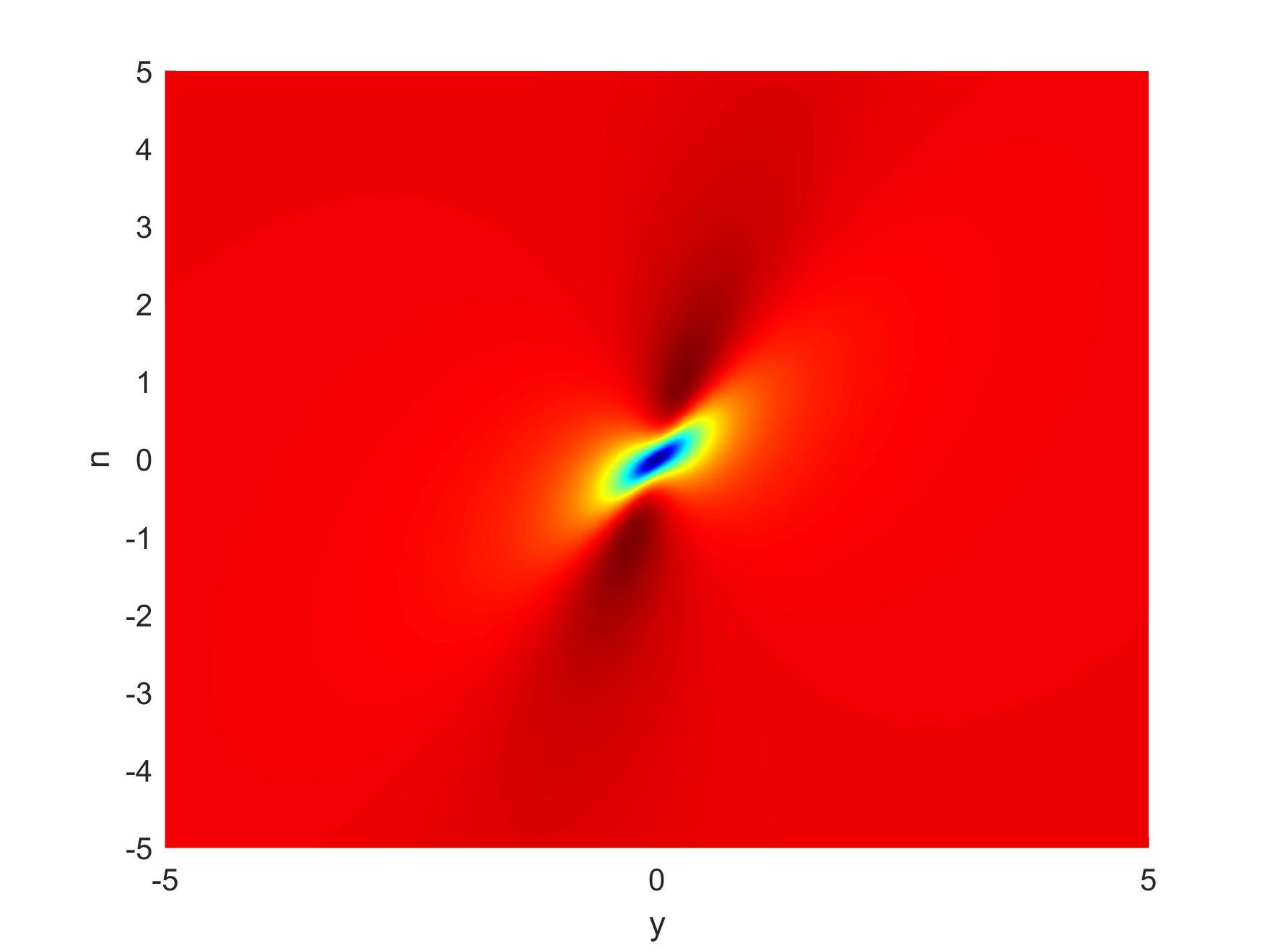} &
\includegraphics[height=0.220\textwidth,angle=0]{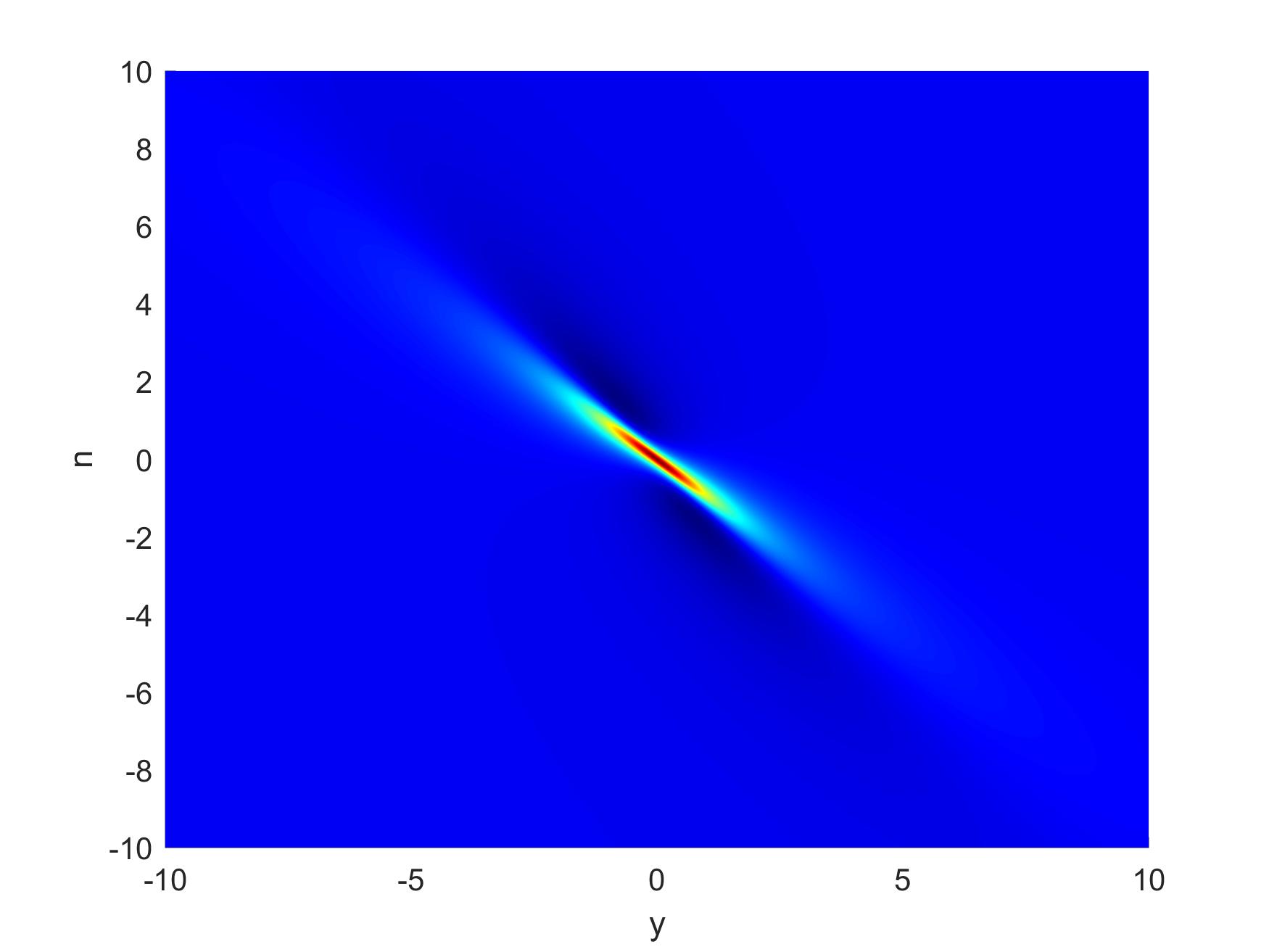} \\
(c) density plot of (a)  & \quad  (d) density plot of (b)
\end{tabular}
\end{center}
\caption{
One-lump solution $w(n)$ with the parameters:
(a) $a=1$, $b=1$; (b) $a=-1$, $b=1$.}\label{wlump}
\end{figure}

Furthermore, we can apply the nonlinear superposition formula (\ref{NSF}) repeatedly and derive $N-$order determinant solution  to the variant bilinear BS lattice equations,
\begin{align}
    \tau(n)=\left|(-\lambda_k)^{j-1}\tau_k(n-(j-1)\I)\right|_{1\leq j,k\leq N},
\end{align}
where $\tau_k(n)$ satisfy the equation (\ref{lumptau}). By varying the parameters, we can construct high-order lump solutions. For example, choosing the parameters as $N=4$, $\lambda_1=\lambda_3^*=1+\I$, $\lambda_2=\lambda_4^*=2+2\I$, and $\beta_j=\sum_{k\neq j}\frac{\I\lambda_j}{\lambda_j-\lambda_k}$, we arrive at the fundamental two-lump solution, as shown in Fig. \ref{fig4}; when taking parameters as $N=4$, $\lambda_1=\lambda_3^*=1+\I$, $\lambda_2=\lambda_4^*=1+3\I$, $\beta_j=\sum_{k\neq j}\frac{\I\lambda_j}{\lambda_j-\lambda_k}$, we obtain the fundamental-bright two-lump solution, as shown in Fig. \ref{fig4+}.
The expressions of the two-lump solutions are cumbersome and thus omitted here.
\begin{figure}
\begin{center}
\begin{tabular}{ccc}
\includegraphics[height=0.220\textwidth,angle=0]{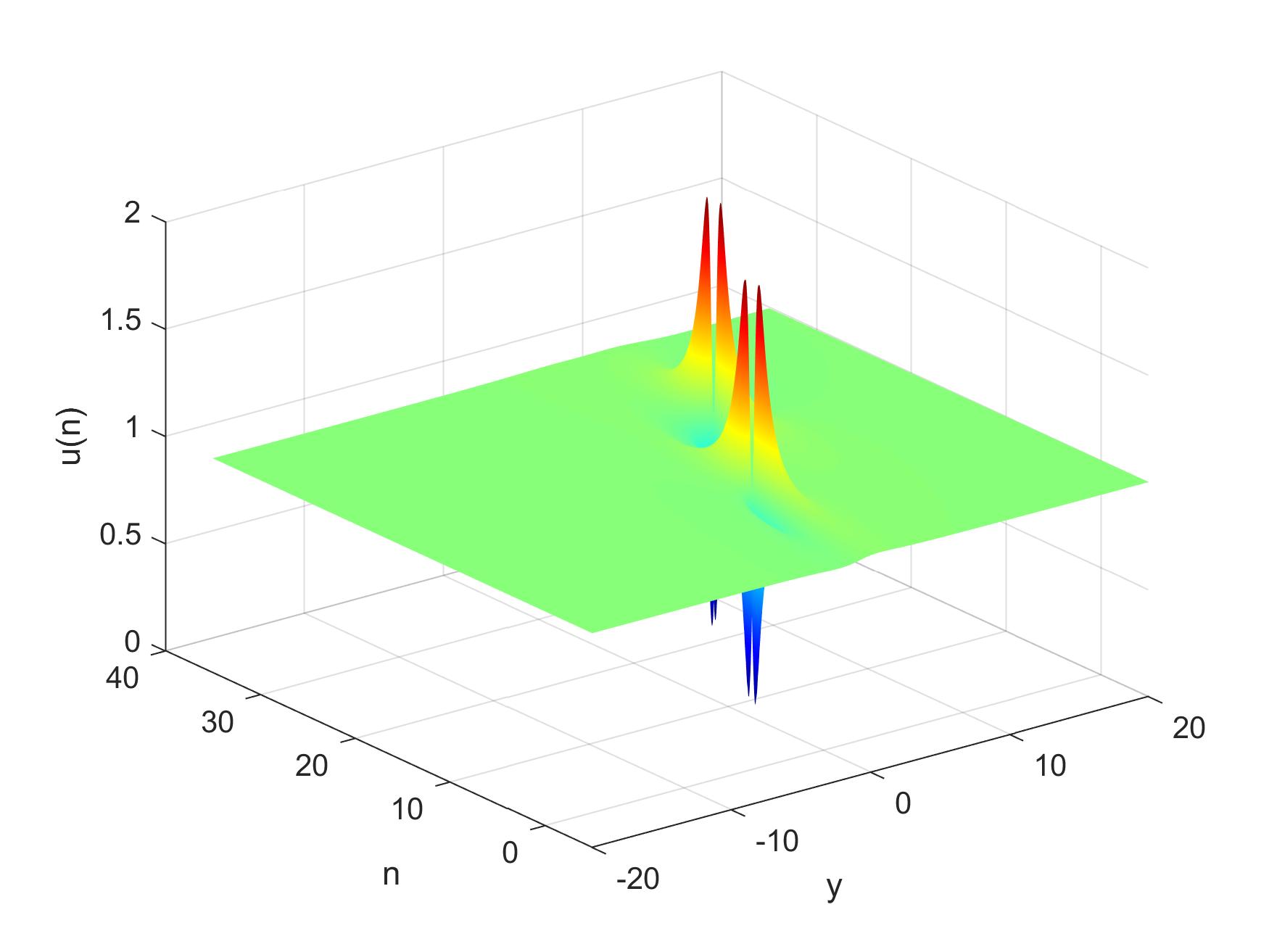} &
\includegraphics[height=0.220\textwidth,angle=0]{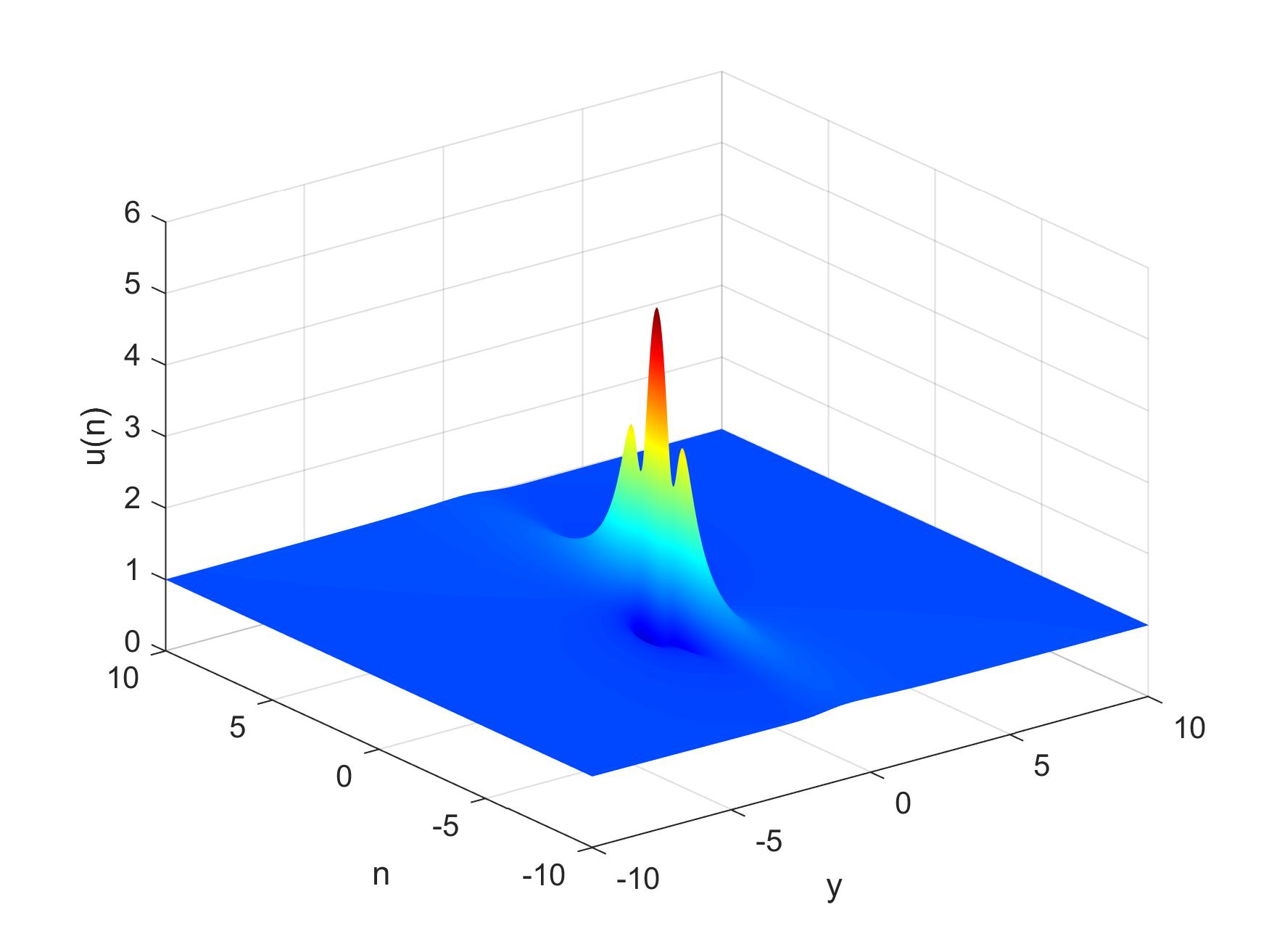} &
\includegraphics[height=0.220\textwidth,angle=0]{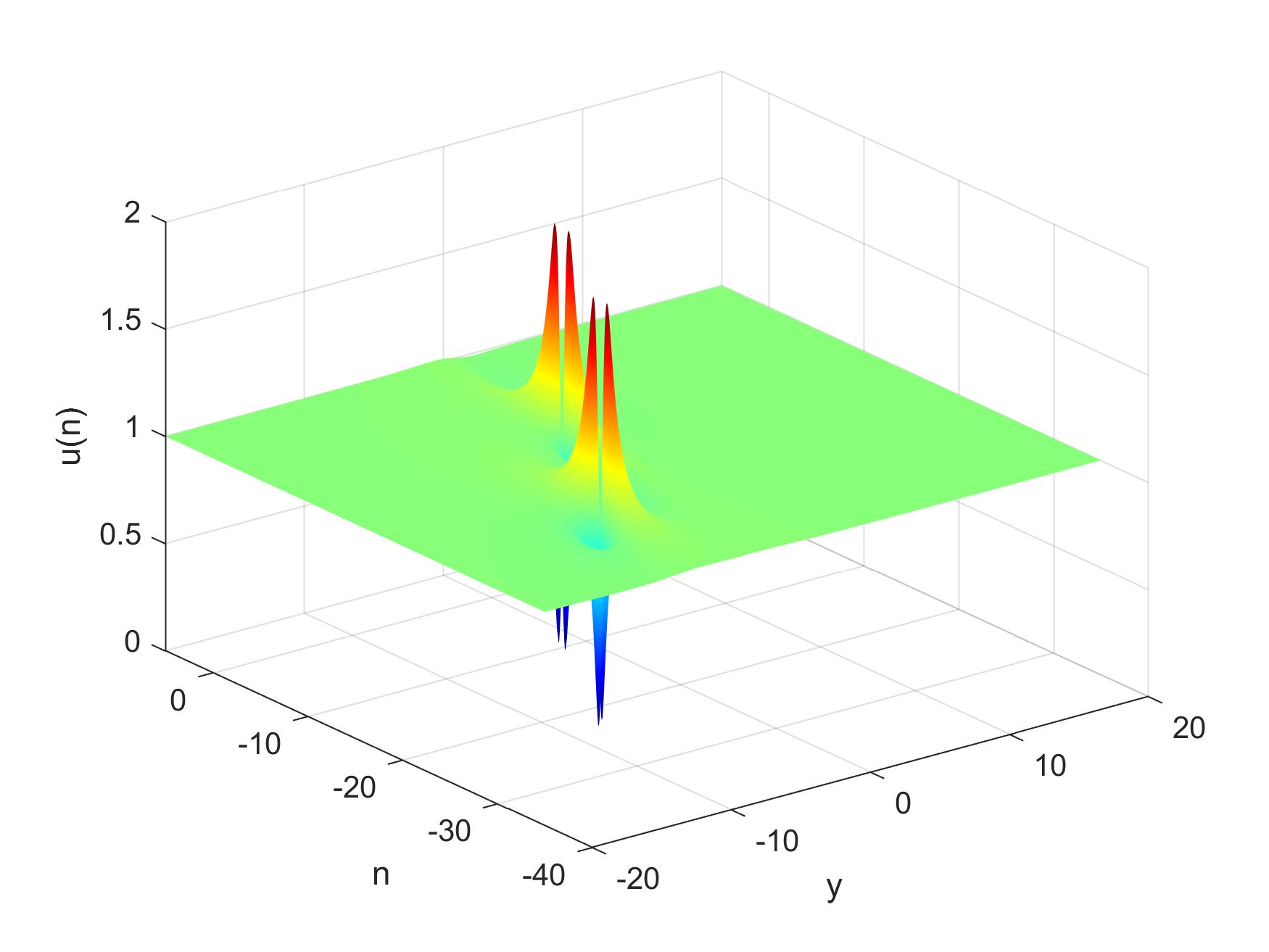} \\
(a) $t=-30$  & \quad  (b) $t=0$ & \quad (c) $t=30$ \\
\includegraphics[height=0.220\textwidth,angle=0]{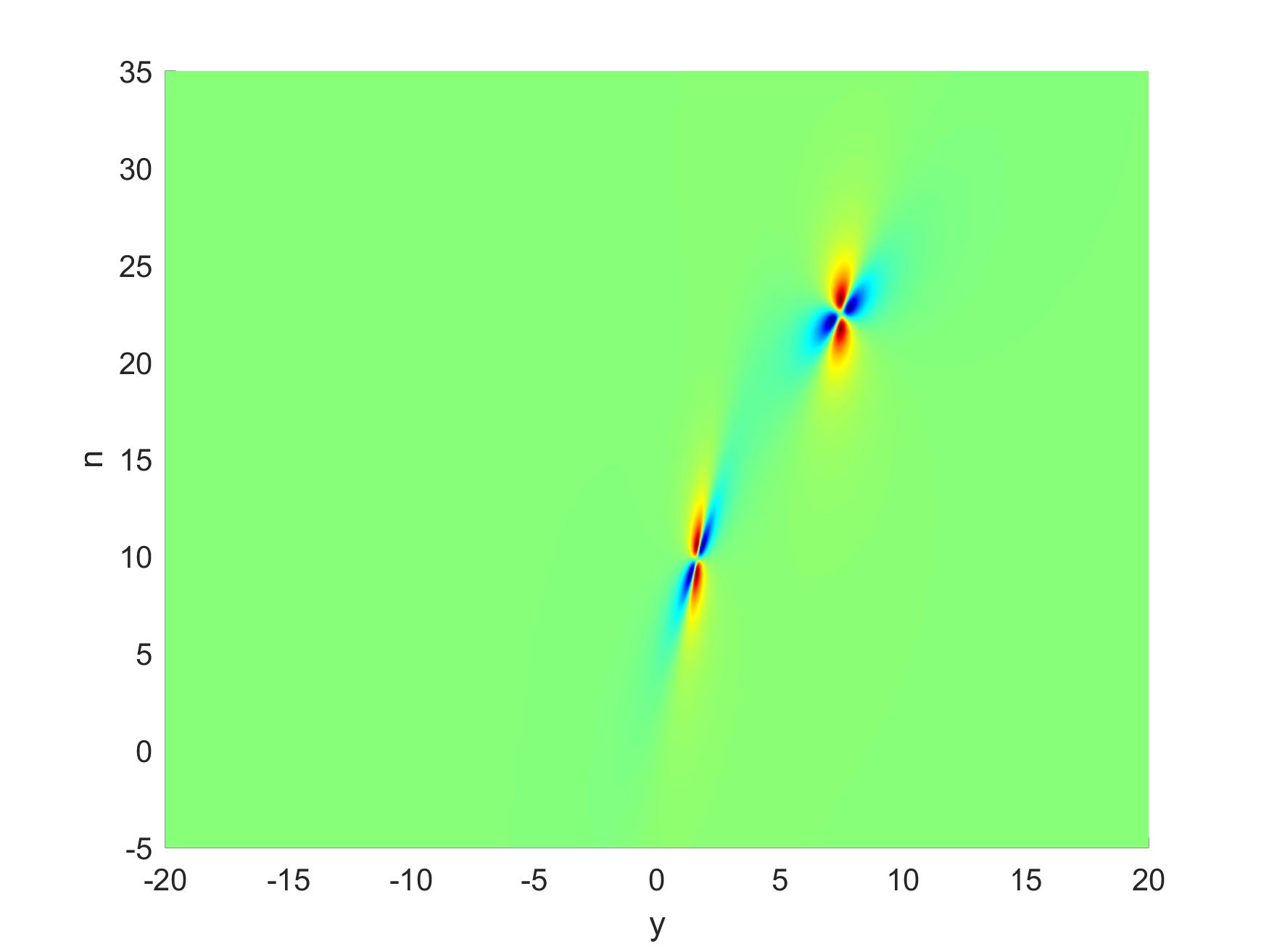} &
\includegraphics[height=0.220\textwidth,angle=0]{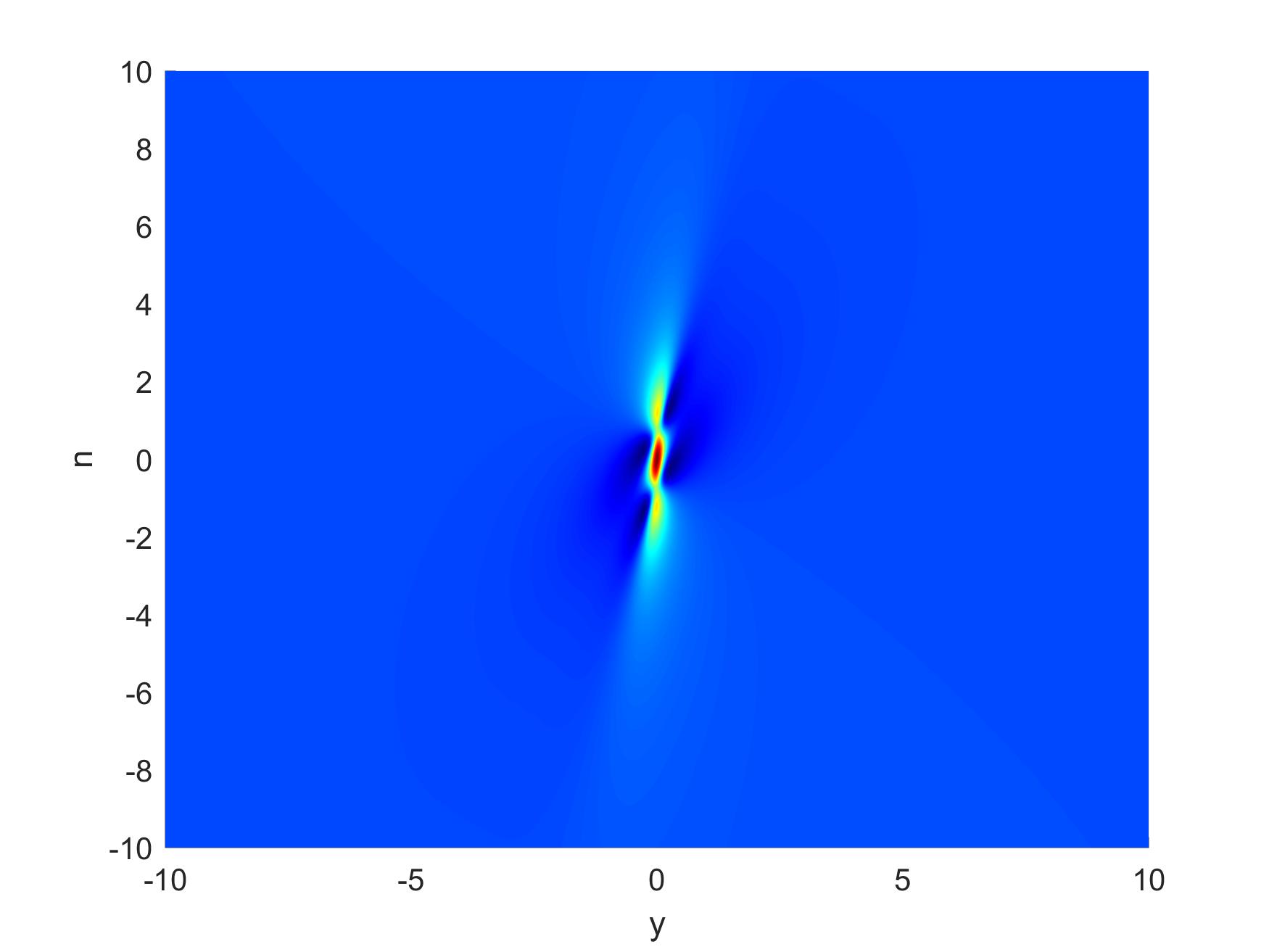} &
\includegraphics[height=0.220\textwidth,angle=0]{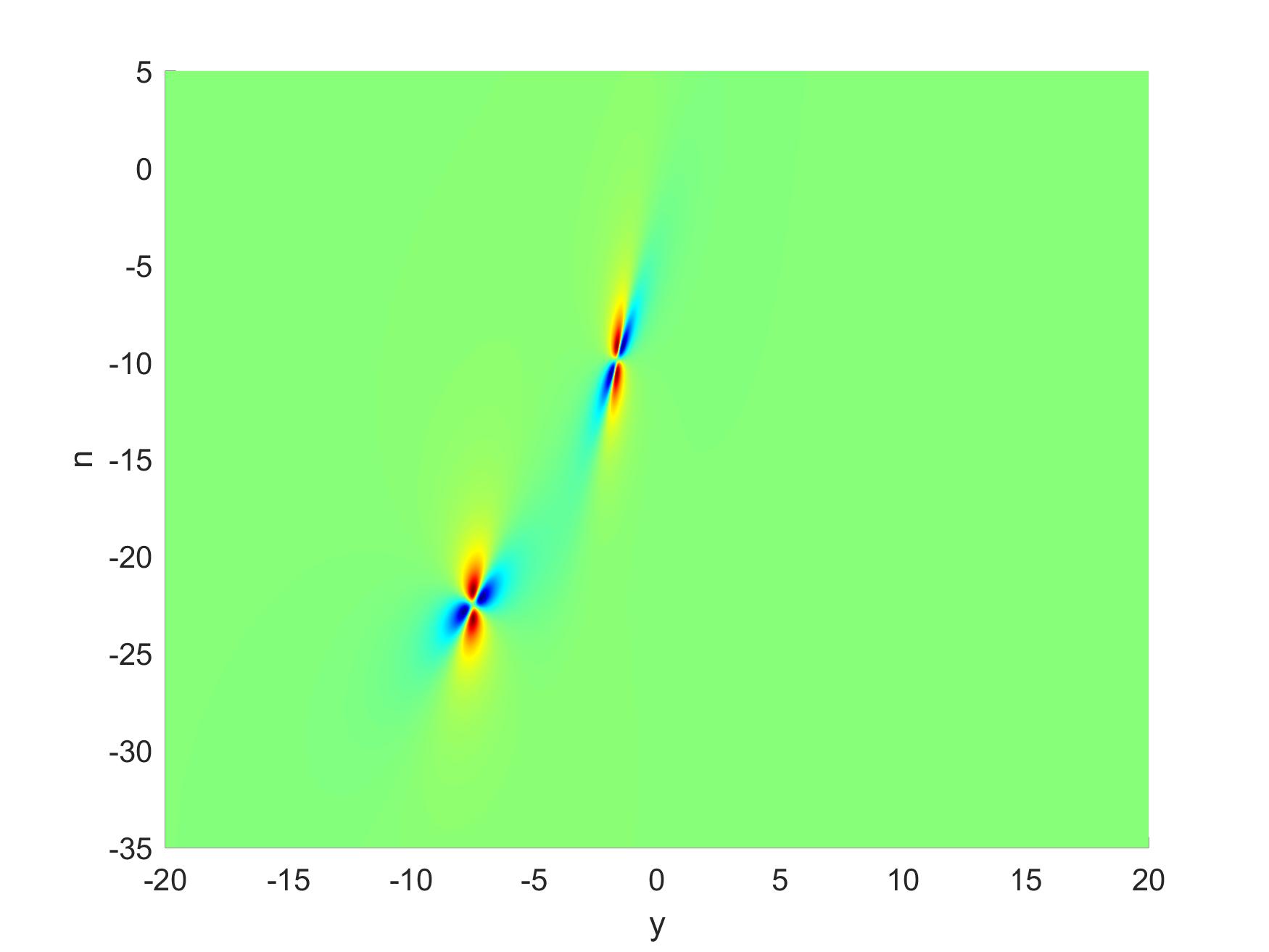} \\
(d) density plot of (a)  & \quad (e) density plot of (b) & \quad (f) density plot of (c)
\end{tabular}
\end{center}
\caption{
Fundamental two-lump solutions with the parameters:
$\lambda_1=1+\I$, $\lambda_2=2+2\I$.}\label{fig4}
\end{figure}
\begin{figure}
\begin{center}
\begin{tabular}{ccc}
\includegraphics[height=0.220\textwidth,angle=0]{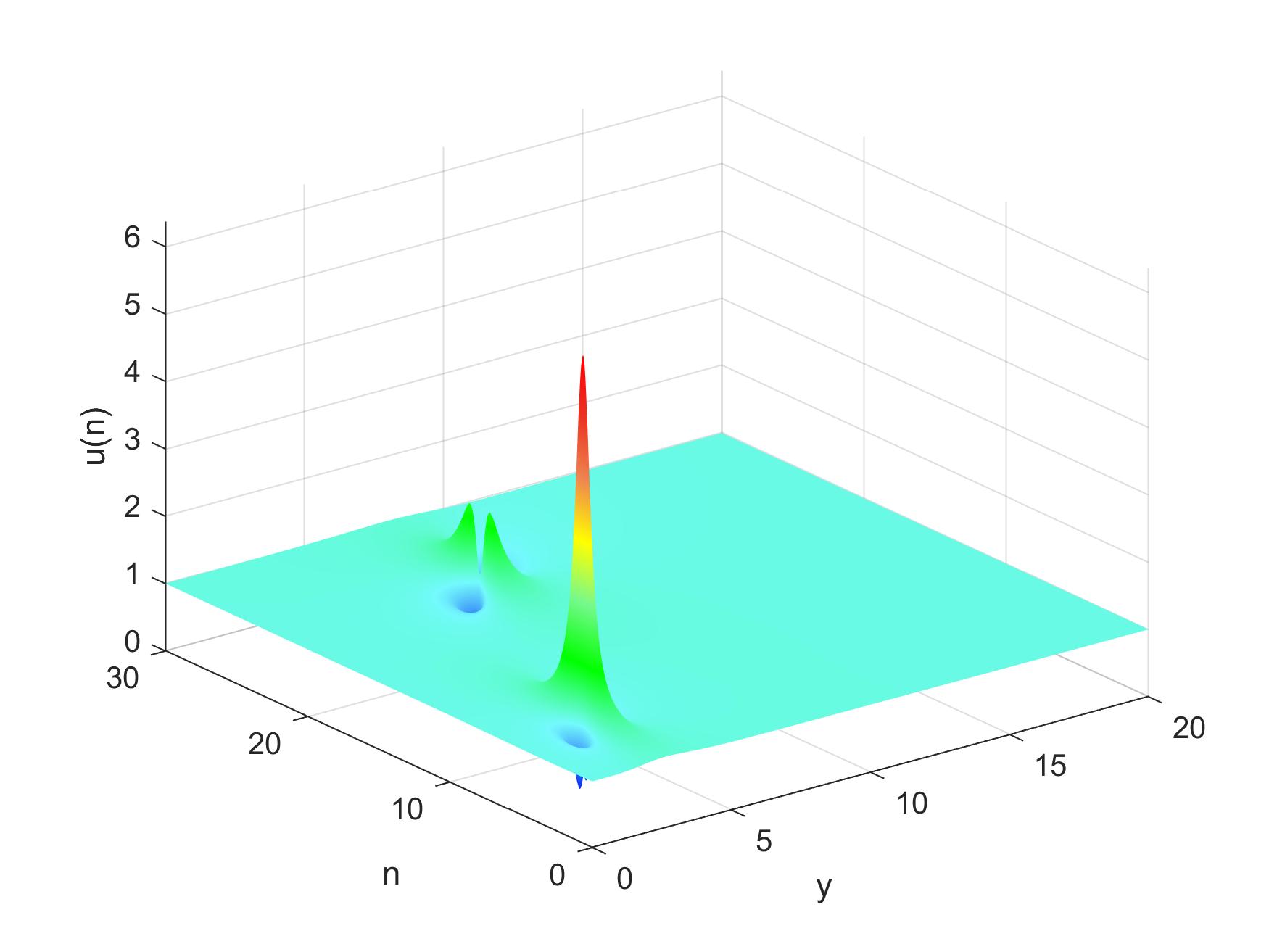} &
\includegraphics[height=0.220\textwidth,angle=0]{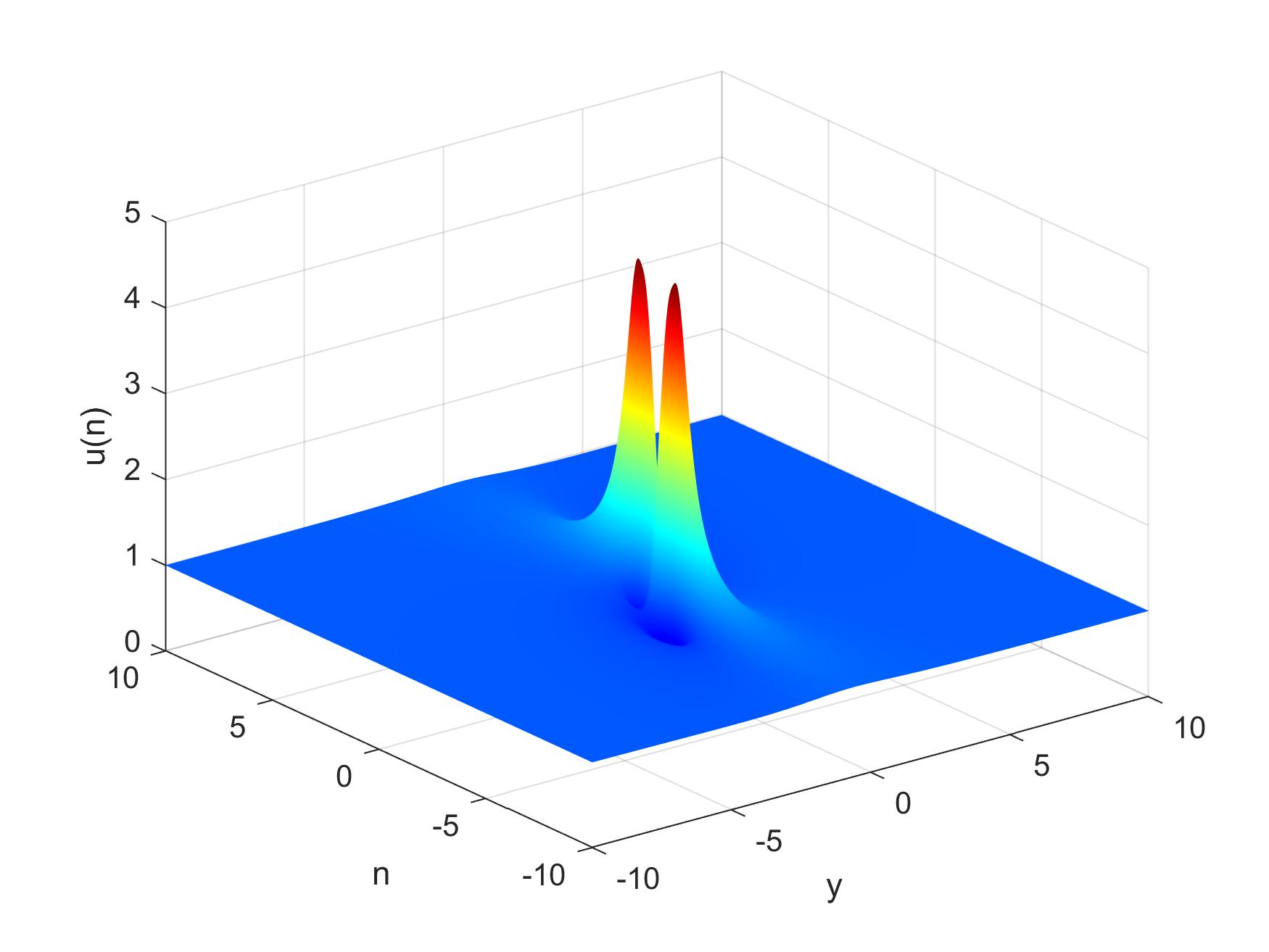} &
\includegraphics[height=0.220\textwidth,angle=0]{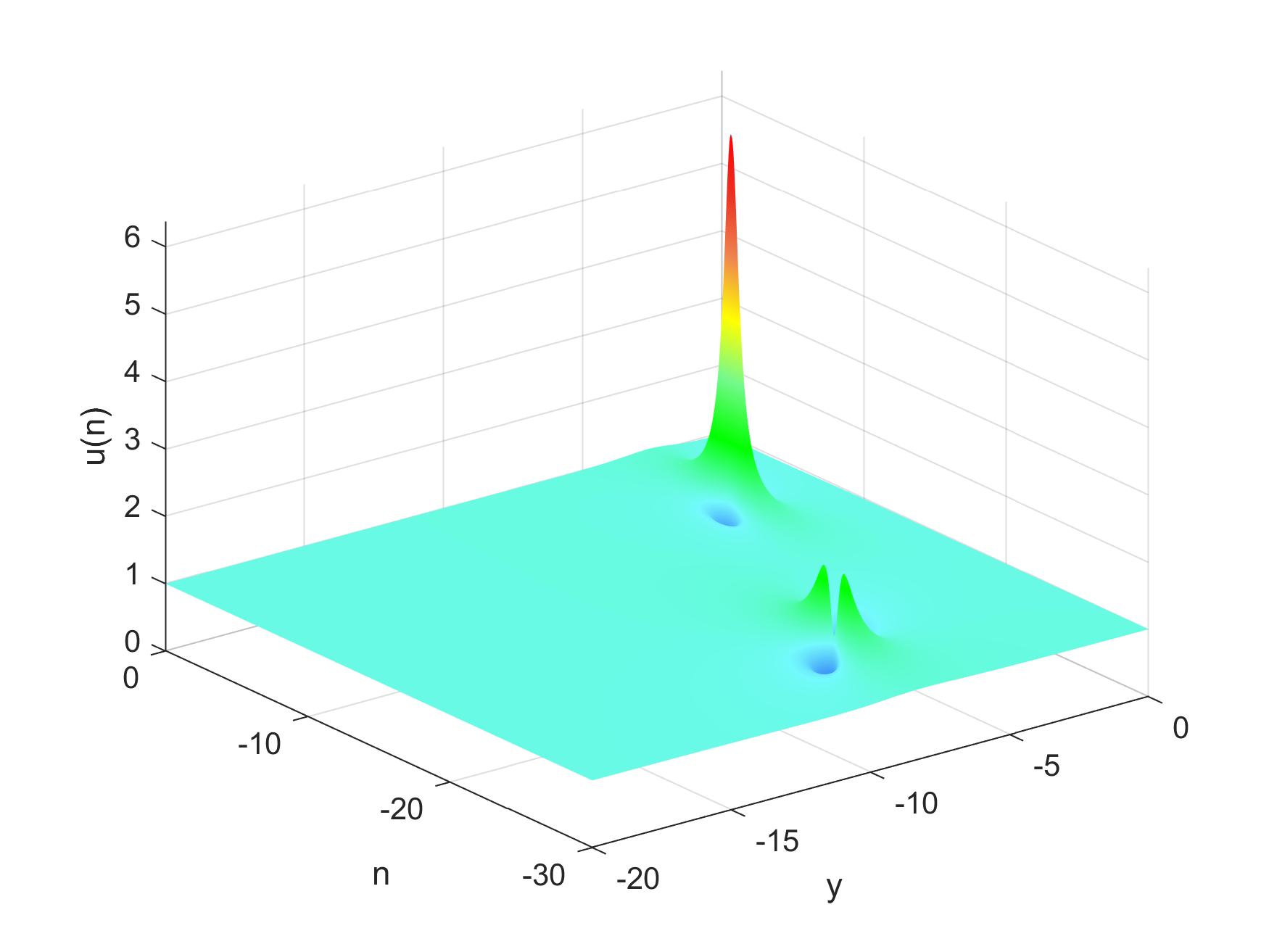} \\
(a) $t=-30$  & \quad  (b) $t=0$ & \quad (c) $t=30$ \\
\includegraphics[height=0.220\textwidth,angle=0]{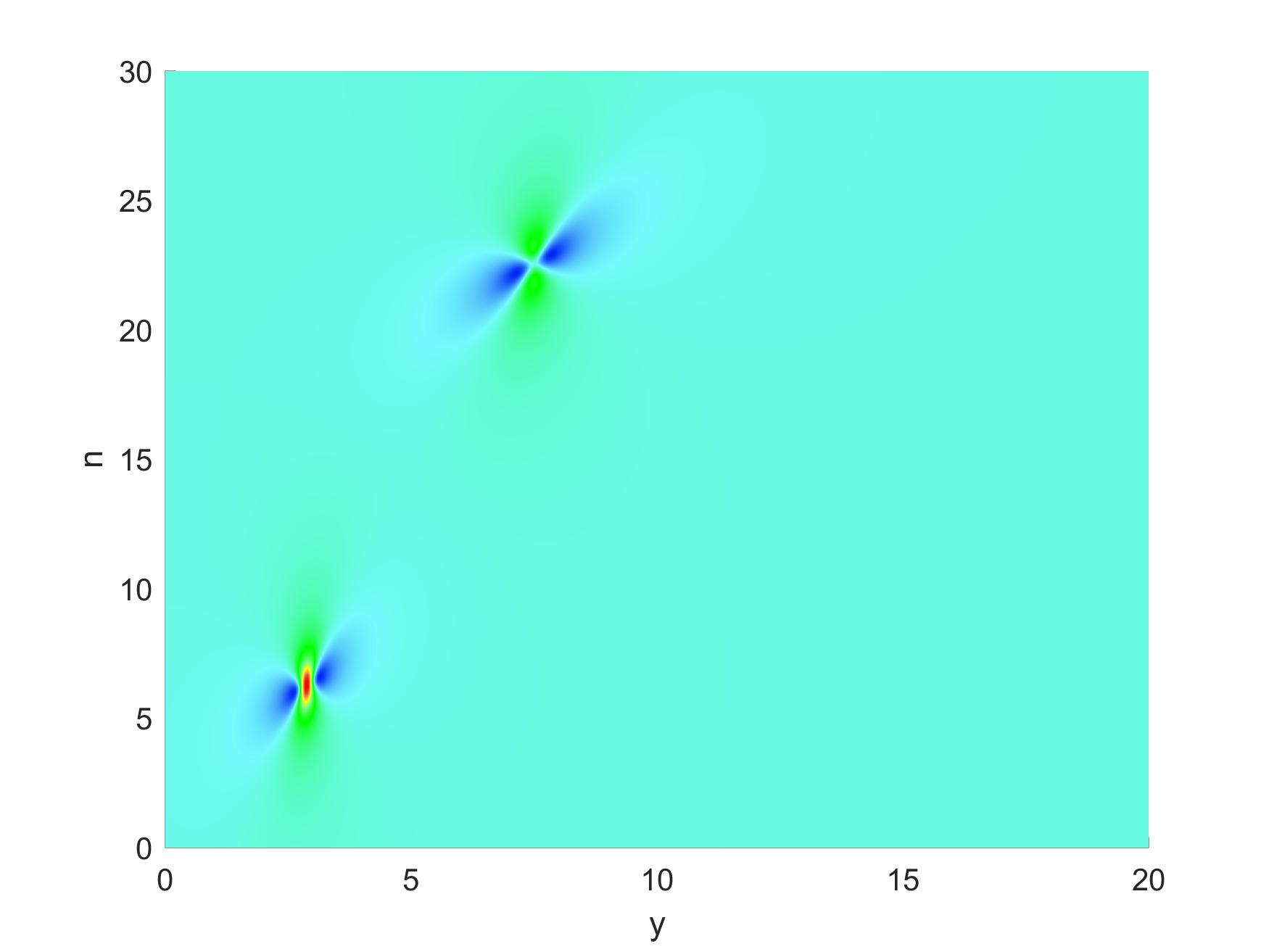} &
\includegraphics[height=0.220\textwidth,angle=0]{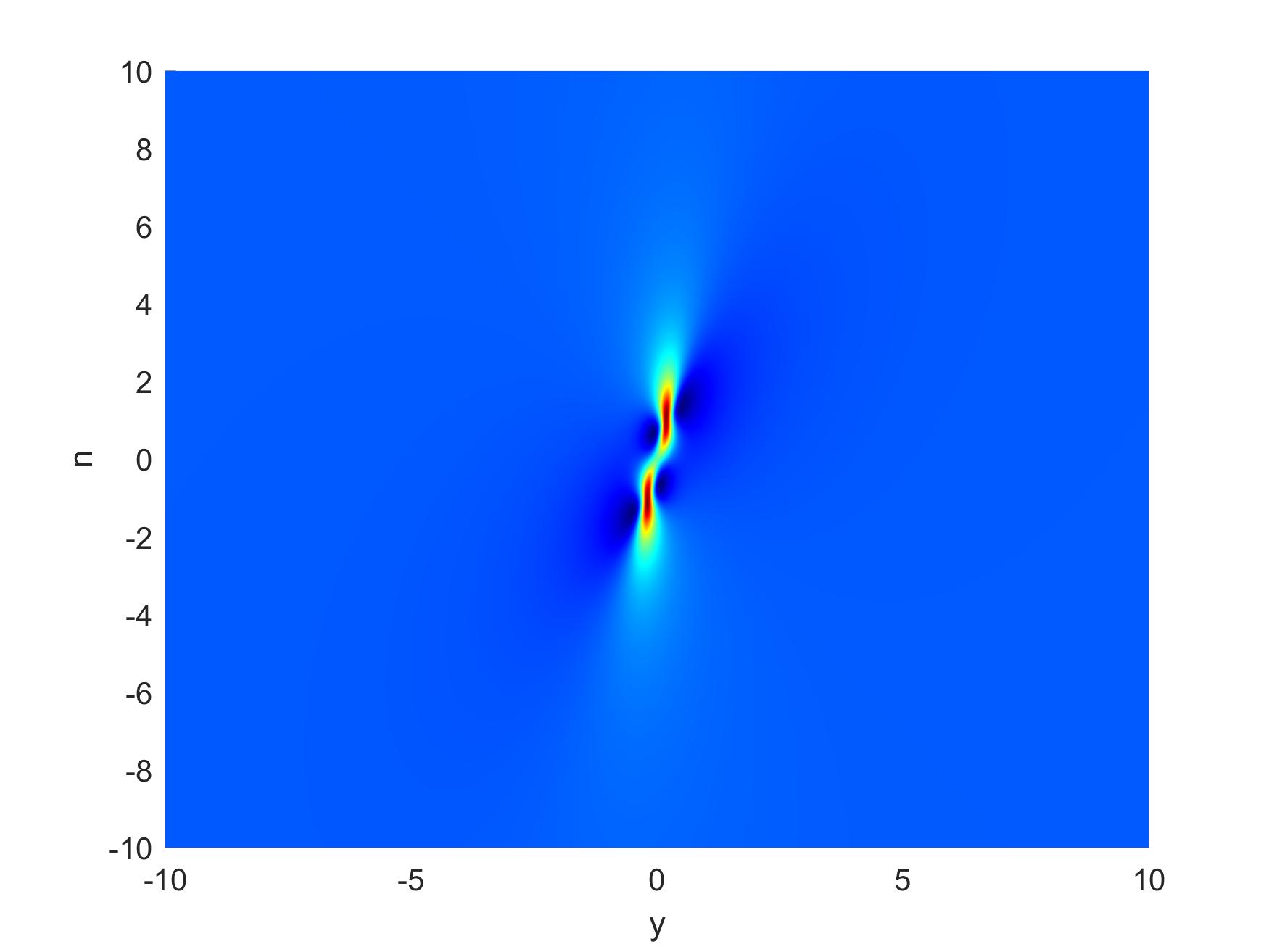} &
\includegraphics[height=0.220\textwidth,angle=0]{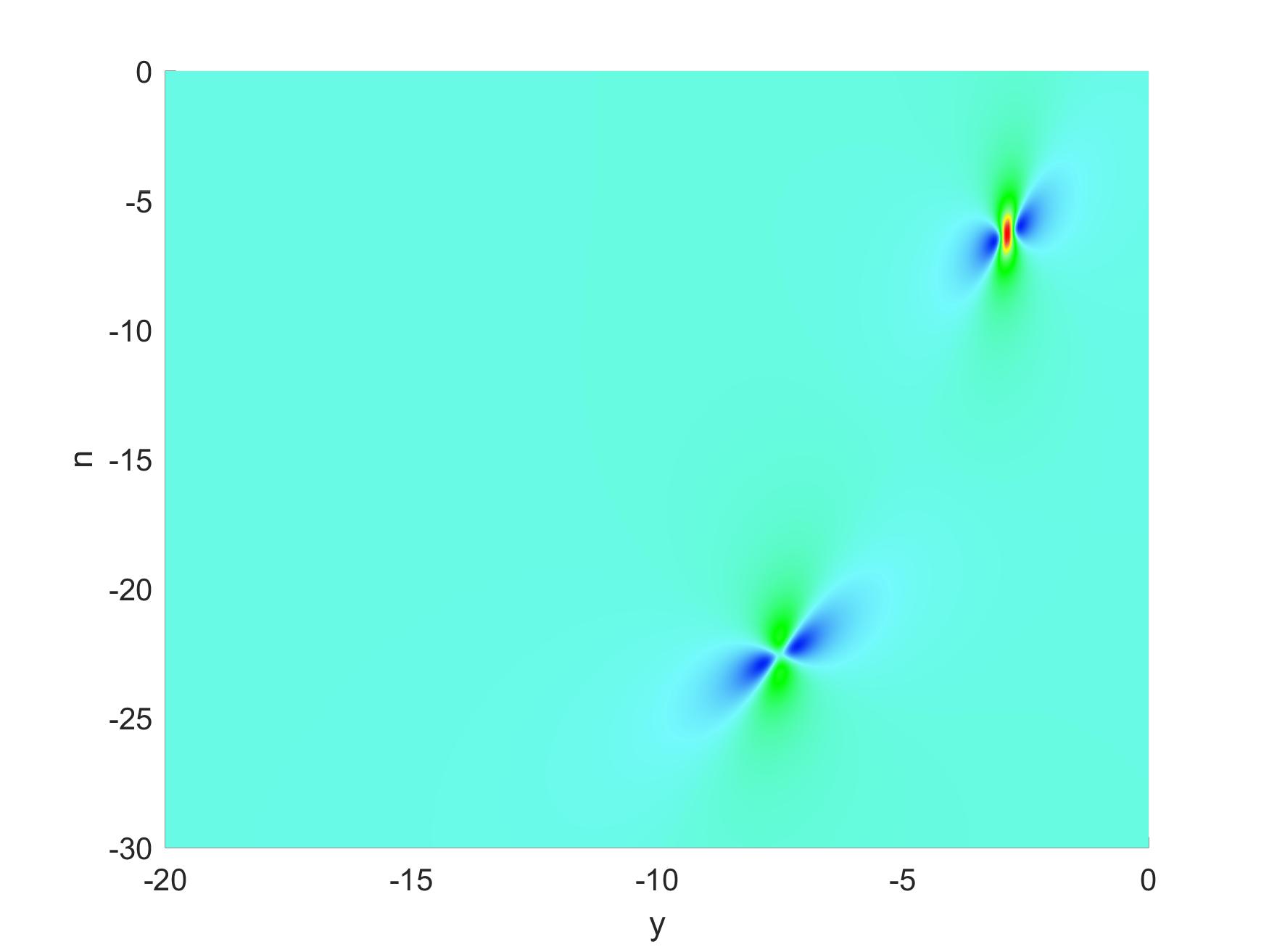} \\
(d) density plot of (a)  & \quad (e) density plot of (b) & \quad (f) density plot of (c)
\end{tabular}
\end{center}
\caption{
Fundamental-bright two-lump solutions with the parameters:
$\lambda_1=1+\I$, $\lambda_2=1+3\I$.}\label{fig4+}
\end{figure}

\subsection{Lump solutions in terms of the Schur polynomials}
In this subsection, we investigate lump solutions by the Schur polynomials. We define
\begin{align}
    \phi=e^{\eta},\quad \psi=e^\xi,\quad m=\frac{\I}{e^{p}-e^{q}}\phi\psi,
\end{align}
where
\begin{align}
    &\eta=-\I pn-\I e^{p}t-\I(e^{-p}+e^{2p})y-\I e^{2p}z,\\
    &\xi=\I qn+\I e^{q}t+\I(e^{-q}+e^{2q})y+\I e^{2q}z.
\end{align}
To construct lump solutions, we introduce two differential operators $A_j$ and $B_k$ as
\begin{align}
    A_j=\frac{1}{j!}\partial_p^j,\quad B_k=\frac{1}{k!}\partial_q^k.
\end{align}
We apply the operators  $A_j, B_k$ to $\phi$ and $\psi$, respectively, and denote
\begin{align}
    \phi_j=A_j\phi,\quad \psi_k=B_k\psi,\quad m_{jk}=A_jB_km.
\end{align}
We introduce the polynomials $P_j$ and $Q_k$ as
\begin{align}
    \frac{1}{j!}\partial_p^je^\eta=P_je^\eta,\quad \frac{1}{k!}\partial_q^k e^\xi=Q_ke^\xi,
\end{align}
where
\begin{align}
    &P_j=S_j(\pmb{\mu}(p)),\quad Q_k=S_k(\pmb{\nu}(q)),\\
    &\pmb{\mu}(p)=(\mu_1(p),\ \mu_2(p),\cdots,\ \mu_n(p),\cdots),\\
    &\pmb{\nu}(q)=(\nu_1(q),\ \nu_2(q),\cdots,\ \nu_n(q),\cdots),\\
    &\mu_1(p)=-\I n-\I e^{p}t+\I(e^{-p}-2e^{2 p})y-2\I e^{2p}z,\quad \mu_j=\frac{1}{j!}\partial_p^{j-1}\mu_1,\\
    &\nu_1(q)=\I n+\I e^{q}t-\I(e^{-q}-2e^{2q})y+2\I e^{2q}z,\quad \nu_j=\frac{1}{j!}\partial_q^{j-1}\nu_1,
\end{align}
and $S_j(\mathbf{x})$ are Schur polynomials with $\mathbf{x}=(x_1, x_2,\ \cdots)$.
\begin{rmk}
    Schur polynomials $S_j(\mathbf{x})$ are defined by
    \begin{align}
        \exp(\sum_{j=1}^\infty x_j\lambda^j)=\sum_{k=0}^{\infty}S_k(\mathbf{x})\lambda^k,
    \end{align}
    where $\mathbf{x}=(x_1, x_2,\ \cdots)$. For instance, the first ones are
    \begin{align}
        S_0(\mathbf{x})=1,\ S_1(\mathbf{x})=x_1,\ S_2(\mathbf{x})=\frac{1}{2}x_1^2+x_2,\ S_k(\mathbf{x})=\sum_{l_1+2l_2+\cdots+ml_m=k}\left( \prod_{j=1}^{m} \frac{x_j^{l_j}}{l_j!} \right).
    \end{align}
\end{rmk}

\begin{prop}
The variant BS lattice equation \eqref{variants-1}-\eqref{variants-3} admits the lump solutions $u(n)=\frac{\tau(n+\I)\tau(n-\I)}{\tau^2(n)}$ with
\begin{align}
    \tau(n)=\det_{1\leq i,j\leq N}(m_{ij}),
\end{align}
and
\begin{align}
    m_{ij}=\sum_{\sigma=0}^{k_i+k_j}\frac{(-1)^\sigma}{(-\I e^{p}+\I e^{q})^{\sigma+1}}\partial_t^\sigma\left(P_{k_i}Q_{k_j}\right).
\end{align}
\end{prop}
Choosing $N=1,k_1=1$, $p=q^*=\frac{\pi}{2}\I$ and $z=0$, we have
\begin{align}
    \tau(n)=\frac{1}{2}\left((n-2y)^2+\left(t+y-\frac{1}{2}\right)^2+\frac{1}{4}\right).
\end{align}
The lump solution $u(n)$ is expressed as
\begin{align}\label{onelump}
    u(n)=1+\frac{2(n-2y)^2-2\left(t+y-\frac{1}{2}\right)^2+\frac{1}{2}}{\left((n-2y)^2+\left(t+y-\frac{1}{2}\right)^2+\frac{1}{4}\right)^2},
\end{align}
which is illustrated in Fig. \ref{fig5}. Furthermore, the expression (\ref{onelump}) indicates that the peak location moves along $n=2y$. Compared with the one-lump (\ref{u-1lump}) obtained by BT, (\ref{onelump}) coincides with (\ref{u-1lump}) if we choose the parameters as $a=0$, $b=1$ with $t\rightarrow t-\frac{1}{2}$ in (\ref{u-1lump}). 
\begin{figure}
\begin{center}
\begin{tabular}{ccc}
\includegraphics[height=0.220\textwidth,angle=0]{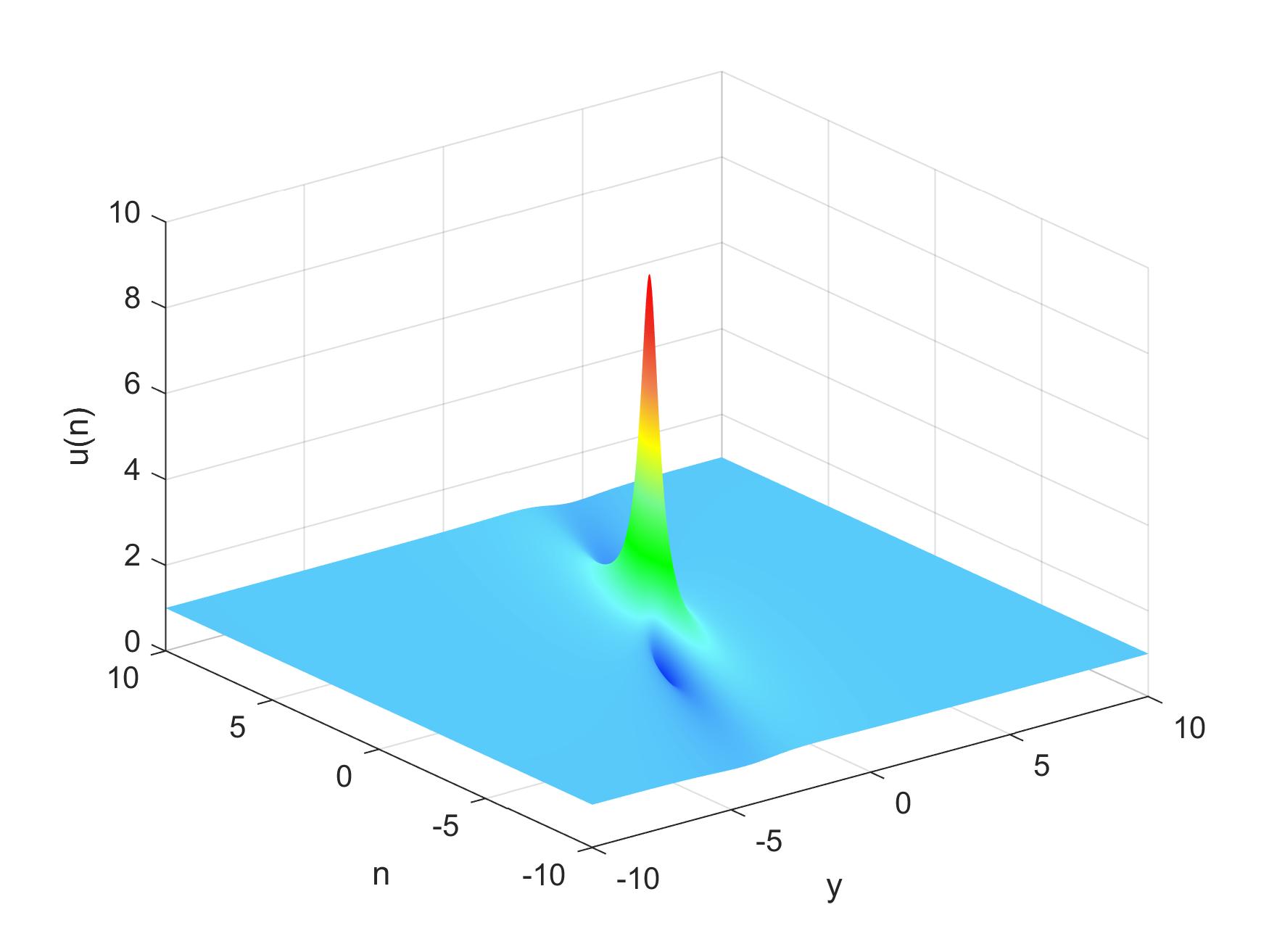} &
\includegraphics[height=0.220\textwidth,angle=0]{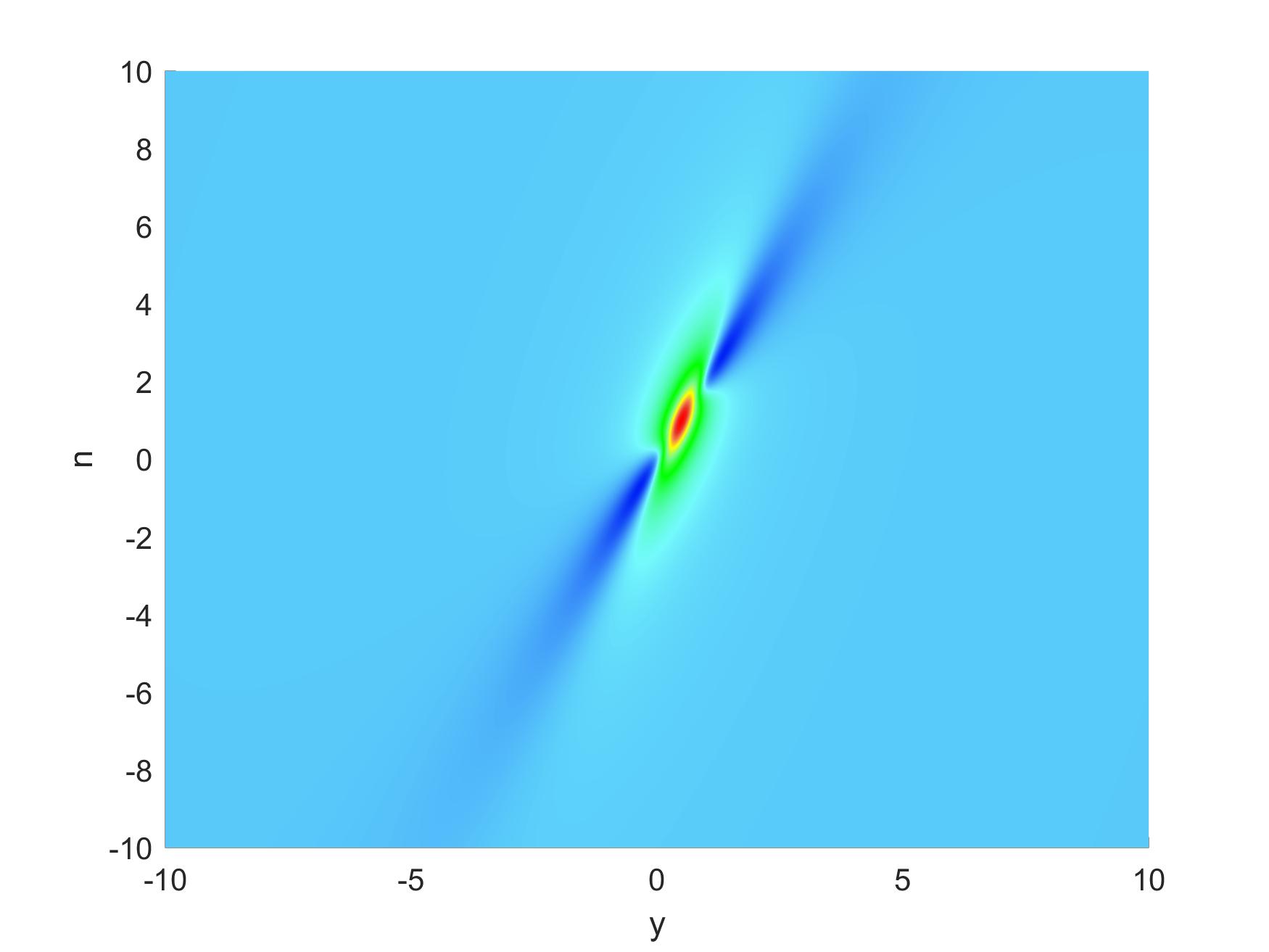} \\
(a) $u(n)$ with $t=0$  & \quad  (b) density plot of (a)
\end{tabular}
\end{center}
\caption{
Lump solution with parameters:
$p=q^*=\frac{\pi}{2}\I$.}\label{fig5}
\end{figure}

With the parameters assigned as  $N=1$, $k_1=2$ and $p=q^*=\frac{1}{2}\ln 2+\I\frac{3\pi}{4}$,  the two-lump solution is derived,
\begin{equation}
    \begin{split}
        \tau(n)=&\frac{1}{128}\left((2t-7y)^2-(2n-2t+y+1)^2-20y-1\right)^2+\frac{1}{32}\left(8y+1-(2n-2t+y)(2t-7y-1)\right)^2\\
        &+\frac{1}{16}(2n-2t+y+2)^2+\frac{1}{16}(2t-7y-1)^2+\frac{1}{8}.
    \end{split}
\end{equation}
The plot of the two-lump solution $u(n)$ is displayed in Fig. \ref{diff2lump}.

\begin{rmk}\label{remark2}
When compared with the BT in Subsection \ref{sec4.1}, it is evident that constructing lump solutions using Schur polynomials necessitates only two parameters $(p$  and $q$), which are responsible for defining the shape of the lump solution. This restriction means that the Schur polynomials method lacks the ability to yield interaction solutions—including the fundamental-bright lump solutions shown in Fig. \ref{fig4+}—a result that the BT can readily achieve.
\end{rmk}

\begin{figure}
\begin{center}
\begin{tabular}{ccc}
\includegraphics[height=0.220\textwidth,angle=0]{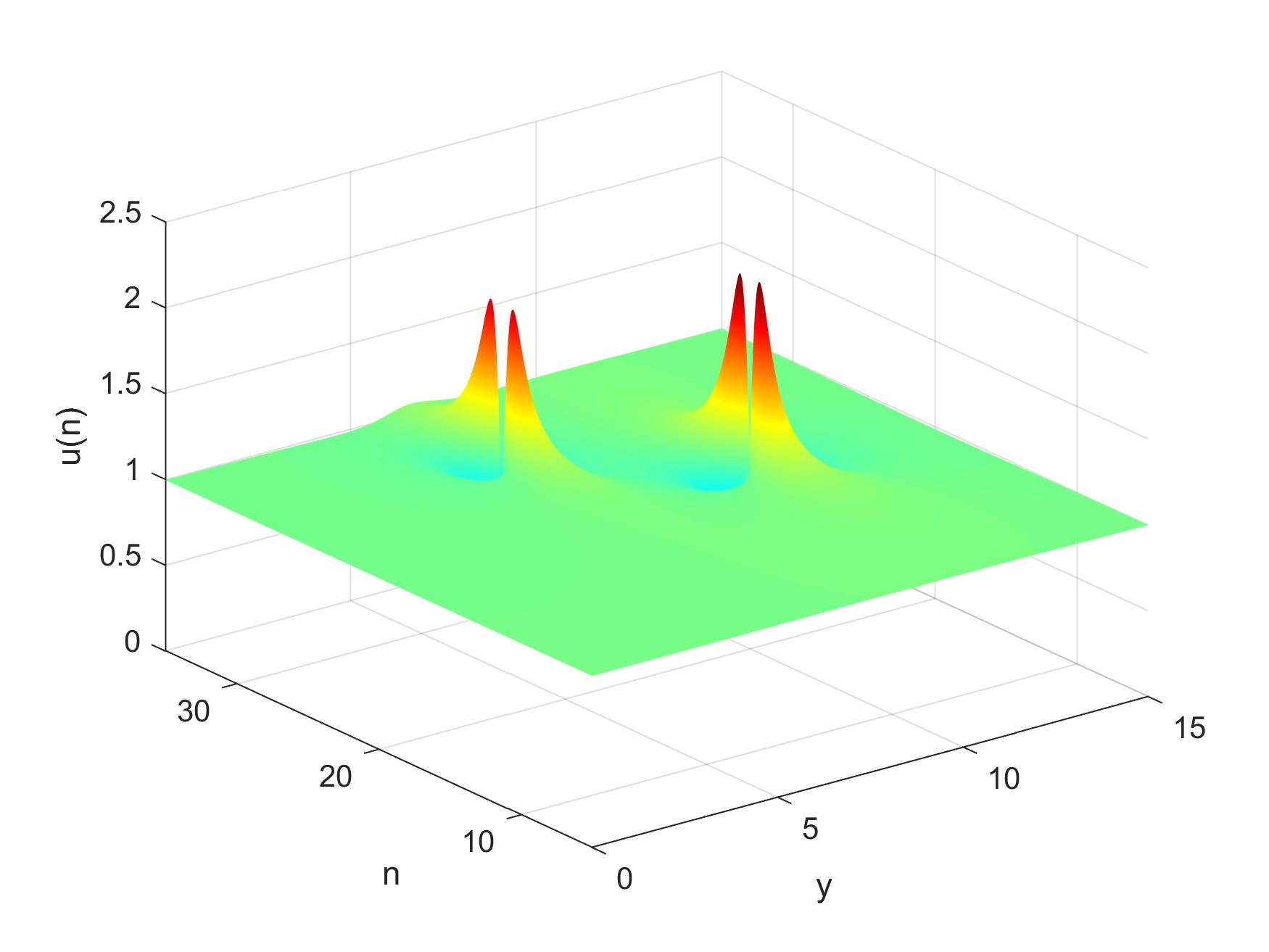} &
\includegraphics[height=0.220\textwidth,angle=0]{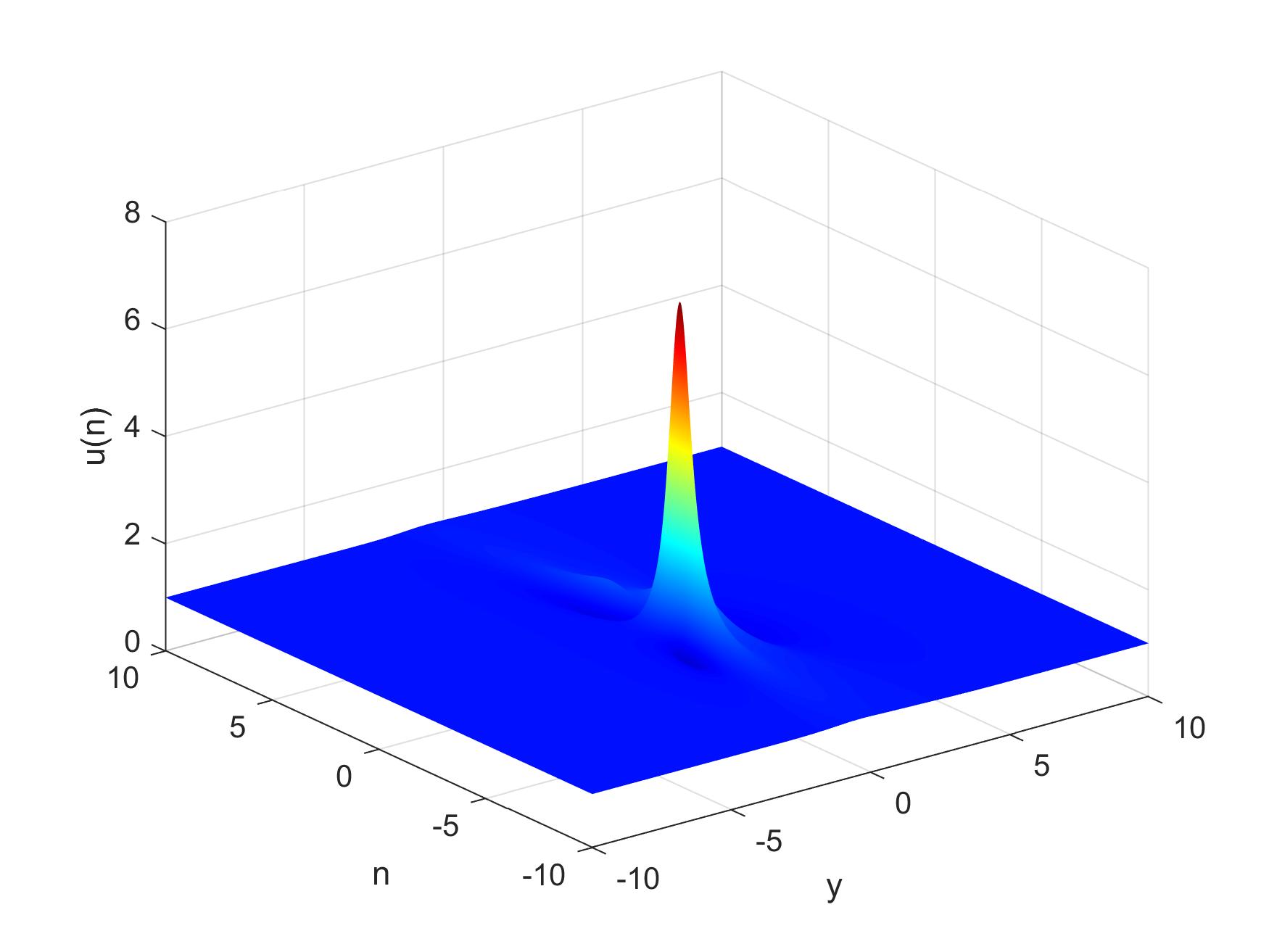} &
\includegraphics[height=0.220\textwidth,angle=0]{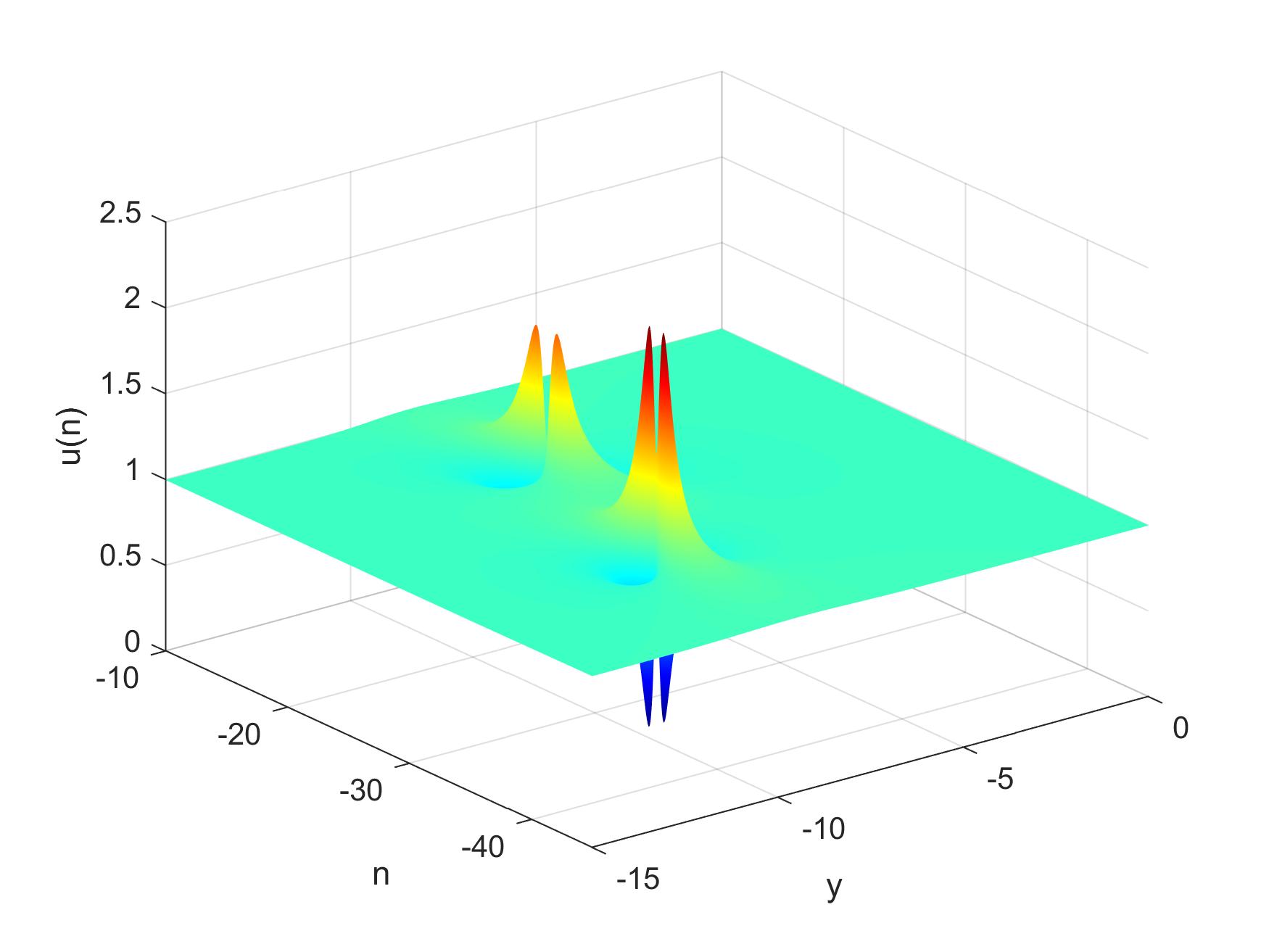} \\
(a) $t=30$  & \quad  (b) $t=0$ & \quad (c) $t=-30$ \\
\includegraphics[height=0.220\textwidth,angle=0]{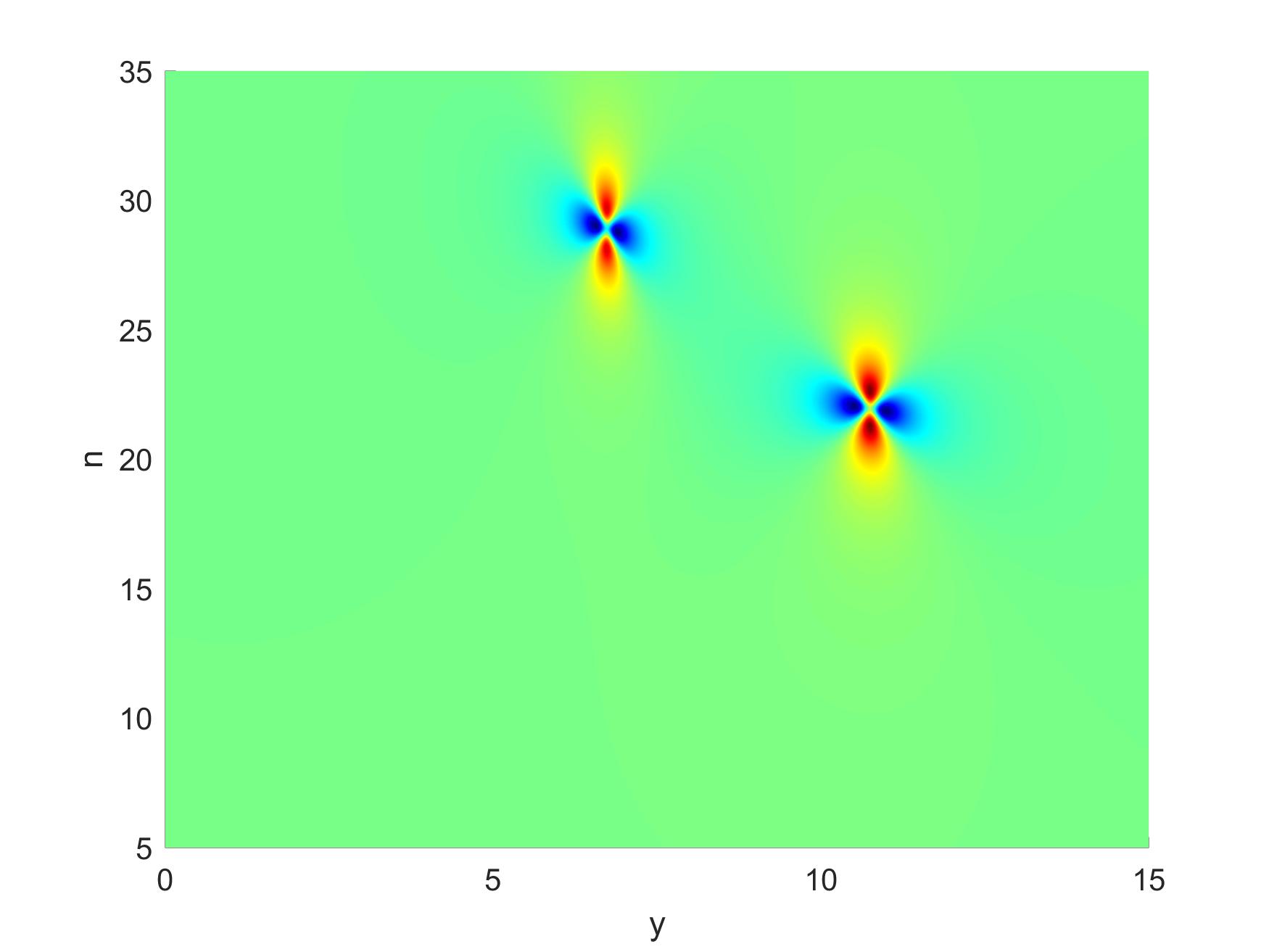} &
\includegraphics[height=0.220\textwidth,angle=0]{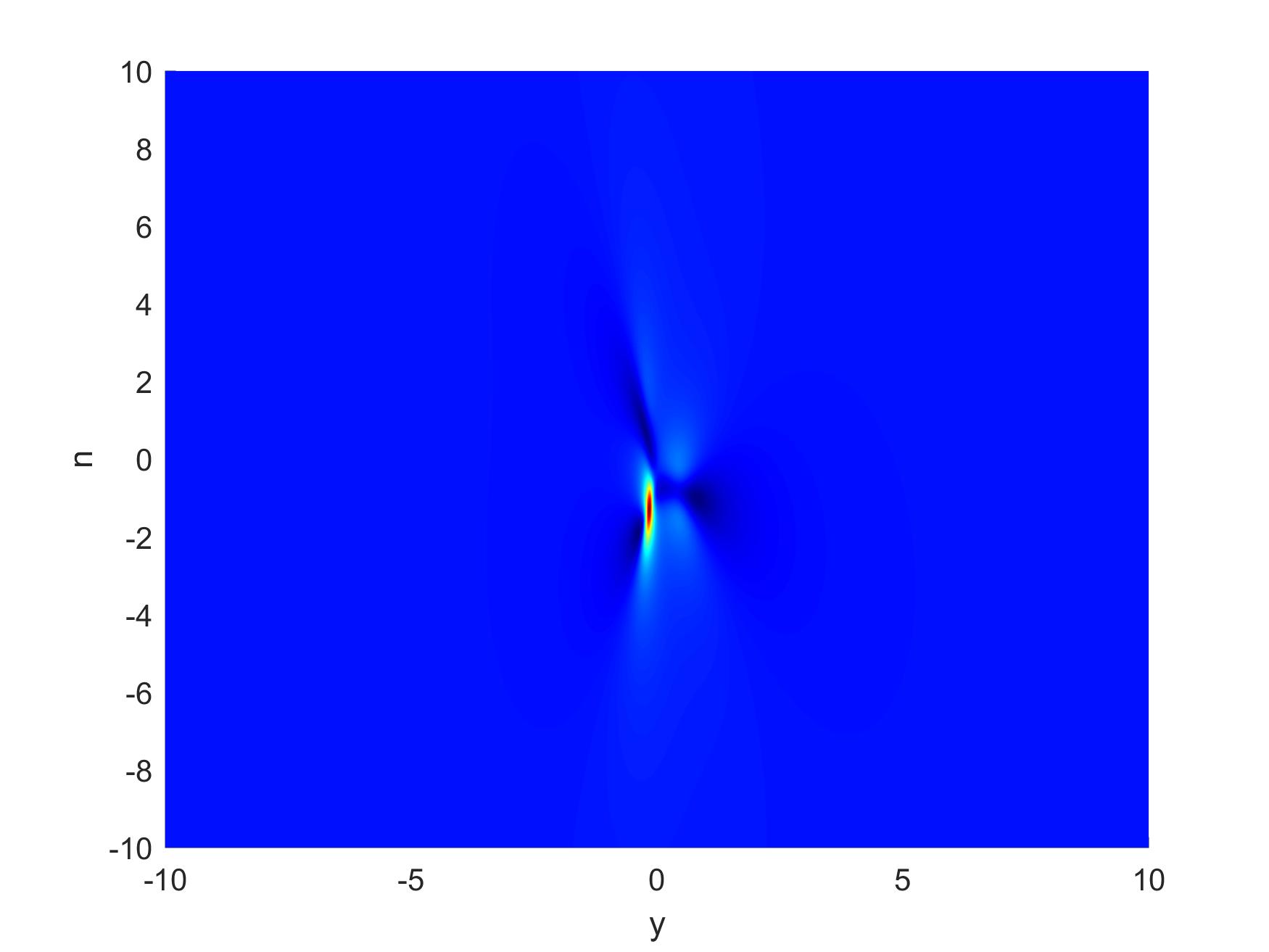} &
\includegraphics[height=0.220\textwidth,angle=0]{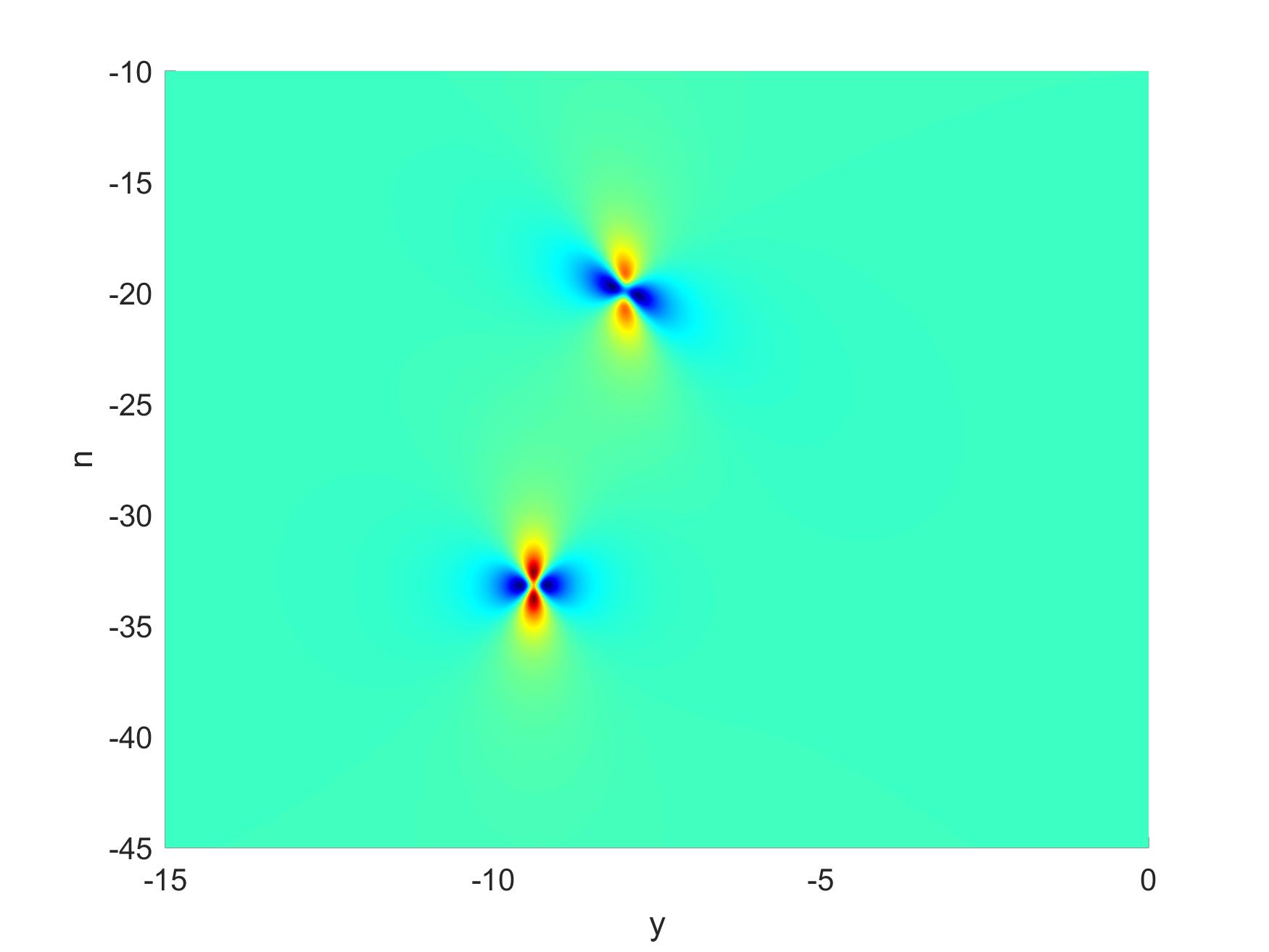} \\
(d) density plot of (a)  & \quad (e) density plot of (b) & \quad (f) density plot of (c)
\end{tabular}
\end{center}
\caption{
Two-lump solutions with parameters
$N=1$, $k_1=2$, $p=\frac{1}{2}\ln 2+\I\frac{3\pi}{4}$.}\label{diff2lump}
\end{figure}

In order to construct higher-order lump solutions, we set $e^{p}=a+\I b$. Then we have
\begin{align}
    m_{ij}=\frac{1}{2b}\sum_{r=0}^{k_N}\sum_{s=0}^{k_N}\frac{1}{(-2b)^{r+s}}\left(\begin{array}{c}
        r+s \\
        s
    \end{array}\right)\partial_t^rP_{k_i}\partial_t^sP^*_{k_j}.
\end{align}
By using the judicious technique in \cite{Chakravarty}, we can rewrite the tau function $\tau(n)$ as
\begin{align}
    \tau(n)=\frac{1}{(2b)^{N}}\sum_{0\leq l_1<\cdots<l_N\leq k_N}\frac{|Q(l_1\cdots l_N)|^2}{(-2b)^{2(l_1+\cdots+l_N)}},
\end{align}
where
\begin{align}
    &Q(l_1\ \cdots\ l_N)=\sum_{0\leq r_1<\cdots<r_n\leq k_N}U\left(
    \begin{array}{c}
        l_1\ \cdots\ l_N \\
        r_1\ \cdots\ r_N
    \end{array}\right)\mathcal{P}(r_1\ \cdots\ r_N),\\
    &\mathcal{P}=\left(\begin{array}{cccc}
        P_{k_1} & P_{k_2} &\cdots & P_{k_N} \\
        \partial_tP_{k_1} & \partial_tP_{k_2} &\cdots & \partial_tP_{k_N} \\
        \vdots & \vdots & \vdots & \vdots\\
        \partial_t^{k_N}P_{k_1} & \partial_t^{k_N}P_{k_2} &\cdots &\partial_t^{k_N}P_{k_N}
    \end{array}\right),
\end{align}
and $U\left(
    \begin{array}{c}
        l_1\ \cdots\ l_N \\
        r_1\ \cdots\ r_N
\end{array}\right)$ denotes  the $N\times N$ minor of $U$ corresponding to the rows  $(l_1\ \cdots\ l_N)$ and columns $(r_1\ \cdots\ r_N)$. $U$ and $D$ are $(k_N+1)\times(k_N+1)$ matrix, whose elements are given by
\begin{align}
    U_{rs}&=\left\{
    \begin{array}{cc}
        \frac{1}{(-2b)^{s-r}}\left(
        \begin{array}{c}
            s \\
            r
        \end{array}\right),& r\leq s \\
        0, & r>s
    \end{array}\right.,\\
    D_{rr}&=(-2b)^{-2r},\quad r,s=0,1,\cdots,k_n .
\end{align}
By virtue of the properties of Schur polynomials, we have
\begin{align}
    \partial_{\mu_j}P_{k_i}=P_{k_i-j},\quad \partial_tP_{k_i}=\sum_{j=1}^{k_i}\frac{\partial \mu_j}{\partial t}P_{k_i-j}=-\I e^{p}\sum_{j=1}^{k_i}\frac{1}{j!}P_{k_i-j}.
\end{align}
Therefore, $\mathcal{P}$ admits the expression
\begin{align}
    \mathcal{P}=\left(\begin{array}{c}
        e_1^{(0)} \\
        e_1^{(1)}A \\
        \vdots \\
        e_1^{(k_N)}A^{k_N}
    \end{array}\right)\left(\begin{array}{cccc}
        P_{k_1} & P_{k_2} &\cdots &P_{k_N} \\
        P_{k_1-1} & P_{k_2-1} &\cdots &P_{k_N-1} \\
        \vdots & \vdots & \vdots & \vdots \\
        P_{k_1-k_N} & P_{k_2-k_N} &\cdots &P_{k_N-k_N}
    \end{array}\right),
\end{align}
where $A$ is a $(k_N+1)\times(k_N+1)$ matrix in the  form
\begin{align}
    A=\left(\begin{array}{cccccc}
        0 & 1 & \frac{1}{2} & \frac{1}{3!} & \cdots &\frac{1}{k_N!}\\
        0 & 0 & 1 & \frac{1}{2} & \cdots &\frac{1}{(k_N-1)!}\\
        0 & 0 & 0 & 1 & \cdots &\frac{1}{(k_N-2)!}\\
        \vdots & \vdots & \vdots & \vdots & \vdots & \vdots \\
        0 & 0 & 0 & 0 & \cdots & 1 \\
        0 & 0 & 0 & 0 & \cdots & 0 \\
    \end{array}\right),
\end{align}
and $e_1^{(j)}=((-\I)^je^{jp},0,\cdots,0)$ is a $1\times(k_N+1)$ row vector.

For example, with the parameters selected as $N=2$, $k_1=1$, $k_2=3$, $p=q^*=\frac{\pi}{2}\I$, we  obtain
\begin{align}
    U=\left(\begin{array}{cccc}
        1 & -\frac{1}{2}  & \frac{1}{4} & -\frac{1}{8} \\
        0 & 1             & -1          & \frac{3}{4} \\
        0 & 0             & 1           & -\frac{3}{2} \\
        0 & 0             & 0           & 1
    \end{array}\right), \quad
    \mathcal{P}=\left(\begin{array}{cc}
        P_1 & P_3 \\
        1   & P_2+\frac{1}{2}P_1+\frac{1}{3!} \\
        0   & P_1+1 \\
        0   & 1
    \end{array}\right),
\end{align}
\begin{align}
    \tau(n)=\frac{1}{4}\left(\frac{\left|\frac{1}{3}\mu_1^3-\frac{1}{2}\mu_1^2+\frac{1}{6}\mu_1-\mu_3\right|^2}{4}+\frac{\left(\left|\mu_1-\frac{1}{2}\right|^2+\frac{1}{4}\right)^2}{16}\right).
\end{align}
Under the variable transformation (\ref{solution}), the three-lump solutions can be constructed, as displayed in Fig. \ref{fig6}. The expression of the three-lump solution is rather complicated and is therefore omitted herein.

\begin{figure}
\begin{center}
\begin{tabular}{ccc}
\includegraphics[height=0.220\textwidth,angle=0]{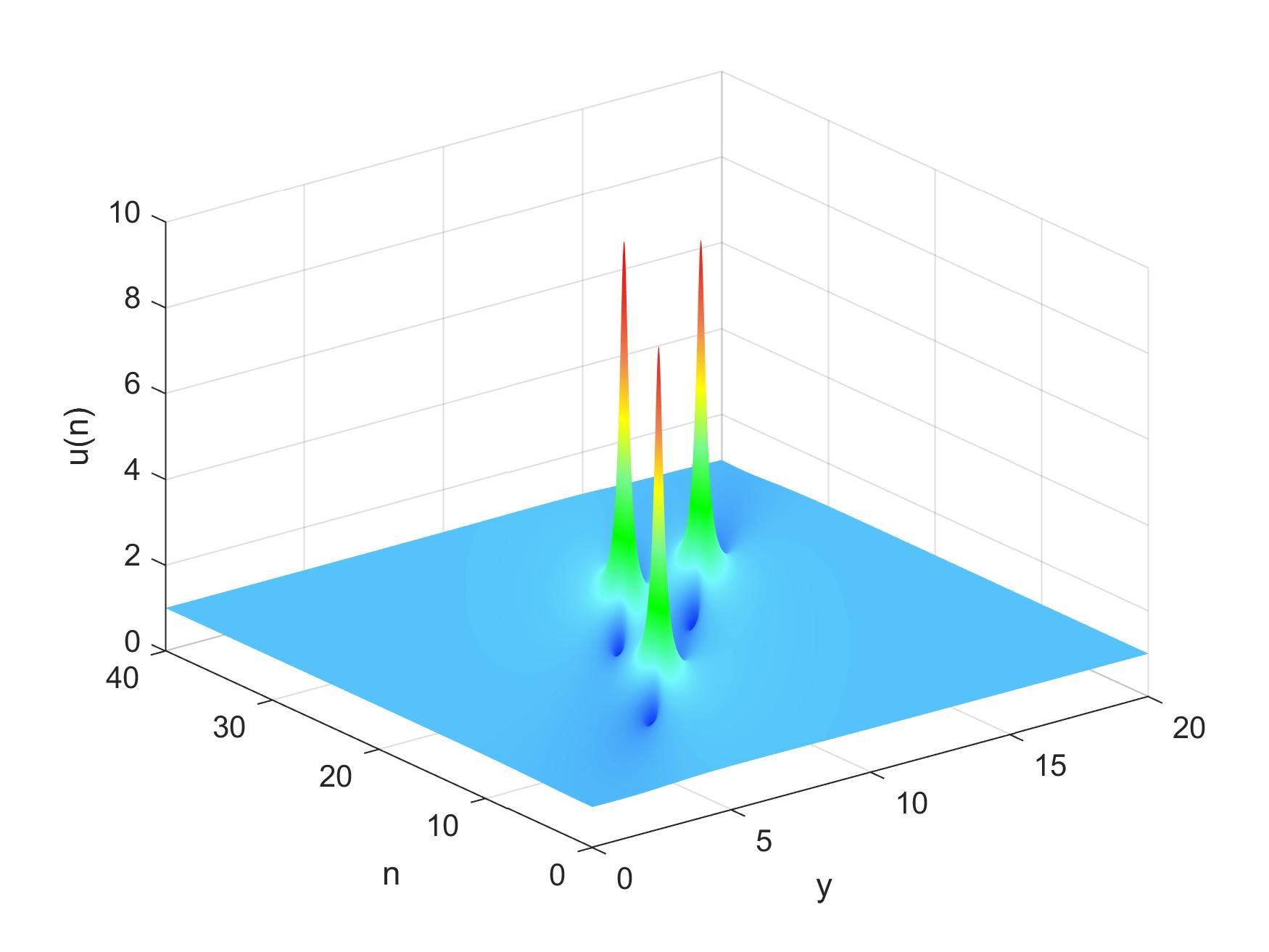} &
\includegraphics[height=0.220\textwidth,angle=0]{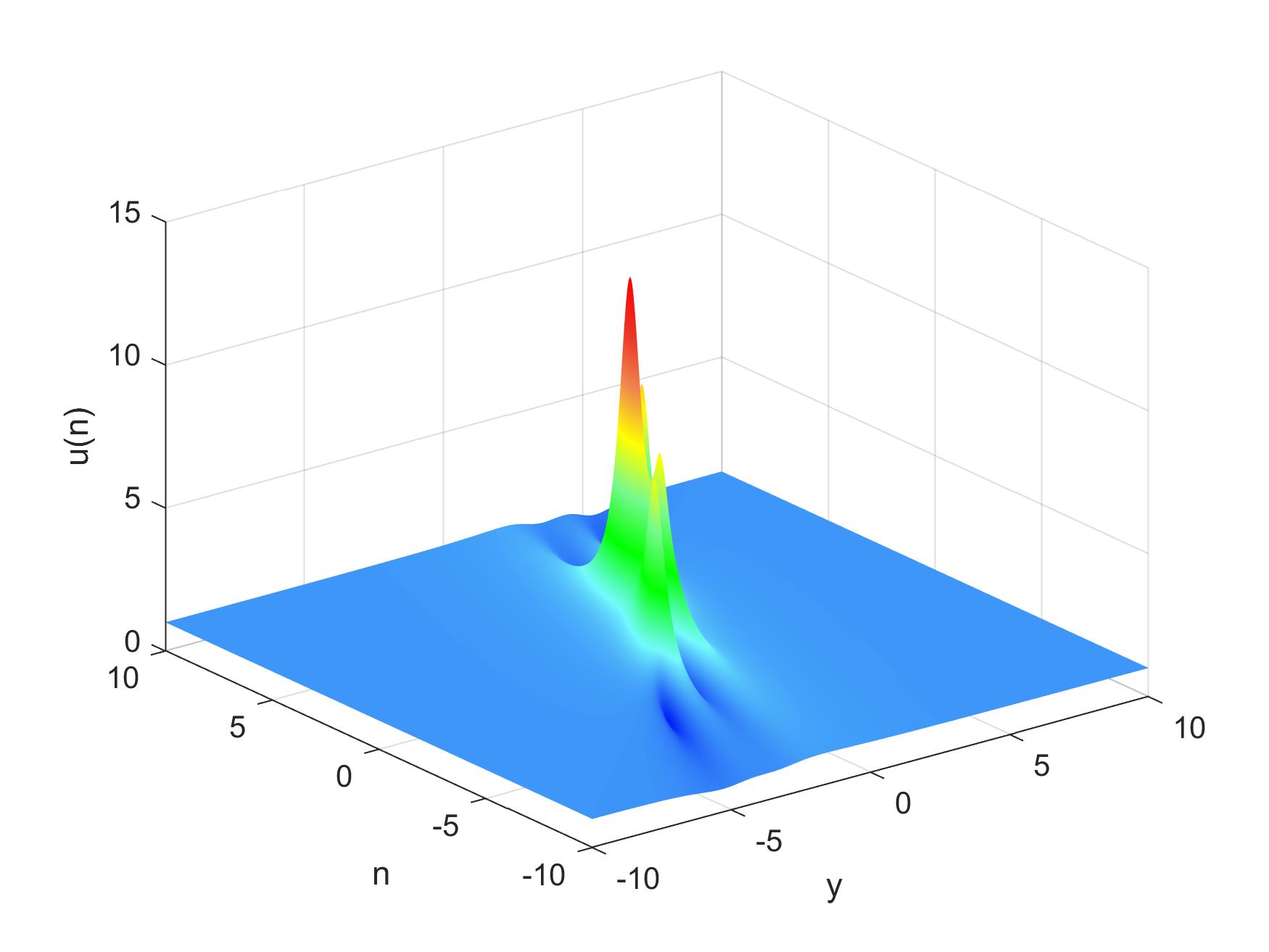} &
\includegraphics[height=0.220\textwidth,angle=0]{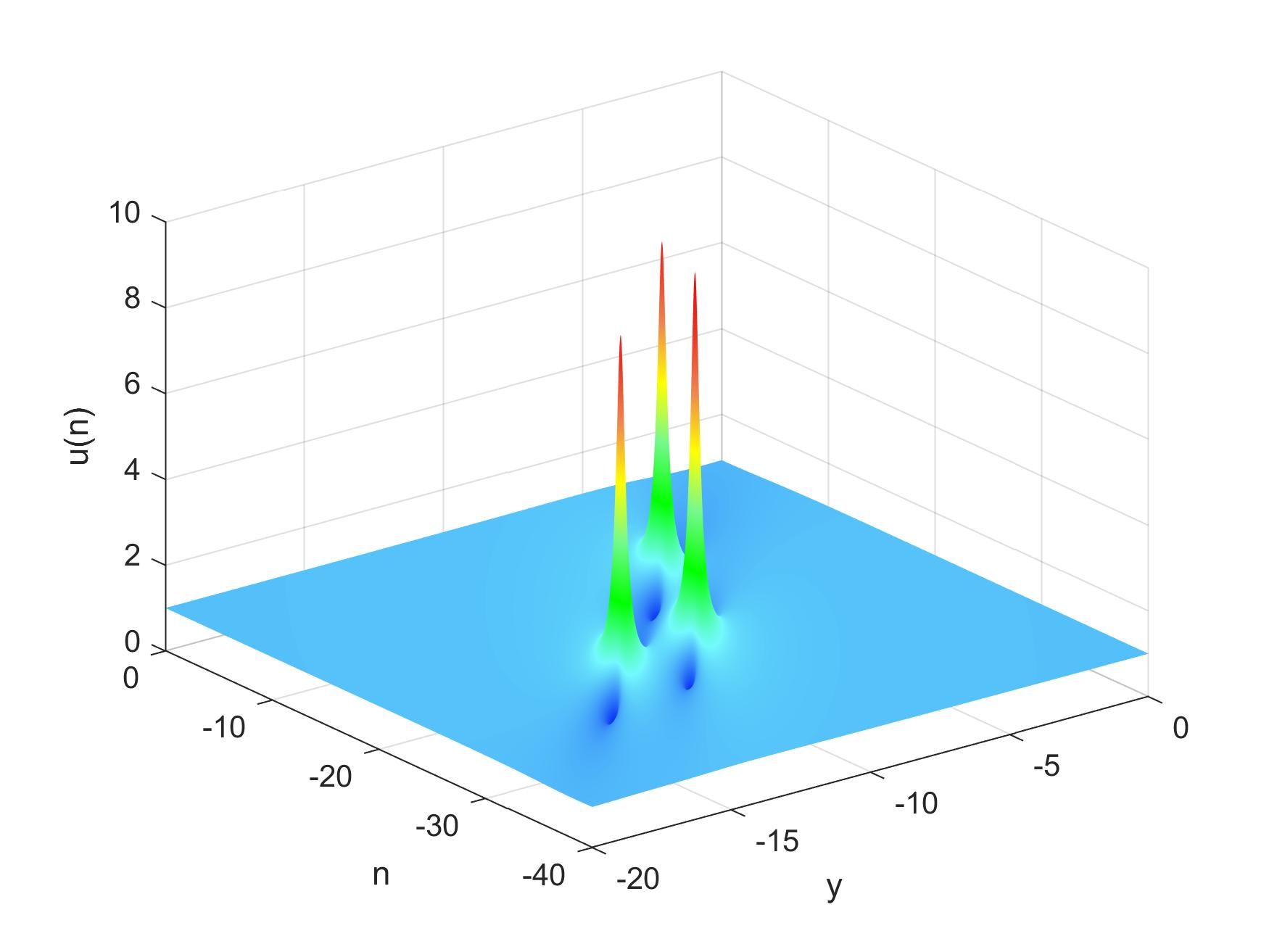} \\
(a) $t=-10$  & \quad  (b) $t=0$ & \quad (c) $t=10$ \\
\includegraphics[height=0.220\textwidth,angle=0]{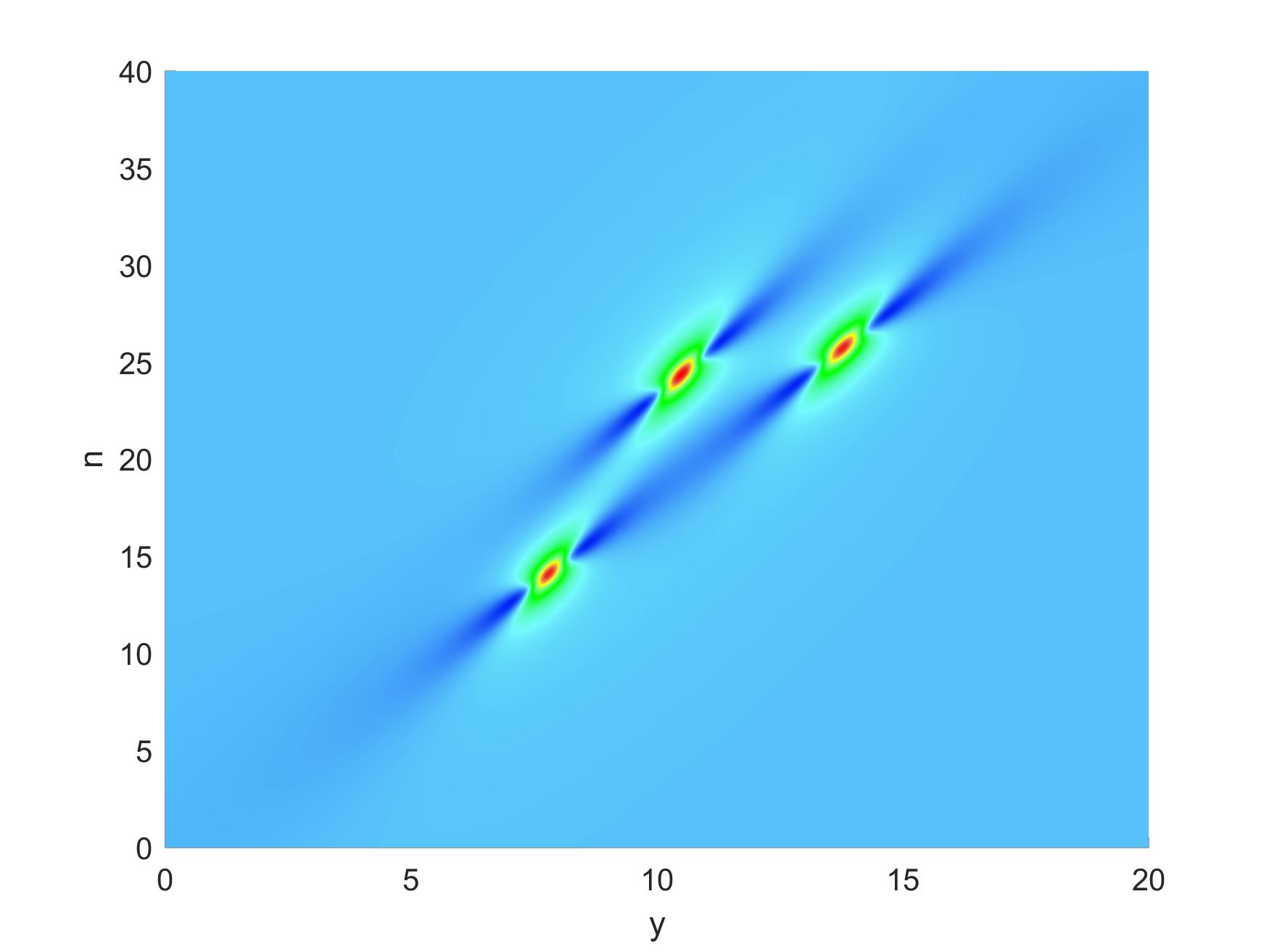} &
\includegraphics[height=0.220\textwidth,angle=0]{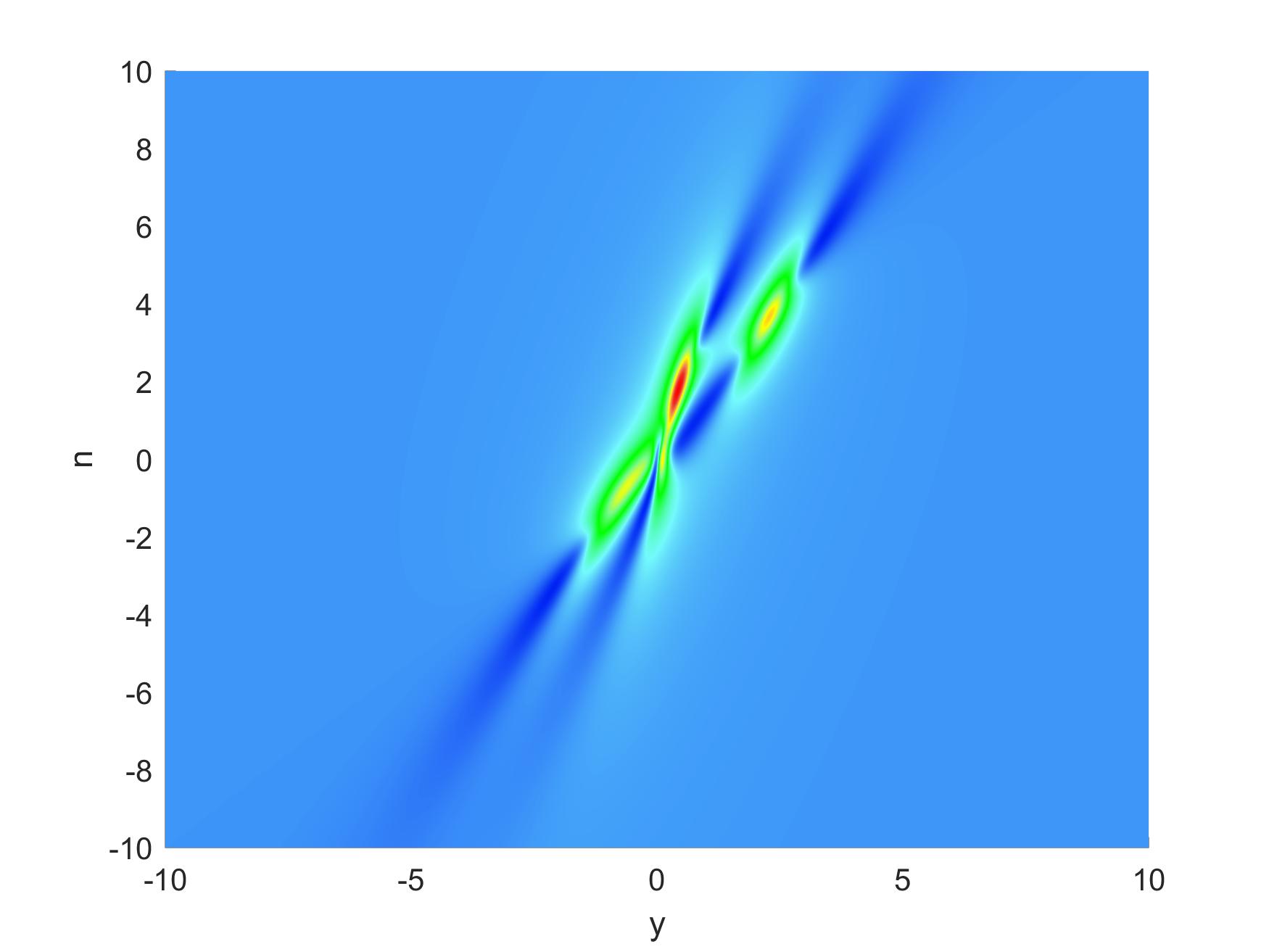} &
\includegraphics[height=0.220\textwidth,angle=0]{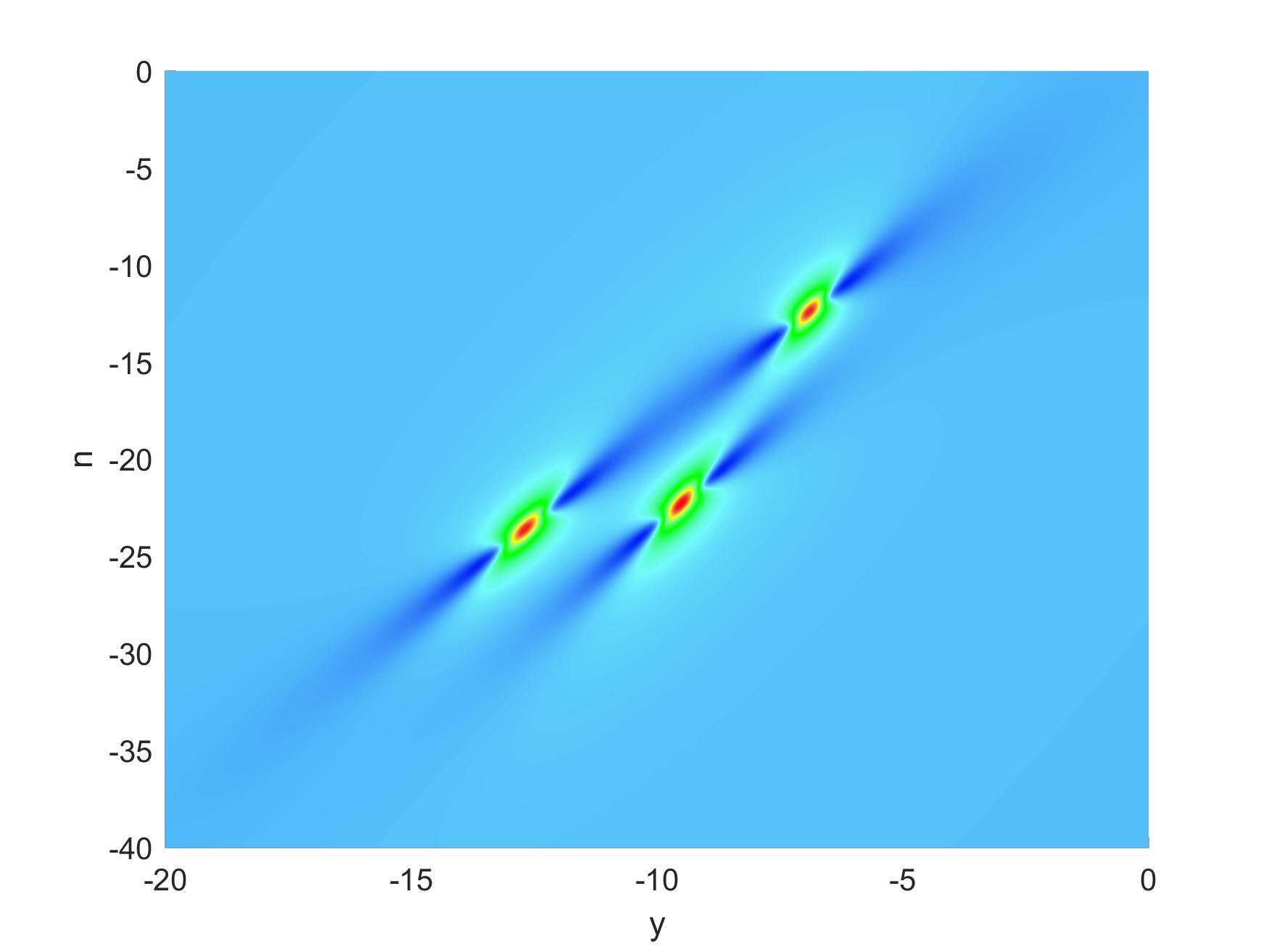} \\
(d) density plot of (a)  & \quad (e) density plot of (b) & \quad (f) density plot of (c)
\end{tabular}
\end{center}
\caption{
Three-lump solutions with parameters:
$N=2$, $p=q^*=\frac{\pi}{2}\I$, $k_1=1$, $k_2=3$.}\label{fig6}
\end{figure}

\section{Breather solutions to the variant BS lattice equation}\label{sec5}
In this section, we derive breather solutions by setting
\begin{align}
    \phi_j(n)=\sum_{r=1}^me^{\eta_{rj}},\quad \psi_k(n)=\sum_{s=1}^me^{\xi_{sk}},
   \end{align}
where
\begin{align}
    &\eta_{rj}=p_{rj}n-\I e^{\I p_{rj}}t-\I(e^{-\I p_{rj}}+e^{2\I p_{rj}})y-\I e^{2\I p_{rj}}z+\eta_{0,rj},\\
    &\xi_{sk}=q_{sk}n+\I e^{-\I q_{sk}}t+\I(e^{\I q_{sk}}+e^{-2\I q_{sk}})y+\I e^{-2\I q_{sk}}z+\xi_{0,sk},
\end{align}
and thus  we have 
\begin{align}
m_{jk}(n)=\sum_{r,s=1}^m\frac{\I}{e^{\I p_{rj}}-e^{-\I q_{sk}}}e^{\eta_{rj}+\xi_{sk}}.
\end{align}
Here $p_{rj}$, $q_{sk}$, $\eta_{0,rj}$, $\xi_{0,sk}$ are arbitrary complex constants. Taking $m=2$, $z=0$, $p_{1j}=q_{1j}^*$, $p_{2j}=q_{2j}^*$ and $\eta_{0,ij}=\xi_{0,ij}^*$, we can construct a general breather solution. For example, we take $N=1$, $p_{11}=q_{11}^*=a_1+\I b_1$, $p_{21}=q_{21}^*=a_2+\I b_2$ and $\xi_{0,i1}=\eta_{0,i1}^*$, which leads to
\begin{align}
    \tau(n)=m_{11}(n)=\frac{e^{b_1}}{2\sin(a_1)}e^{\zeta_1}+\frac{e^{b_2}}{2\sin(a_2)}e^{\zeta_2}+\frac{2Ae^{\zeta_3}\cos(\zeta_4)}{A^2+B^2}-\frac{2 Be^{\zeta_3}\sin(\zeta_4)}{A^2+B^2},
\end{align}
where
\begin{align*}
    &\zeta_1=2a_1n+2e^{-b_1}\sin(a_1)t+2\left(-e^{b_1}\sin(a_1)+e^{-2b_1}\sin(2a_1)\right)y+\xi_{0,1}+\eta_{0,1},\\
    &\zeta_2=2a_2n+2e^{-b_2}\sin(a_2)t+2\left(-e^{b_2}\sin(a_2)+e^{-2b_2}\sin(2a_2)\right)y+\xi_{0,2}+\eta_{0,2},\\
    &\zeta_3=(a_1+a_2)n+At+\left(-e^{b_1}\sin(a_1)-e^{b_2}\sin(a_2)+e^{-2b_1}\sin(2a_1)+e^{-2b_2}\sin(2a_2)\right)y+\Re(\eta_{0,2}+\xi_{0,1}),\\
    &\zeta_4=(b_1-b_2)n-Bt+\left(-e^{b_1}\cos(a_1)+e^{b_2}\cos(a_2)-e^{-2b_1}\cos(2a_1)+e^{-2b_2}\cos(2a_2)\right)y+\Im(\eta_{0,2}+\xi_{0,1}),\\
    &A=e^{-b_1}\sin(a_1)+e^{-b_2}\sin(a_2),\quad B=e^{-b_1}\cos(a_1)-e^{-b_2}\cos(a_2).
\end{align*}
One can verify that $\zeta_1-\zeta_3=-(\zeta_2-\zeta_3)$, and then we have the expression
\begin{equation}
    \begin{split}
        e^{-\zeta_3}\tau(n)&=\frac{e^{b_1}}{2\sin(a_1)}e^{\zeta_1-\zeta_3}+\frac{e^{b_2}}{2\sin(a_2)}e^{\zeta_2-\zeta_3}+\frac{2A\cos(\zeta_4)}{A^2+B^2}-\frac{2 B\sin(\zeta_4)}{A^2+B^2}.
    \end{split}
\end{equation}
It can be readily verified that the breather solution $u(n)=\frac{\tau(n-\I)\tau(n+\I)}{\tau(n)^2}$ is non-singular provided that the parameters satisfy $\sin(a_1)\sin(a_2)>0$. Given the cumbersomeness of its expression, it is omitted herein. While different parameters yield a general one-breather solution, they are insufficient to derive the specific forms of Kuznetsov-Ma breathers and Akhmediev breathers.

For the sake of simplicity, we denote $e^{\I p_{11}}=\alpha_1+\I \beta_1$, $e^{\I p_{21}}=\alpha_2+\I \beta_2$, and suppose $q_{ij}=p_{ij}^*$ and $\xi_{0,i}=\eta_{0,i}=0$. Then  $\tau(n)$ is given by
\begin{align}
    \tau(n)=m_{11}(n)=\frac{1}{2\beta_1}e^{2\tilde{\zeta}_1}+\frac{1}{2\beta_2}e^{2\tilde{\zeta}_2}+\frac{2(\beta_1+\beta_2)e^{\tilde{\zeta}_1+\tilde{\zeta}_2}\cos(\tilde{\zeta}_3)}{(\beta_1+\beta_2)^2+(\alpha_1-\alpha_2)^2}+\frac{2(\alpha_2-\alpha_1)e^{\tilde{\zeta}_1+\tilde{\zeta}_2}\sin(\tilde{\zeta}_3)}{(\beta_1+\beta_2)^2+(\alpha_1-\alpha_2)^2},
\end{align}
where
\begin{align}
    &\tilde{\zeta}_1=\arg(\alpha_1+\I\beta_1)n+\beta_1t+\left(-\frac{\beta_1}{\alpha_1^2+\beta_1^2}+2\alpha_1\beta_1\right)y,\\
    &\tilde{\zeta}_2=\arg(\alpha_2+\I\beta_2)n+\beta_2t+\left(-\frac{\beta_2}{\alpha_2^2+\beta_2^2}+2\alpha_2\beta_2\right)y,\\
    &\tilde{\zeta}_3=\frac{1}{2}\left(\ln(\alpha_2^2+\beta_2^2)-\ln(\alpha_1^2+\beta_1^2)\right)n+(\alpha_2-\alpha_1)t+\left(-\frac{\alpha_1}{\alpha_1^2+\beta_1^2}+\frac{\alpha_2}{\alpha_2^2+\beta_2^2}+\alpha_2^2-\beta_2^2-\alpha_1^2+\beta_1^2\right)y.
\end{align}
It is easy to prove $\tau(n)\neq0$ when $\beta_1\beta_2>0$, which means $u(n)$ is non-singular. For example, by taking $n=0$, $\alpha_1=\beta_1=1$, $\alpha_2=\beta_2=\frac{1}{2}$, we derive the general one-breather solution
\begin{align}
    u(0)=\frac{\left(\sqrt{2}\cosh(\frac{t}{2}+2y-\frac{1}{2}\ln 2)+\frac{3\sqrt{2}}{2}\cos(\frac{y}{2}-\frac{t}{2}-\frac{\pi}{4})\right)^2+\frac{81}{50}\cos^2(\frac{y}{2}-\frac{t}{2}+\frac{\pi}{4})}{\left(\sqrt{2}\cosh(\frac{t}{2}+2y-\frac{1}{2}\ln 2)+\frac{6}{5}\cos(\frac{1}{2}y-\frac{1}{2}t)+\frac{6}{5}\sin(\frac{1}{2}y-\frac{1}{2}t)\right)^2}.
\end{align}
The plot is illustrated in Fig. \ref{fig7} (a) and (d). In addition, by selecting parameters such that $\alpha_1=\alpha_2$, $-\frac{\beta_1}{\alpha_1^2+\beta_1^2}+2\alpha_1\beta_1+\frac{\beta_2}{\alpha_2^2+\beta_2^2}-2\alpha_2\beta_2=0$, we obtain the Akhmediev breather. This breather is localized in the $t-$direction and periodic in the $y-$direction, as illustrated in Fig. \ref{fig7} (b) and (e). Similarly,  setting $\beta_1=\beta_2$, $\frac{\alpha_1}{\alpha_1^2+\beta_1^2}-\frac{\alpha_2}{\alpha_2^2+\beta_2^2}+\alpha_2^2-\beta_2^2-\alpha_1^2+\beta_1^2=0$, we get the Kuznetsov-Ma breather, which is localized in the $y-$direction and periodic in the $t-$direction (see Fig. \ref{fig7} (c) and (f)). Due to their considerable complexity, the explicit expressions of the Akhmediev breather and Kuznetsov-Ma breather are omitted herein. Furthermore, for $m>2$, resonantly interacting breathers can be derived.

\begin{figure}
\begin{center}
\begin{tabular}{ccc}
\includegraphics[height=0.220\textwidth,angle=0]{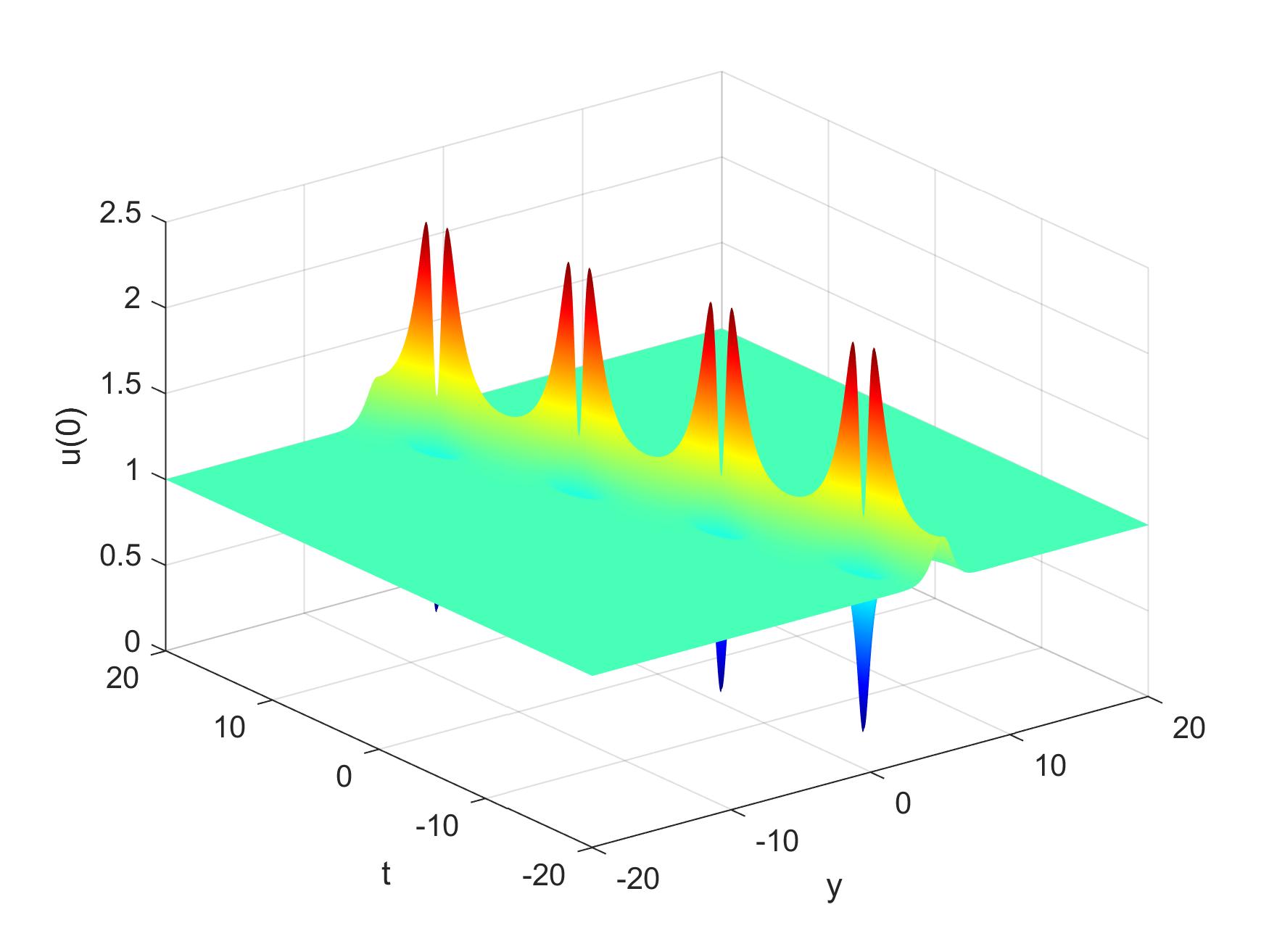} &
\includegraphics[height=0.220\textwidth,angle=0]{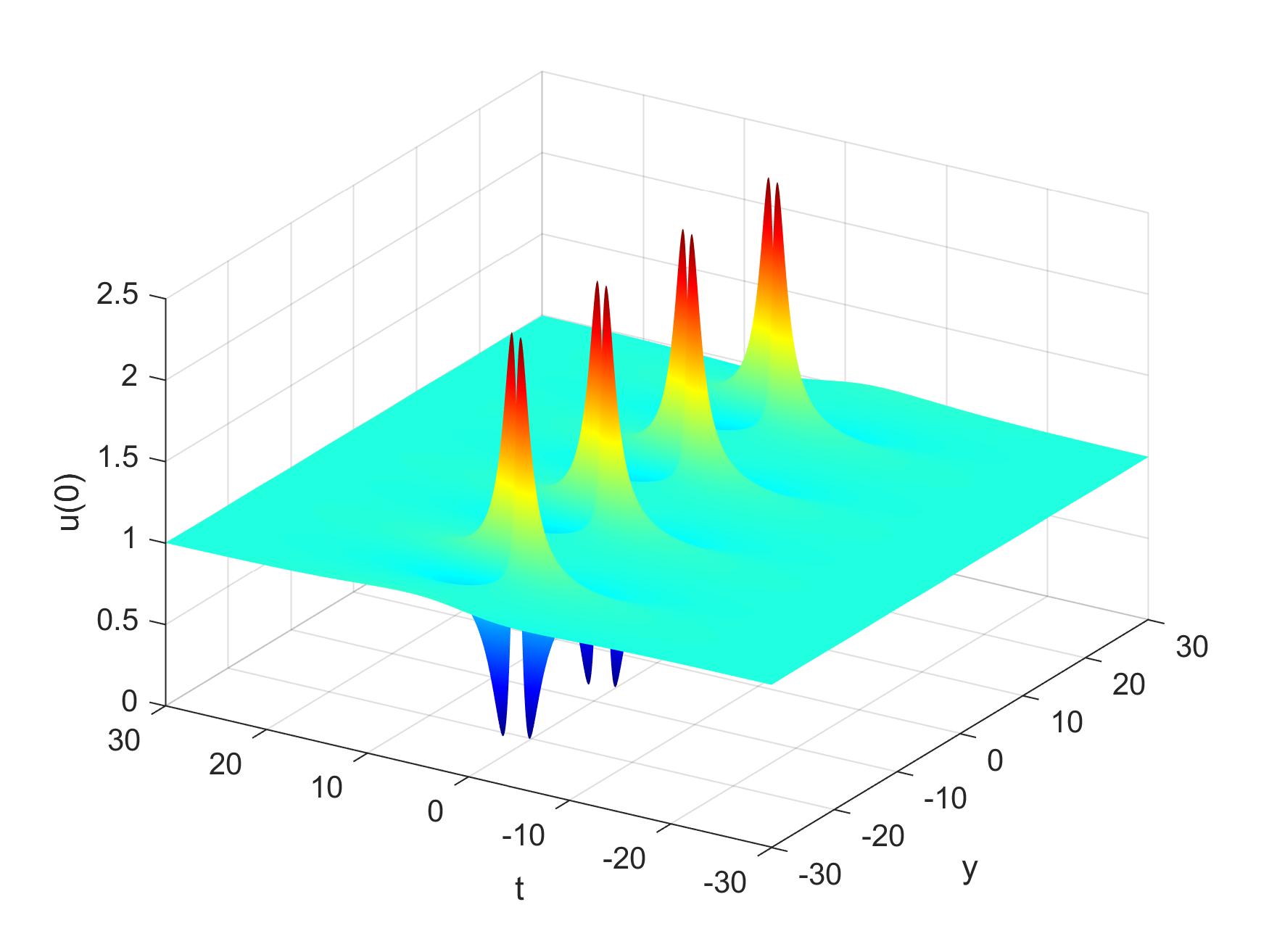} &
\includegraphics[height=0.220\textwidth,angle=0]{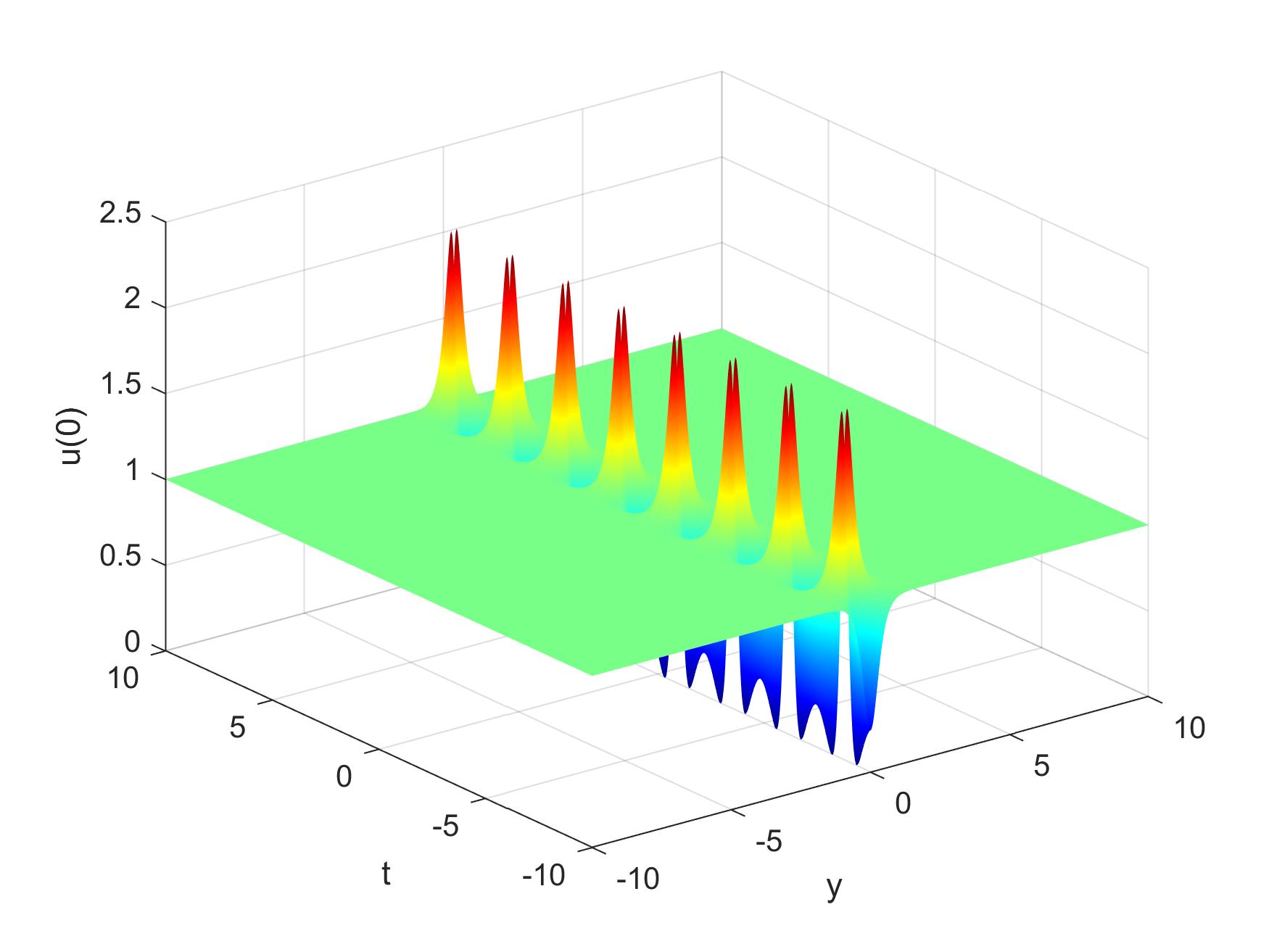} \\
(a) One-breather  & \quad  (b) Akhmediev breather & \quad (c) Kuznetsov-Ma breather \\
\includegraphics[height=0.220\textwidth,angle=0]{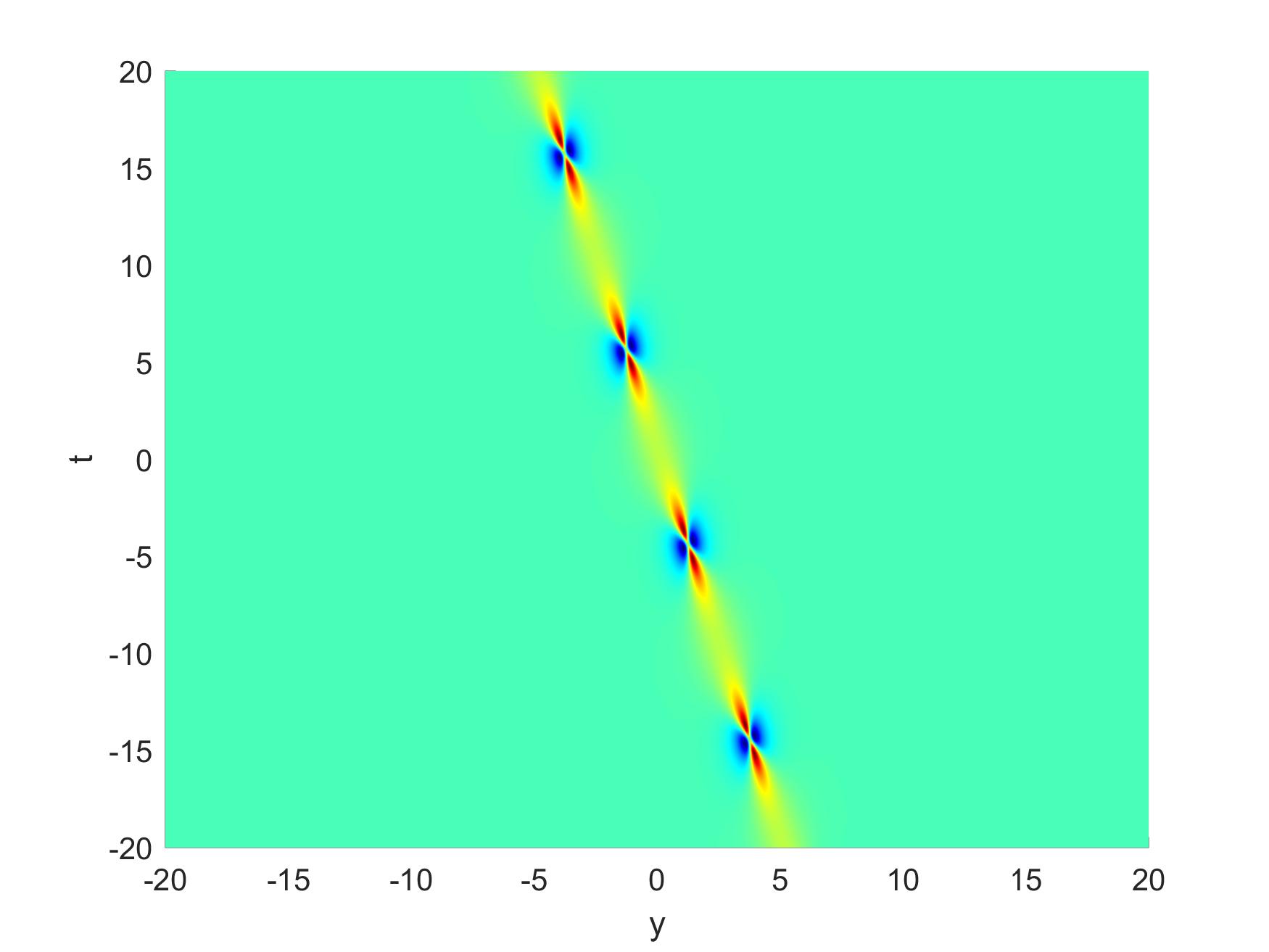} &
\includegraphics[height=0.220\textwidth,angle=0]{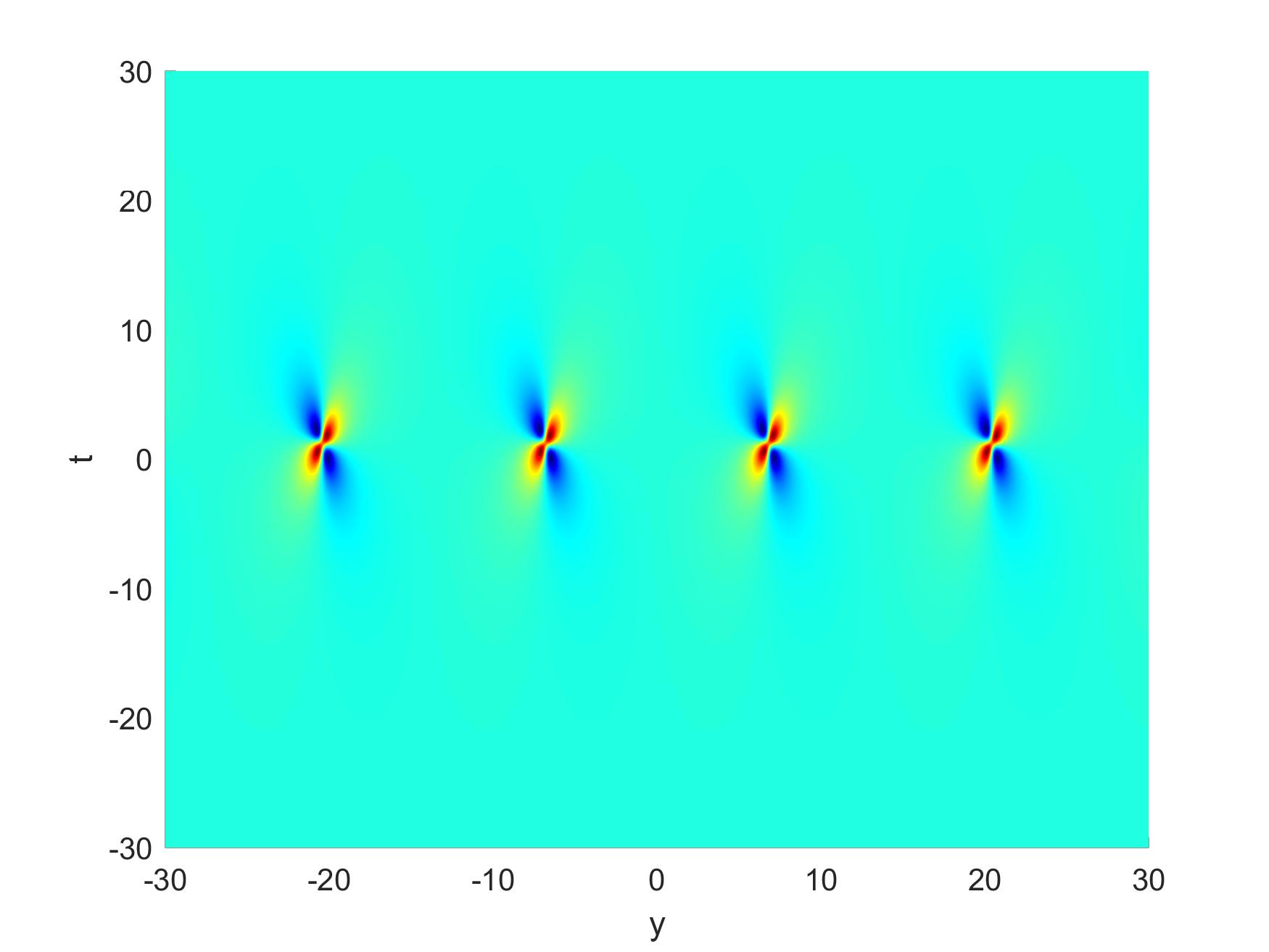} &
\includegraphics[height=0.220\textwidth,angle=0]{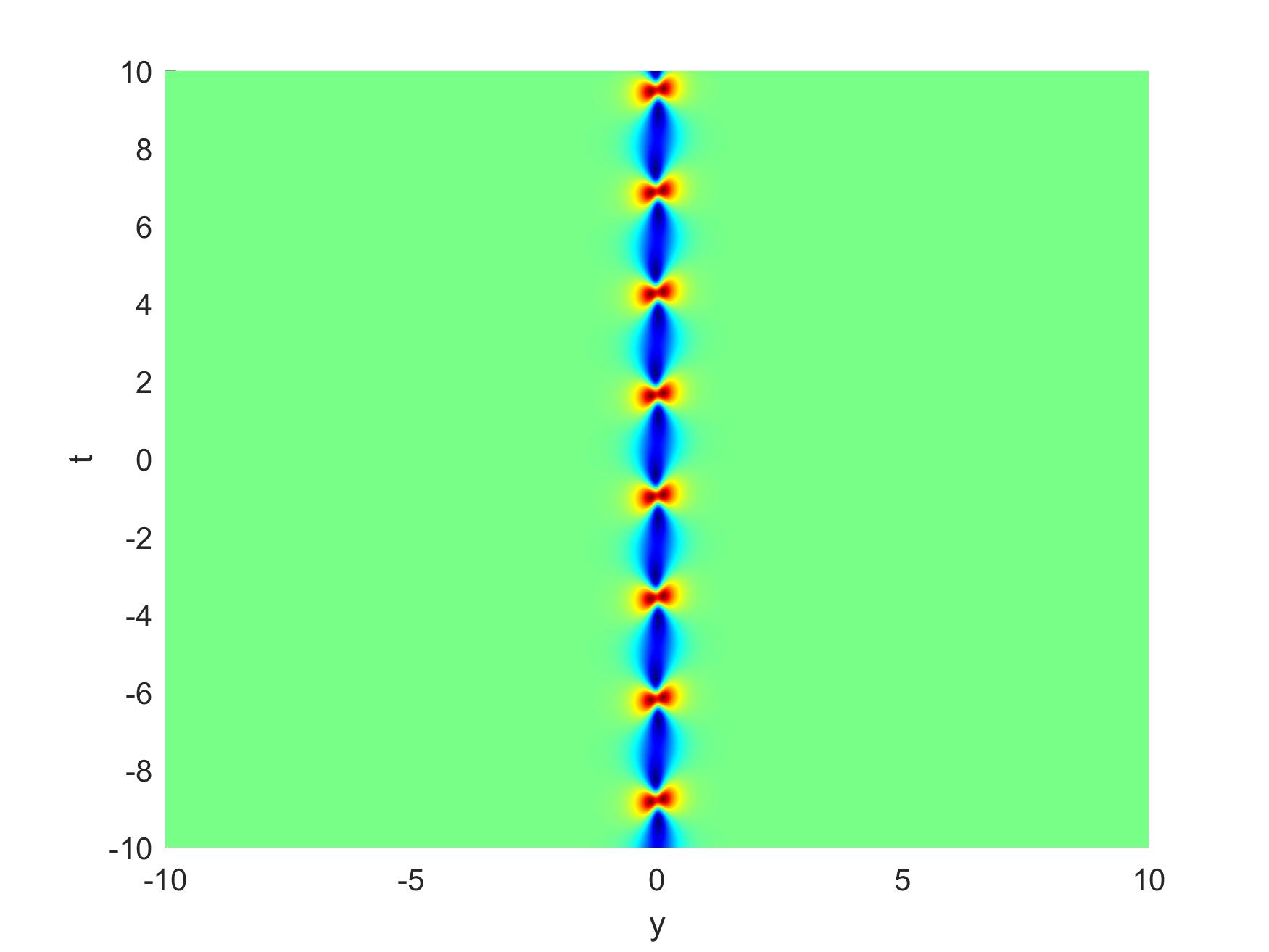} \\
(d) density plot of (a)  & \quad (e) density plot of (b) & \quad (f) density plot of (c)
\end{tabular}
\end{center}
\caption{
Breather solutions with parameters:
(a) $\alpha_1=1$, $\beta_1=1$, $\alpha_2=\frac{1}{2}$, $\beta_2=\frac{1}{2}$; (b) $\alpha_1=\alpha_2=-\frac{1}{3}$, $\beta_1=\frac{1}{2}$, $\beta_2=\frac{27}{26}-\frac{\sqrt{3077}}{78}$; (c) $\alpha_1=1$, $\alpha_2=-\frac{(152+6\sqrt{642})^{1/3}}{6}+\frac{1}{3(152+6\sqrt{642})^{1/3}}-\frac{1}{3}$, $\beta_1=\beta_2=1$.}\label{fig7}
\end{figure}

\section{Periodic wave solutions}\label{sec6}
In this section, we derive the three-periodic wave solutions by using the numerical method \cite{Nakamura1,JCP2018,Liang,ZhangYN}. We consider the following bilinear equations
\begin{align}
    &F_1\left(D_t,\ D_y,\ D_z,\ D_n,\ c_1\right)f\cdot f=0,\label{newbilinear1}\\
    &F_2\left(D_t,\ D_y,\ D_z,\ D_n,\ c_2\right)f\cdot f=0,\label{newbilinear2}
\end{align}
where $F_1$ and $F_2$ are certain unspecified functions of $D_t$, $D_y$, $D_z$, $D_n$, and $c_1$, $c_2$ are integration constants. The bilinear equations (\ref{newbilinear1}) and (\ref{newbilinear2}) have $g$-periodic wave solutions expressed by the Riemann theta function \cite{Nakamura1}
\begin{align}
    f=\sum_{m_1}\sum_{m_2}\cdots\sum_{m_g=-\infty}^{\infty}\exp\left[\I\sum_{j=1}^g(m_j+s_j)\eta_j-\frac{1}{2}\sum_{j,k=1}^g(m_j+s_j)\tau_{j,k}(m_k+s_k)\right],
\end{align}
if
\begin{equation}\label{F1}
    \begin{split}
        \sum_{m_1}\sum_{m_2}\cdots\sum_{m_g=-\infty}^{\infty}&F_1\left[2\I\sum_{j=1}^g\left(m_j-\frac{\mu_j}{2}\right)w_j,\ 2\I\sum_{j=1}^g\left(m_j-\frac{\mu_j}{2}\right)k_j,\ 2\I\sum_{j=1}^g\left(m_j-\frac{\mu_j}{2}\right)l_j,\right.\\
        &~~\left.2\I\sum_{j=1}^g\left(m_j-\frac{\mu_j}{2}\right)v_j\right]\times\exp\left[-\sum_{j,k=1}^g\left(m_j-\frac{\mu_j}{2}\right)\tau_{jk}\left(m_k-\frac{\mu_k}{2}\right)\right]=0,
    \end{split}
\end{equation}
\begin{equation}\label{F2}
    \begin{split}
        \sum_{m_1}\sum_{m_2}\cdots\sum_{m_g=-\infty}^{\infty}&F_2\left[2\I\sum_{j=1}^g\left(m_j-\frac{\mu_j}{2}\right)w_j,\ 2\I\sum_{j=1}^g\left(m_j-\frac{\mu_j}{2}\right)k_j,\ 2\I\sum_{j=1}^g\left(m_j-\frac{\mu_j}{2}\right)l_j,\right.\\
        &~~\left.2\I\sum_{j=1}^g\left(m_j-\frac{\mu_j}{2}\right)v_j\right]\times\exp\left[-\sum_{j,k=1}^g\left(m_j-\frac{\mu_j}{2}\right)\tau_{jk}\left(m_k-\frac{\mu_k}{2}\right)\right]=0,
    \end{split}
\end{equation}
where $\eta_j=w_jt+k_jy+l_jz+v_jn+\eta_j^0$, with $k_j$, $w_j$, $l_j$, $v_j$ associated with the wave numbers and the frequencies. Note that the system comprises $2^{g+1}$ equations from (\ref{F1}) and (\ref{F2}), involving a total of $4g+2+g(g+1)/2$ parameters. By fixing $k_i$, $l_i$, $\tau_{ii}$ as given parameters, we are left with $(g^2+3g+4)/2$ unknown parameters. Consequently, an over-determined nonlinear algebraic system must be solved to derive a $g$-periodic wave for $g\geq 2$.

\subsection{Three-periodic wave solutions to the variant BS lattice equation}\label{sec6.1}
To compute numerical three-periodic wave solutions of the variant BS lattice \eqref{variants-1}-\eqref{variants-3}, we introduce two constants $c_1$ and $c_2$ into the original bilinear equations (\ref{bilinear-1}) and (\ref{bilinear-2}), as shown below:
\begin{align}
    &\left(D_z\sin\left(\frac{1}{2}D_n\right)-D_t^2\cos\left(\frac{1}{2}D_n\right)+c_1\cos\left(\frac{1}{2}D_n\right)\right)\tau(n)\cdot\tau(n)=0,\\
    &\left(D_tD_z-D_tD_y-4\sin^2\left(\frac{1}{2}D_n\right)+c_2\right)\tau(n)\cdot\tau(n)=0.
\end{align}
Then we use the Gauss-Newton method to solve the over-determined system. We rewrite equations (\ref{F1}) and (\ref{F2}) as
\begin{align}
    \mathbf{H}(w_1,\ w_2,\ w_3,\ v_1,\ v_2,\ v_3,\ \tau_{12},\ \tau_{13},\ \tau_{23},\ c_1,\ c_2)=(H_1, H_2, \cdots,H_{16})^T=0,
\end{align}
where $H_i=0$, $(i=1,\ 2,\ \cdots,\ 16)$ is one of the equations in (\ref{F1})-(\ref{F2}) with $11$ unknown parameters. The objective function of the nonlinear least-square problem is
\begin{align}
    S(\mathbf{x})=\frac{1}{2}\mathbf{H}(\mathbf{x})^T\mathbf{H}(\mathbf{x}),
\end{align}
where $\mathbf{x}=(w_1,\ w_2,\ w_3,\ v_1,\ v_2,\ v_3,\ \tau_{12},\ \tau_{13},\ \tau_{23},\ c_1,\ c_2)^T$.

By utilizing the Gauss-Newton method, we have the iterative formula \cite{Liang,ZhangYN}
\begin{align}
    \mathbf{x}^{(i+1)}=\mathbf{x}^{(i)}-(\mathbf{J}^T\mathbf{J})^{-1}\mathbf{J}^T\mathbf{H}\big|_{\mathbf{x}=\mathbf{x}^{(i)}}
\end{align}
with an initial guess $\mathbf{x}^{(0)}$, where $\mathbf{x}^{(k)}$ is the $k$th output, and $\mathbf{J}$ is the Jacobian matrix of $\mathbf{H}$, i.e.
\begin{align}
    \mathbf{J}=\left[\frac{\partial H_i}{\partial x_j}\right]_{i=1,\cdots,16,\ j=1,\cdots, 11}.
\end{align}
The numerical solution, obtained through an iterative process using the Gauss-Newton method with appropriate parameters, is presented in Table \ref{tab1}. Setting $z=0$ and $t=0$, we plot the  three-periodic wave solution $u(n)$, as shown in Fig. \ref{p-wave}. Furthermore, Fig. \ref{p-wave-2} (a) displays these three-periodic waves for the variant BS lattice at different values of $n$ ($n=1, 3, 7, 15$) with $z=0$, whereas Fig. \ref{p-wave-2} (b) illustrates the solution profile $u(1)$ with different $y$ values.
\begin{rmk}
    When using the Gauss-Newton method for the iteration, we set the error tolerance to be $\epsilon=10^{-15}$, which means that the iteration ends when $\lVert\mathrm{x}^{(i+1)}-\mathrm{x}^{(i)}\rVert_2<\epsilon$ or $\lVert\mathrm{H}(\mathrm{x}^{(i)})\rVert_2<\epsilon$.
\end{rmk}
\begin{rmk}
    The Gauss-Newton method is highly sensitive to initial values. As a result, different initial guesses can lead to divergent numerical solutions or even cause the algorithm to fail entirely.
\end{rmk}
\begin{table}[htbp]
  \centering
  \caption{3-periodic waves to the variant BS lattice (\ref{variants-1})-\eqref{variants-3}.}
  \begin{tabular}{ccccccccccc}
    \hline\\[-3mm]
    $k_1$ & $k_2$ & $k_3$ & $l_1$ & $l_2$ & $l_3$ & $\tau_{11}$ & $\tau_{22}$ & $\tau_{33}$ & $c_1^{(0)}$ & $c_2^{(0)}$\\
    $\frac{2\pi}{10}$ & $2\times\frac{2\pi}{10}$ & $3\times\frac{2\pi}{10}$ & $\frac{2\pi}{8}$ & $2\times\frac{2\pi}{8}$ & $3\times\frac{2\pi}{8}$ & $0.67\times2\pi$ & $0.86\times2\pi$ & $1.02\times2\pi$ & $-1$ & $1$\\[1mm]
    \hline
    $w_1$ & $w_2$ & $w_3$ & $v_1$ & $v_2$ & $v_3$ & $\tau_{12}$ & $\tau_{13}$ & $\tau_{23}$ & $c_1$ & $c_2$\\
    $0.2887$ & $0.5713$ & $0.8422$ & $0.2127$ & $0.4209$ & $0.6204$ & $2.0438$ & $1.3937$ & $3.1119$ & $-0.0014$ & $-0.0008$\\
    \hline
  \end{tabular}
  \label{tab1}
\end{table}

\begin{figure}
\begin{center}
\begin{tabular}{ccc}
\includegraphics[height=0.220\textwidth,angle=0]{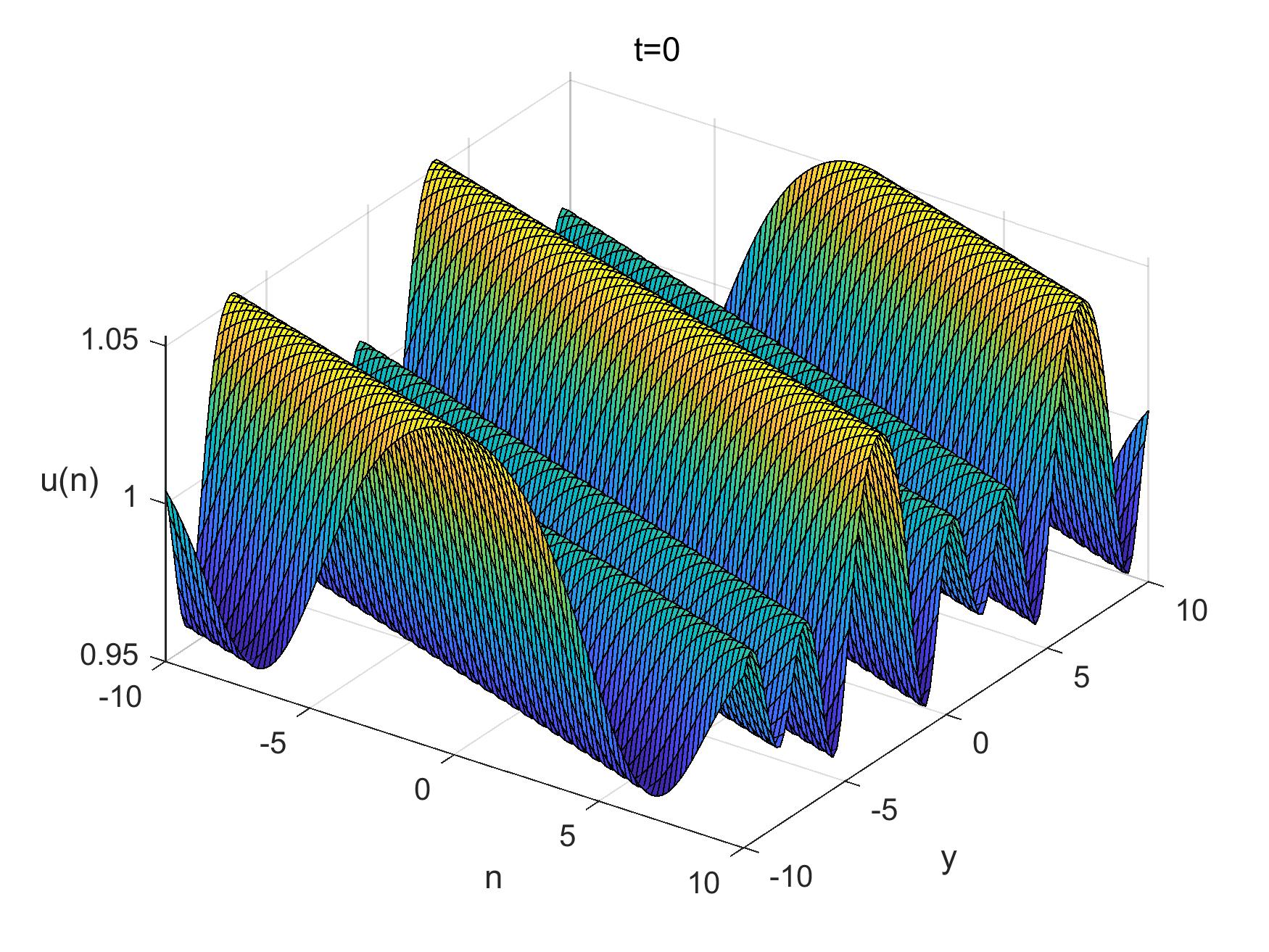} &
\includegraphics[height=0.220\textwidth,angle=0]{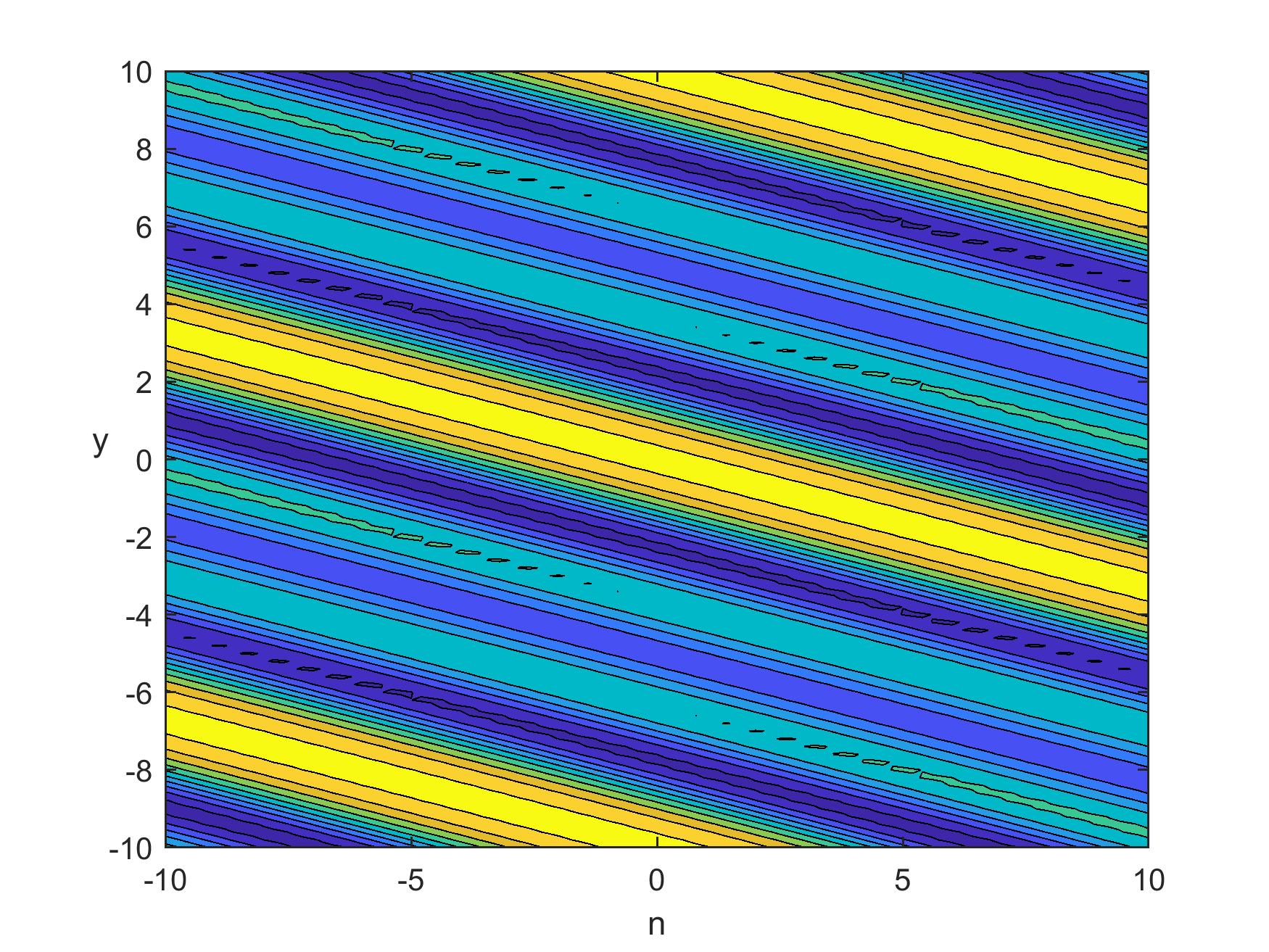}\\
(a) Three-dimensional plot of $u(n)$  & \quad  (b) Contour plot of $u(n)$
\end{tabular}
\end{center}
\caption{
Three-dimensional of $u(n)$ in Table \ref{tab1}.}\label{p-wave}
\end{figure}
\begin{figure}
\begin{center}
\begin{tabular}{ccc}
\includegraphics[height=0.220\textwidth,angle=0]{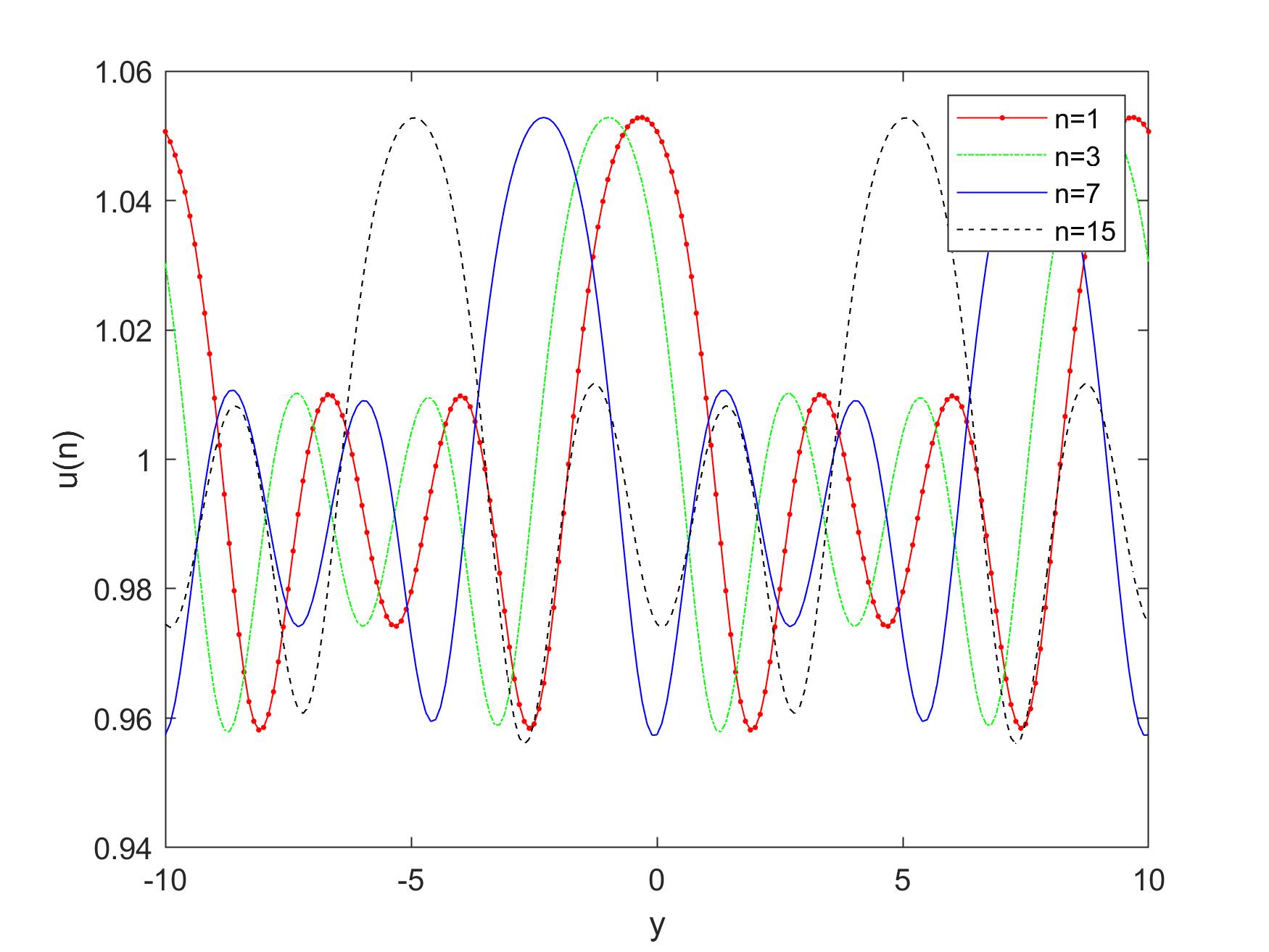} &
\includegraphics[height=0.220\textwidth,angle=0]{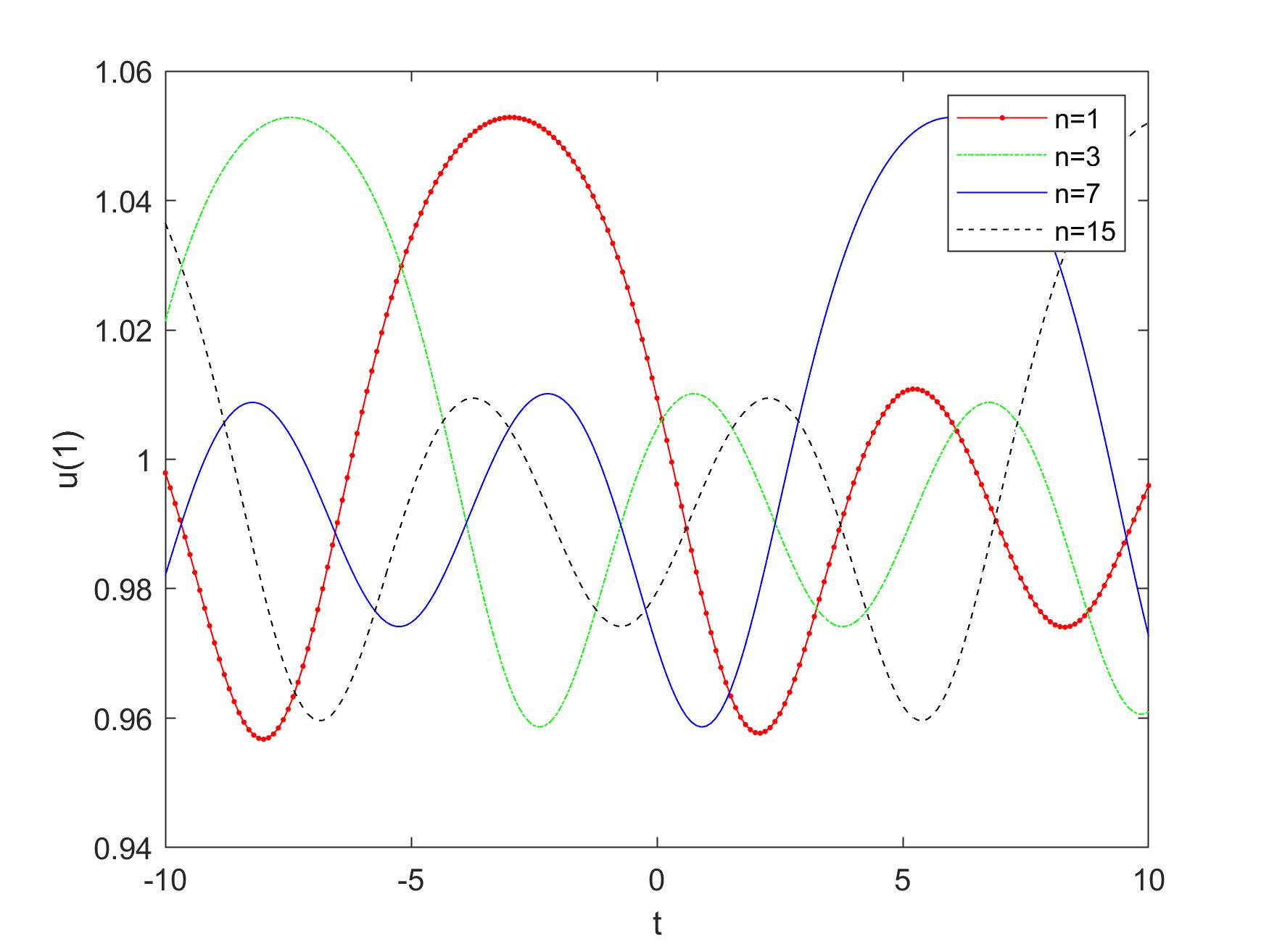}\\
(a) Three-periodic wave $u(n)$ with $t=0$  & \quad  (b) Three-periodic wave $u(n)$ with $n=1$
\end{tabular}
\end{center}
\caption{
Figures of 3-periodic waves to the variant BS lattice \eqref{variants-1}-\eqref{variants-3} with $z=0$ and parameters in Table \ref{tab1}.}\label{p-wave-2}
\end{figure}

\subsection{Three-periodic wave solutions to the BS lattice equations}
To compute the numerical $3$-periodic wave solutions of the BS lattice (\ref{gBS-1})-(\ref{gBS-3}), we introduce two integration constants, $d_1$ and $d_2$, into the bilinear equations (\ref{BSb-1}) and (\ref{BSb-2}),
\begin{align}
    &\left(D_z\sinh\left(\frac{1}{2}D_n\right)-D_t^2\cosh\left(\frac{1}{2}D_n\right)+d_1\cosh\left(\frac{1}{2}D_n\right)\right)\tau(n)\cdot\tau(n)=0,\\
    &\left(D_tD_z-D_tD_y-4\sinh^2\left(\frac{1}{2}D_n\right)+d_2\right)\tau(n)\cdot\tau(n)=0.
\end{align}
Using the same procedure described in subsection \ref{sec6.1}, we get the numerical 3-periodic wave solution with appropriate parameters, as shown in Table \ref{tab2}. With $z=0$ and $t=0$, Fig. \ref{p-wave-bs} shows the $3$-periodic wave solution $u(n)$. Fig. \ref{p-wave-bs-2} (a) depicts these 3-periodic waves to the BS lattice (\ref{gBS-1})-(\ref{gBS-3}) at different values of $n$ ($n=1,\ 3,\ 7,\ 15$), under the condition $z=0$. Meanwhile, Fig. \ref{p-wave-bs-2} (b) illustrates the solution profile of $u(n)$ at $n=1$ for different $y$ values. Under the same initial conditions, when compared with the results in Table \ref{tab1} and Table \ref{tab2}, the Gauss-Newton method can converge to the 3-periodic wave solutions of both the BS lattice and the variant BS lattice.

\begin{table}[htbp]
  \centering
  \caption{3-periodic waves to the BS lattice (\ref{gBS-1})-(\ref{gBS-3}).}
  \begin{tabular}{ccccccccccc}
    \hline\\[-3mm]
    $k_1$ & $k_2$ & $k_3$ & $l_1$ & $l_2$ & $l_3$ & $\tau_{11}$ & $\tau_{22}$ & $\tau_{33}$ & $d_1^{(0)}$ & $d_2^{(0)}$\\
    $\frac{2\pi}{10}$ & $2\times\frac{2\pi}{10}$ & $3\times\frac{2\pi}{10}$ & $\frac{2\pi}{8}$ & $2\times\frac{2\pi}{8}$ & $3\times\frac{2\pi}{8}$ & $0.67\times2\pi$ & $0.86\times2\pi$ & $1.02\times2\pi$ & $-1$ & $1$\\[1mm]
    \hline
    $w_1$ & $w_2$ & $w_3$ & $v_1$ & $v_2$ & $v_3$ & $\tau_{12}$ & $\tau_{13}$ & $\tau_{23}$ & $d_1$ & $d_2$\\
    $0.2900$ & $0.5868$ & $0.8997$ & $0.2137$ & $0.4323$ & $0.6627$ & $1.9790$ & $1.2802$ & $2.9081$ & $0.0017$ & $0.0009$\\
    \hline
  \end{tabular}
  \label{tab2}
\end{table}

\begin{figure}
\begin{center}
\begin{tabular}{ccc}
\includegraphics[height=0.220\textwidth,angle=0]{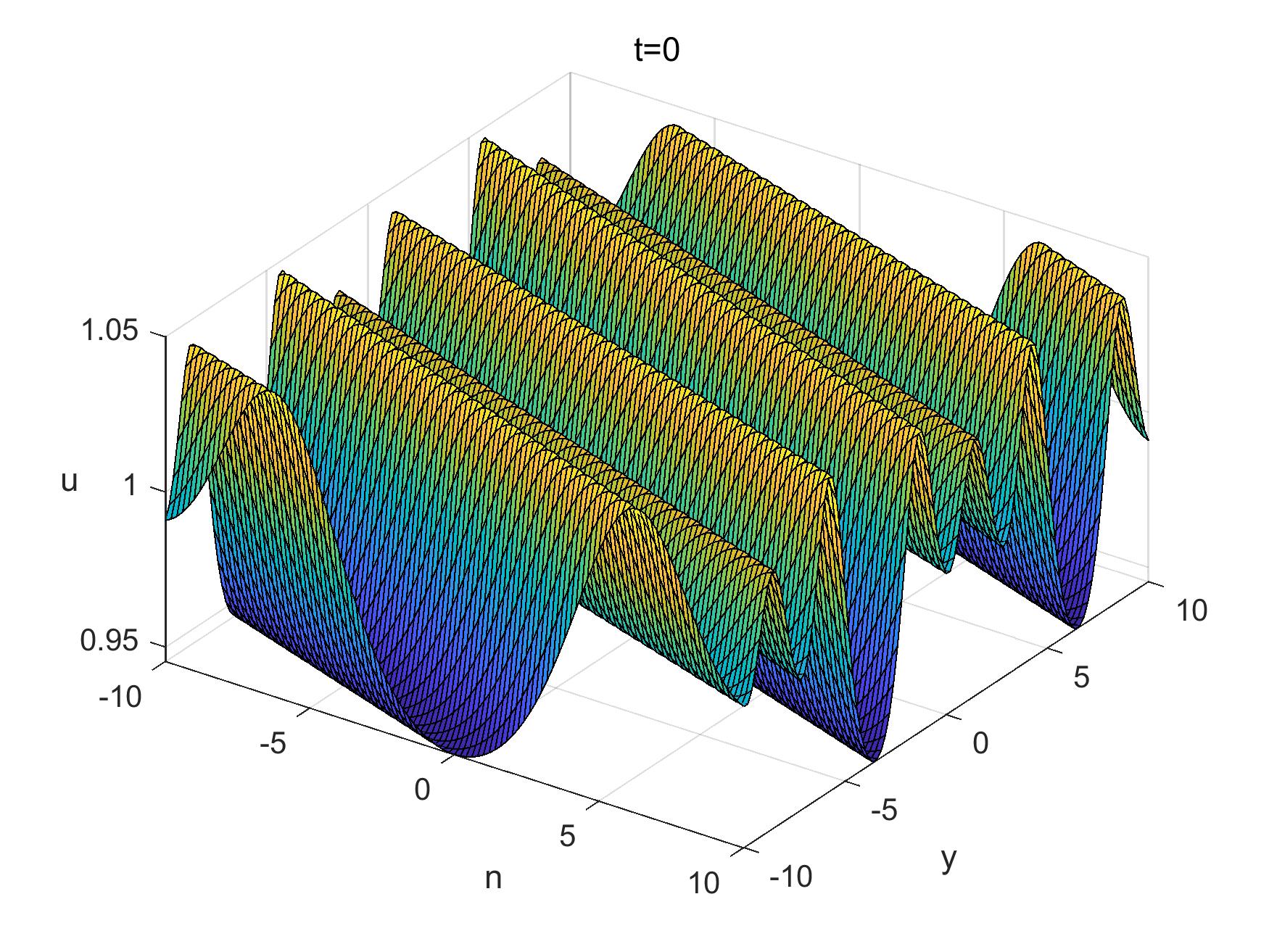} &
\includegraphics[height=0.220\textwidth,angle=0]{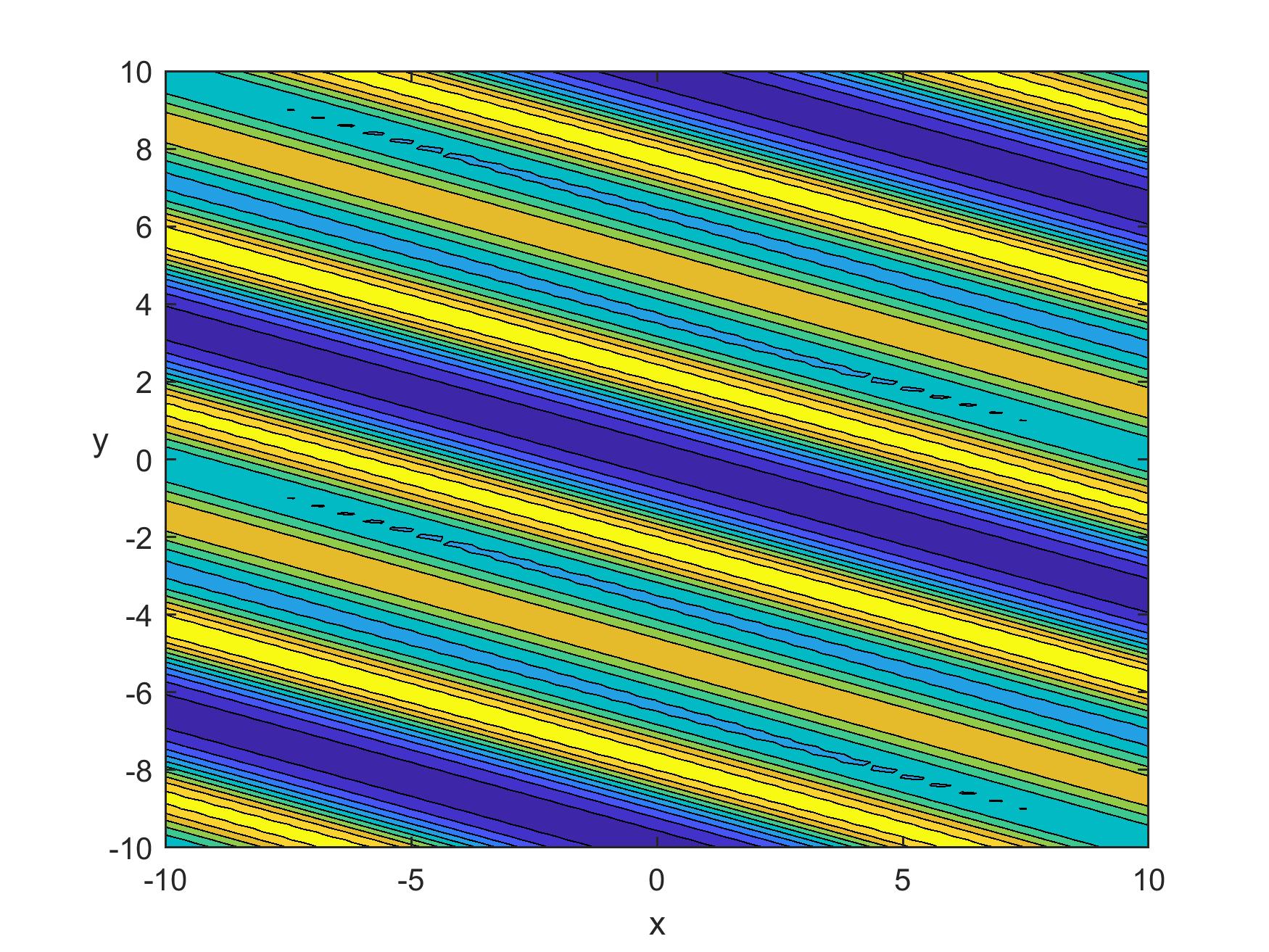}\\
(a) Three-dimensional plot of $u(n)$  & \quad  (b) Contour plot of $u(n)$
\end{tabular}
\end{center}
\caption{
Three-periodic waves of $u(n)$ with parameters in Table \ref{tab2}.}\label{p-wave-bs}
\end{figure}
\begin{figure}
\begin{center}
\begin{tabular}{ccc}
\includegraphics[height=0.220\textwidth,angle=0]{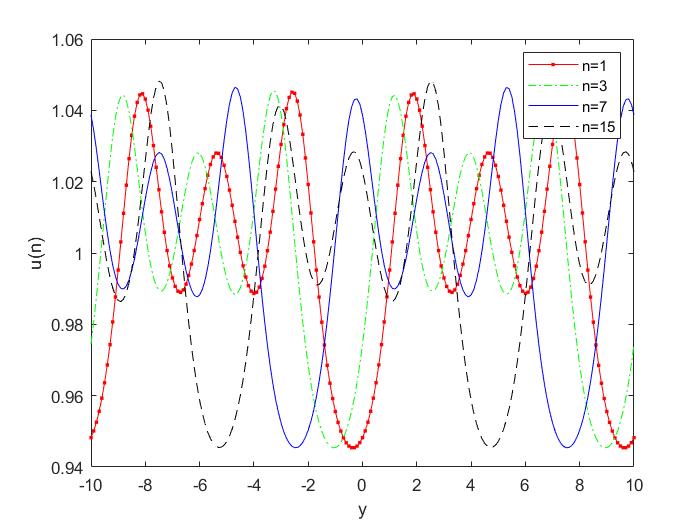} &
\includegraphics[height=0.220\textwidth,angle=0]{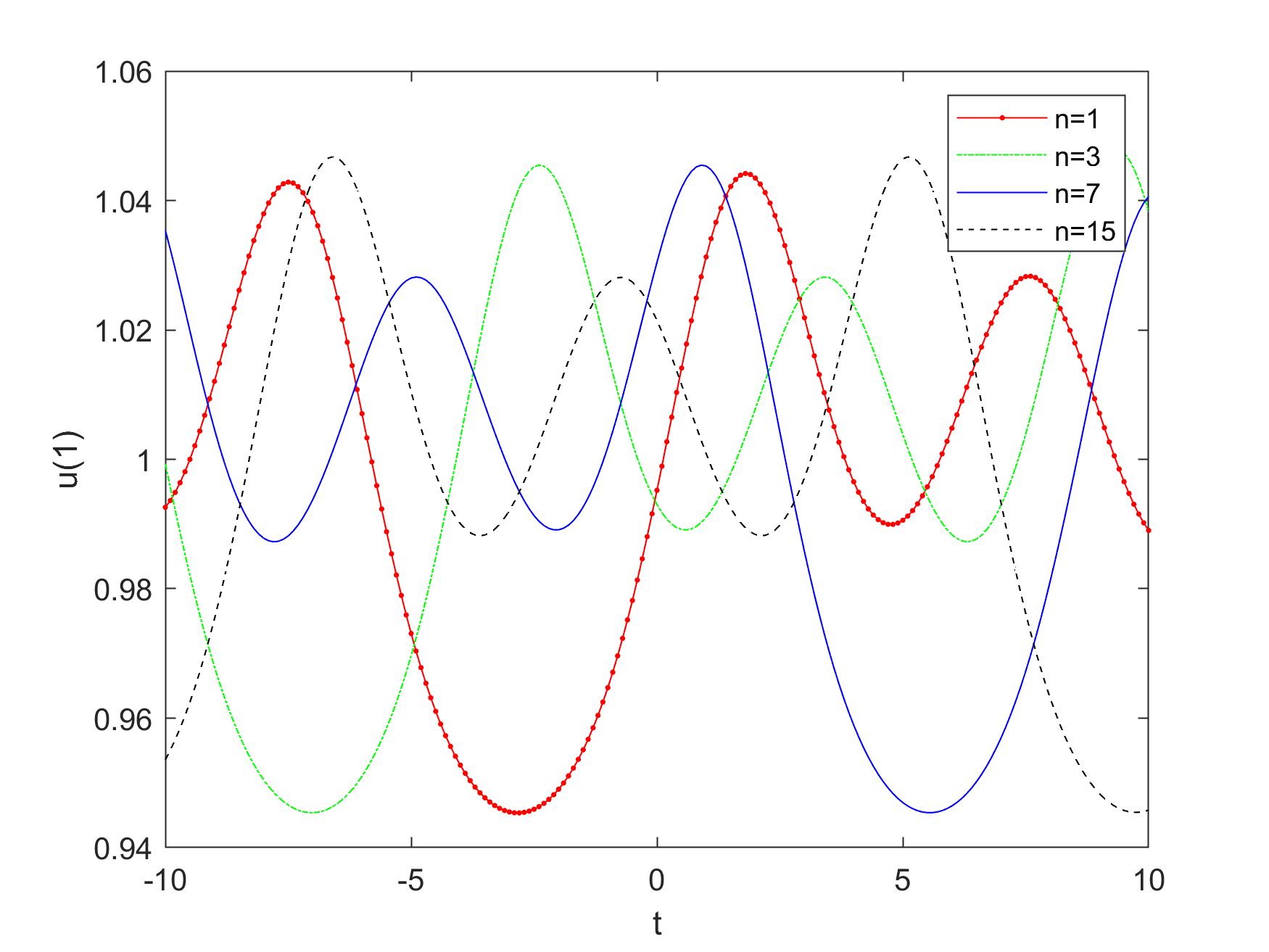}\\
(a) Three-periodic wave $u(n)$ with $t=0$  & \quad  (b) Three-periodic wave $u(n)$ with $n=1$
\end{tabular}
\end{center}
\caption{
Figures of 3-periodic waves to the BS lattice (\ref{gBS-1})-(\ref{gBS-3}) with $z=0$ and parameters in Table \ref{tab2}.}\label{p-wave-bs-2}
\end{figure}

\section{Conclusion and discussion}\label{sec7}
In this paper, we derive a novel BS lattice equations by introducing a new class of trigonometric-type bilinear operators $\sin(\delta D_z)$ and $\cos(\delta D_z)$. We present a unified Gram-type determinant expression of the solution to the variant BS lattice equation by using Hirota's bilinear method.  We investigate the interaction of two solitons and prove the elastic collision. Lump solutions are obtained via two approaches, the B\"acklund transformation and applying the differential operators. As to the second method, the lump solutions are represented in terms of Schur polynomials. We demonstrate three types of breather solutions, including Akhmediev breathers, Kuznetsov-Ma breathers, and generalized breathers. The numerical three-periodic wave solutions to the BS lattice and the variant BS lattice equation are derived by the Gauss-Newton method.

The geometric patterns of the rational solutions represented by Schur polynomials can be analyzed. The distribution of rogue waves or lumps is associated with the roots of the classic polynomials, such as Yablonskii-Vorobev polynomials and Wronskian-Hermite polynomials. In future work, we will investigate the geometric patterns exhibited by the variant BS lattice equations.

	\section*{\bf Acknowledgements}
   The work is supported by National Natural Science Foundation of
	China (Grant no. 12371251, 12175155) and Shanghai Frontier Research Institute for Modern Analysis and the Fundamental Research Funds for the Central Universities. \vskip .5cm

\section*{\bf  Declarations }

\section*{\bf Competing interests }
The authors declare that they have no competing financial interests. Additionally, there are no personal relationships or affiliations that could have influenced the work presented in this manuscript.

\section*{\bf Author contributions }
Two authors equally contributed to writing the paper.

\section*{\bf Data availability } 	No datasets were generated or utilized in this study.

\section*{\bf Conflict of interest }
The authors declare that they have no conflict of interest.

\section*{\bf Ethics approval and consent to participate }
There are no human or animal participants in this article and
informed consent is not applicable.

\section*{\bf Consent for publication}
The authors confirm the work described has not been published before.

\end{document}